\documentclass[10pt]{article}
\usepackage[utf8]{inputenc}
\usepackage[margin=1in]{geometry}
\usepackage{graphicx} % Required for inserting images
\usepackage{tikz}
\usetikzlibrary{positioning}
\usepackage[dvipsnames]{xcolor}
\usetikzlibrary{arrows.meta, positioning, decorations.pathreplacing}
\usepackage{hyperref}

\usepackage{multirow}
\usepackage{amsmath}
\usepackage{amsthm}
\usepackage{amssymb}
\usepackage{mathtools}
\usepackage{amsfonts}
\usepackage{thmtools, thm-restate}
\usepackage{hhline}
\usepackage{booktabs}
\usepackage{xspace}
\usepackage{colortbl}
\usepackage{enumitem}
\usepackage{algorithm}
\usepackage{algpseudocode}
\usepackage{caption}
\usepackage{mathtools}
\usepackage{authblk}
\usepackage[normalem]{ulem}
\algrenewcommand{\algorithmicrequire}{\textbf{Input:}}
\algrenewcommand{\algorithmicensure}{\textbf{Output:}}

\definecolor{lgray}{rgb}{.6, .6, .6}
\definecolor{llgray}{rgb}{.8, .8, .8}

\newtheorem{theorem}{Theorem}
\newtheorem{lemma}[theorem]{Lemma}

\newtheorem{proposition}[theorem]{Proposition}

\newtheorem{example}[theorem]{Example}
\newtheorem{definition}[theorem]{Definition}

\newtheorem{remark}[theorem]{Remark}

\newcommand{\andrei}[1]{{\color{blue}\textbf{AD:} #1}}
\newcommand{\potential}[1]{{\color{purple}#1}}
\newcommand{\dan}[1]{{\color{orange}\textbf{Dan:} #1}}
\newcommand{\ahmet}[1]{{\color{blue}\textbf{Ahmet:} #1}}

\newcommand{\nop}[1]{}

\newcommand{\N}{\mathbb{N}}
\newcommand{\C}{\mathcal{C}}

\newcommand{\bigO}{\mathcal{O}}

\newcommand{\calC}{\mathcal{C}}
\newcommand{\calF}{\mathcal{F}}
\newcommand{\calP}{\mathcal{P}}
\newcommand{\calT}{\mathcal{T}}

\newcommand{\calV}{\mathcal{V}}

\newcommand{\vect}[1]{\boldsymbol{#1}}

\newcommand{\set}[1]{\left\{#1\right\}}

\newcommand{\defeq}{\stackrel{\text{def}}{=}}

\newcommand{\at}{\mathsf{at}}
\newcommand{\leaves}{\mathsf{leaves}}
\newcommand{\subw}{\mathsf{subw}}

\newcommand{\vars}{\mathsf{vars}}
\newcommand{\guard}{\mathsf{guard}}
\newcommand{\Dom}{\mathsf{Dom}}
\newcommand{\out}{\mathsf{OUT}}

\newcommand{\bound}{\mathsf{PBD}}
\newcommand{\DC}{\mathrm{DC}}

\newcommand{\obj}{\mathrm{obj}}
\newcommand{\mw}{\mathrm{mw}}
\newcommand{\dw}{\mathrm{dw}}

\newcounter{magicrownumbers}
\newcommand\rownumber{\footnotesize\stepcounter{magicrownumbers}\arabic{magicrownumbers}}
\newcommand{\linenumber}{\makebox[2ex][r]{\rownumber\TAB}}
\newcommand{\TAB}{\makebox[2.5ex][r]{}}%
\newcommand{\LET}{\textbf{let}\xspace}%
\newcommand{\FOREACH}{\textbf{foreach}\xspace}%
\newcommand{\diamondchord}{
  \tikz[scale=0.7, line width=0.6pt]{
    \draw (0,0) -- (0.15,0.15) -- (0,0.3) 
          -- (-0.15,0.15) -- cycle;
    \draw (-0.15,0.15) -- (0.15,0.15);
  }
}

\newcommand{\paw}{
    \begin{tikzpicture}[scale=0.5, line width=0.6pt]
        \draw (0,0) -- (0.25,0);
        \draw (0.25,0) -- (0.5,0.2) -- (0.5,-0.2) -- cycle;
    \end{tikzpicture}
}

\newcommand{\threepath}{
    \begin{tikzpicture}[scale=0.5, line width=0.5pt]
        \draw (0,0) -- (0.15,0.5) -- (0.3,0) -- (0.45,0.5);
    \end{tikzpicture}
}

\newcommand{\fourpath}{
    \begin{tikzpicture}[scale=0.5, line width=0.5pt]
        \draw (0,0) -- (0.15,0.5) -- (0.3,0) -- (0.45,0.5) -- (0.6,0);
    \end{tikzpicture}
}

\newcommand{\bigpaw}{
    \begin{tikzpicture}[scale=0.5, line width=0.6pt]
        \draw (0,0) -- (0.25,0.2) -- (0.25,-0.2) -- cycle;
        \draw (0.25,-0.2) -- (0.5,-0.2);
        \draw (0.25,0.2) -- (0.5,0.2);
    \end{tikzpicture}
}

\allowdisplaybreaks

\begin{document}

%\title{Incremental View Maintenance of Join Queries using Heavy-Light Partitioning of the Input Relations}
\title{Maintaining Queries under Updates \\Using Heavy-Light Partitioning of the Input Relations}

% Authors
%\author[1]{Mahmoud Abo-Khamis}
%\author[2]{Eden Chmielewski}
%\author[2]{Andrei Draghici}
%\author[3]{Ahmet Kara}
%\author[2]{Dan Olteanu}

\author{Mahmoud Abo-Khamis$^{2}$, Eden Chmielewski$^{1}$, Andrei Draghici$^{1}$, \\ Ahmet Kara$^{3}$, Dan Olteanu}

% Affiliations
\affil[1]{University of Zurich \nop{\texttt{edenviolet.chmielewski@uzh.ch, andrei.draghici@uzh.ch, dan.olteanu@uzh.ch}}, $^2$ Relational AI, $^3$ OTH Regensburg}

%\affil[2]{Relational AI \nop{\texttt{mahmoudabo@gmail.com}}}
%\affil[3]{OTH Regensburg \nop{\texttt{ahmet.kara@oth-regensburg.de}}}

\date{}

\maketitle

\begin{abstract}
    We study the classical incremental view maintenance problem: Given a query and a database, maintain the query output under single-tuple updates (inserts or deletes) to the database such that the tuples in the query output can be enumerated with constant delay after any update.

We introduce a maintenance approach whose update time matches or improves the best update time reported in prior work. Whereas prior approaches are manually tailored to each of a handful of queries, our approach generalizes to arbitrary join queries.
It combines three techniques: delta queries, trees of materialized views, and heavy-light data partitioning. The overall update time incurred by our approach for a given join query is characterized by the maintenance width, a new measure that is parameterized by the heavy-light threshold for data partitioning. We show how to find the threshold that minimizes the maintenance width.
\end{abstract}

\paragraph{Acknowledgements}
This work is partially supported by SNSF 200021-231956.

%     <ccs2012>
%     <concept>
%     <concept_id>10003752.10010070.10010111.10011711</concept_id>
%     <concept_desc>Theory of computation~Database query processing and optimization (theory)</concept_desc>
%     <concept_significance>500</concept_significance>
%     </concept>
%     <concept>
%     <concept_id>10003752.10003809.10010047</concept_id>
%     <concept_desc>Theory of computation~Online algorithms</concept_desc>
%     <concept_significance>500</concept_significance>
%     </concept>
%     </ccs2012>

% \ccsdesc[500]{Theory of computation~Database query processing and optimization (theory)}
% \ccsdesc[500]{Theory of computation~Online algorithms}

% \keywords{incremental view maintenance; data partitioning; delta queries; join queries; materialized views}

\section{Introduction}
\label{sec:intro}

In this paper, we study the classical incremental view maintenance (IVM) problem for join (or full conjunctive) queries:  Given a query and a database, we want to maintain the query output under single-tuple updates (inserts or deletes) to the database such that the tuples in the query output can be enumerated with constant delay after each update. This problem is central to databases and received attention from both database systems and theory communities over the past decades. As highlighted in a recent overview~\cite{IVM:GemsPODS:2024}, there has been renewed interest in charting the complexity of the IVM problem~\cite{IVMeps:ICDT:2019,IVMeps:PODS:2020} and in developing IVM systems in academia~\cite{DBT:VLDBJ:2014,DynYannakakis:SIGMOD:2017,FIVM,CROWN} and industry~\cite{DBSP:VLDBJ:2025,DiffDataFlow:CACM:2016,SnowflakeDT:SIGMOD:2025}.

The classical IVM approach is based on delta queries~\cite{DeltaViews}. More recently, two further techniques made their way into the IVM theory and systems: using a hierarchy, or tree, of materialized views~\cite{DBT:VLDBJ:2014,FIVM}, and heavy-light data partitioning~\cite{IVMeps:ICDT:2019,IVMeps:PODS:2020}. 
View trees (and variants thereof) have been previously used for maintenance by several IVM systems. DBToaster~\cite{DBT:VLDBJ:2014} compiles the given query into a set of view trees, one for each updatable relation. Dynamic Yannakakis~\cite{DynYannakakis:SIGMOD:2017}, F-IVM~\cite{FIVM,dynamicrelation}, and Crown~\cite{CROWN} compile the given query into one view tree, which is then maintained under updates. 
MVIVM translates the given query into a so-called multivariate extension query, which is maintained using view trees~\cite{InsertsVSDeletes}. With the exception of $q$-hierarchical queries, which can be maintained using (a variant of) view trees with constant update time and constant enumeration delay as shown in a seminal work~\cite{QHierarchical} and readily adopted by all subsequent approaches, the update time achieved by these approaches for arbitrary queries can be suboptimal. 

\begin{table}[t]
\centering
\renewcommand{\arraystretch}{1.4} % extra row spacing
\resizebox{0.95\textwidth}{!}{
\begin{tabular}{|l|l|llll|l|}
\hline
\textbf{\small Join Queries} & \textbf{\small Update Time} & \multicolumn{5}{c|}{\textbf{Update Time Upper Bounds}} \\
 & \textbf{\small Lower Bounds} & \textbf{\small F-IVM~\cite{FIVM}} & \textbf{\small IVM$^\epsilon$~\cite{TriangleQuery,IVMeps:PODS:2020}} & \textbf{\small HHH~\cite{4Cycle}} & \textbf{\small MVIVM~\cite{InsertsVSDeletes}} & \textbf{\small This Paper}\\

\hline
hierarchical &
    - &
    $\bigO(1)$ &
    $\bigO(1)$ &
    N/A &
    $\bigO(1)$ 
    &
    $\bigO(1)$ \\

3-path (\threepath) &
    $\Omega(N^{1/2-\gamma})$~\cite{QHierarchical,InsertsVSDeletes} &
    $\bigO(N)$ &
    $\bigO(N^{1/2})^*$ &
    $\bigO(N^{1/2})^*$ &
    $\bigO(N^{1/2})$ &
    $\bigO(N^{1/2})$ \\

4-path (\fourpath) &
    $\Omega(N^{1/2-\gamma})$~\cite{QHierarchical,InsertsVSDeletes} &
    $\bigO(N)$ &
    $\bigO(N^{1/2})^*$ &
    N/A &
    $\bigO(N)$ &
    $\bigO(N^{1/2})$ \\

triangle ($\triangle$) &
    $\Omega(N^{1/2-\gamma})$~\cite{QHierarchical,InsertsVSDeletes} &
    $\bigO(N)$ &
    $\bigO(N^{1/2})$ &
    $\bigO(N^{1/2})^*$ &
    $\bigO(N^{1/2})$ &
    $\bigO(N^{1/2})$\\

LW-$k$ &
    $\Omega(N^{1/2-\gamma})$~\cite{QHierarchical,InsertsVSDeletes} &
    $\bigO(N)$ &
    $\bigO(N^{1/2})$ &
    N/A &
    $\bigO(N^{1/2})$ &
    $\bigO(N^{1/2})$\\

4-cycle ($\square$) &
    $\Omega(N^{2/3-\gamma})$~\cite{InsertsVSDeletes} &
    $\bigO(N)$ &
    N/A &
    $\bigO(N^{2/3})^*$ &
    $\bigO(N)$ &
    $\bigO(N^{2/3})$ \\
    
diamond (\diamondchord) &
    $\Omega(N^{2/3-\gamma})$~\cite{InsertsVSDeletes} & 
    $\bigO(N)$ &
    N/A &
    $\bigO(N^{2/3})^*$ &
    $\bigO(N)$ & 
    $\bigO(N^{2/3})$
    \\

paw (\paw) &
    $\Omega(N^{2/3-\gamma})$~\cite{InsertsVSDeletes} & 
    $\bigO(N)$ &
    N/A &
    $\bigO(N^{2/3})^*$ &
    $\bigO(N^{2/3})$ & 
    $\bigO(N^{2/3})$
    \\

big paw (\bigpaw) &
    $\Omega(N^{2/3-\gamma})$~\cite{InsertsVSDeletes} & 
    $\bigO(N)$ &
    N/A &
    N/A &
    $\bigO(N^{2/3})$ & 
    $\bigO(N^{2/3})$
    \\

bow tie ($\bowtie$) &
    $\Omega(N^{3/4-\gamma})$~\cite{InsertsVSDeletes} &
    $\bigO(N)$ &
    N/A &
    N/A &
    $\bigO(N)$ & 
    $\bigO(N)$
    \\
\nop{
\rowcolor{llgray}{arbitrary} &
    $\Omega(N^{\text{subw} Q) - 1})$~\cite{InsertsVSDeletes} &
    $\bigO(N^\delta)$ &
    N/A &
    N/A &
    $\bigO(N^{\text{fhtw}(\widehat Q) - 1})$ &
    $\bigO(2^{\text{mw}(Q)})$
    \\
}    
\hline
\end{tabular}
}

\caption{
    Comparison between the update times of our maintenance approach versus the best known combinatorial results for different join queries studied in the literature. Update times of all approaches are amortized, except for F-IVM which is worst-case. $N$ is the size of the database at the time of the update. All results assume that after each update, constant delay enumeration of the query output is supported. ($^*$) entries indicate update times shown only for the counting version of the problem; (N/A) entries indicate that the approach is not applicable; LW-$k$ is the Loomis-Whitney query on $k$ variables:
    $Q(X_1, \ldots, X_k) = \prod_{i \in [k]} R_i(\{X_1, \ldots, X_k\}-\{X_i\})$. All lower bounds assume the query has no self-joins and they hold for any $\gamma > 0$.
}
\label{tab:comparisons}
\end{table}

Motivated by the suboptimality of the aforementioned IVM approaches, a distinct line of theoretical work~\cite{IVMeps:ICDT:2019,TriangleQuery,IVMeps:PODS:2020,4Cycle,InsertsVSDeletes} proposed adaptive maintenance approaches that partition the data and use different view trees for different data parts. This can lead to asymptotically lower and even optimal update times. These adaptive approaches were employed for a handful of queries, for which a careful crafting of view trees and the accompanying complexity analysis were made on a case-by-case basis, with no general approach in sight. One  notable exception is the adaptive maintenance of hierarchical queries with arbitrary free variables~\cite{IVMeps:PODS:2020}. IVM$^\epsilon$ is the first approach to achieve optimal (amortized) update time $\bigO(N^{1/2})$ for the triangle count query~\cite{IVMeps:ICDT:2019} and for the full triangle query~\cite{TriangleQuery}, where $N$ is the database size at the time of update. A further optimality result is known for a subclass of hierarchical (but not $q$-hierarchical) queries~\cite{IVMeps:PODS:2020}, where the update time and enumeration delay are $\bigO(N^{1/2})$. Yet tight bounds on the update time are not known beyond these notable cases. This is primarily due to the scarcity of the available lower bounds, which are conditional on the Online Matrix-Vector-Multiplication (OMv) conjecture~\cite{OMV:STOC:2015,QHierarchical}\footnote{In the OMv problem, we are given an $n \times n$ Boolean matrix $\mathbf{M}$ and receive $n$ column vectors of size $n$ denoted by $\mathbf{v}_1,\ldots, \mathbf{v}_n$, one by one; after seeing each  $\mathbf{v}_i$, we output the product $\mathbf{M}\mathbf{v}_i$, before we see the next vector. The OMv conjecture states that for any $\gamma > 0$, there is no combinatorial algorithm that solves OMv in time $\bigO(n^{3 - \gamma})$~\cite{OMV:STOC:2015}. Unless the OMv conjecture fails, there is no dynamic algorithm that can enumerate the output of a non-$q$-hierarchical self-join-free conjunctive query on any database of size $N$ with arbitrary pre-processing time and $\bigO(N^{1/2 - \gamma})$ delay and update time for any $\gamma > 0$~\cite{QHierarchical}.} or on the conjectured optimality of the submodular width for static query evaluation~\cite{InsertsVSDeletes}.\footnote{This conjecture states that for every $\gamma > 0$ and every (full or Boolean) conjunctive query $Q$, there does not exist a combinatorial algorithm that for any database of size $N$ can answer $Q$ in time $\bigO(N^{\subw(Q) - \gamma} + \out)$, where $\subw(Q)$ is the submodular width of $Q$ and \(\out\) is the query output size~\cite{InsertsVSDeletes}. This  conjecture in the static query evaluation setting implies that there is no fully dynamic algorithm that can maintain $Q$ with amortized update time  \(\bigO(N^{\subw(\widehat Q) - 1 - \gamma})\) and constant enumeration delay for any $\gamma>0$, where $\widehat Q$ is the multivariate extension of $Q$~\cite{InsertsVSDeletes}.}
Further approaches fall short of achieving the best known update times. For instance, the MVIVM approach~\cite{InsertsVSDeletes}, which reduces the IVM problem of a query to that of its multivariate extension, was shown to require $\bigO{(N)}$ update time for the 4-cycle query~\cite[Fig. 4]{InsertsVSDeletes}, whereas the best known update time is $\bigO{(N^{2/3})}$. MVIVM also needs $\bigO{(N)}$ update time for the 4-path query, whereas the best update time is $\bigO{(N^{1/2})}$.
Table~\ref{tab:comparisons} overviews the update times (lower and upper bounds) achieved by representative {\em combinatorial}\footnote{Using fast-matrix multiplication, a recently proposed non-combinatorial IVM algorithm~\cite{FMM4Cycle} can achieve $\bigO(N^{2/3-\gamma})$ for the 4-cycle query, where $\gamma = 0.009811$.} approaches for queries studied in the literature.

{\em In this paper, we put forward an adaptive maintenance approach that works for arbitrary join queries and whose update time matches or  improves the best update time reported in prior work}. 

Our approach uses all three aforementioned techniques: delta queries, view trees, and data partitioning. The key challenge addressed by our approach is to algorithmically find the view trees and the heavy-light threshold parameter for data partitioning  that minimize the update time for any given join query. This challenge was not addressed in prior works~\cite{IVMeps:ICDT:2019,TriangleQuery,4Cycle,IVMeps:PODS:2020}, as the choices of view trees and threshold parameter were made manually for each of the considered queries. As shown in Table~\ref{tab:comparisons}, our approach matches the best known upper bounds on the update time for the queries considered in the literature, while also providing a general approach for arbitrary join queries. It also achieves new non-trivial sub-linear update times for other queries, e.g., the optimal $\bigO{(N^{1/2})}$ update time for the $4$-path query.

This paper is organized as follows. Sec.~\ref{sec:prelims} introduces preliminary notions used throughout the paper. Sec.~\ref{sec:overview} overviews our approach, exemplifies it on the $4$-cycle query, and states our main theorem. Sec.~\ref{sec:maintenance-width} introduces the maintenance width, our new measure for the complexity of the maintenance cost. 
Sec.~\ref{sec:comparison} compares the maintenance width and our evaluation strategy with width measures and strategies used in prior approaches.
Sec.~\ref{sec:rebalancing} discusses how to amortize the cost of occasional expensive updates over a sequence of updates. Sec.~\ref{sec:enumeration} reviews the constant-delay enumeration of tuples from the query output represented by the view trees. Sec.~\ref{sec:conclusion} concludes with thoughts on future work. Some proof details are deferred to the appendix. Further details on the amortization and enumeration procedures, as well as the application of our approach to the queries listed in Table~\ref{tab:comparisons}, are provided in the appendix.

\section{Preliminaries}
\label{sec:prelims}

\subsection{Data and Queries}

A {\em schema} $\vect X$ is a tuple of attributes or variables $(X_1,\ldots,X_n)$, which we also conveniently see as a set to allow set operations on tuples. Each variable $X_i$ draws its values from a set $\Dom(X_i)$. 
A tuple $\vect x$ of values over the schema $\vect X$ is an element of the set $\Dom(\vect X)=\Dom(X_1)\times\ldots\times\Dom(X_n)$. Following prior work on IVM, e.g.,~\cite{FIVM}, a relation $R$ over schema $\vect X$ is a function that maps tuples of values over $\vect X$ to multiplicities, which are integers. When applying set operations to $R$, we treat it as the set of tuples $\vect{x}$ with $R(\vect{x}) >0$. For instance, the size of $R$, denoted by $|R|$, is the number of tuples $\vect x$ for which $R(\vect x) > 0$.
We specify queries using a syntax similar to that of functional aggregate queries over the  $(\mathbb{Z},+,\cdot,0,1)$ ring~\cite{FAQ}:
\begin{equation}\label{eq:query-explicit}
Q(\vect{F}) \;=\; \sum_{\vect B} R_1({\vect{X_1}})\cdot\ldots\cdot R_k(\vect{X_k}),
\end{equation}
where $\sum$ and $(\cdot)$ are the summation and respectively multiplication operation from the ring, $R_1,\dots,R_k$ are {\em relation symbols}, each $\vect{X_i}$ is a schema, and each $R_i(\vect{X_i})$ is an \emph{atom} of $Q$. We assume distinct relation symbols in a query. If several relation symbols correspond to the same physical database relation, which happens in case $Q$ has self-joins, then we assume without loss of generality (i.e., without changes in the data complexities stated in the paper) that each such atom gets its own copy of the database relation. The set of variables of $Q$ is  $\vars(Q) \defeq \bigcup_{i\in[k]}\vect{X_i}$. The {\em free} variables of $Q$ are $\vect F\subseteq \vars(Q)$, while $\vect B = \vars(Q)\setminus\vect F$ are the {\em bound variables}. If $\vars(Q) = \vect F$, then $Q$ is a full (or join) query. By $\at(Q)$ and $\at(Y)$ we denote the set of all atoms of $Q$ and the set of all atoms $R_i(\vect{X_i})$ with $Y\in\vect{X_i}$, respectively. For compactness, we write a set of variables as the concatenation of their names, e.g., $\{X,Y,Z\}$ becomes $XYZ$ while $\{X\}$ becomes $X$. 
Each variable $A \in \vect X_i \cap \vect X_j$ expresses an equi-join between $R_i$ and $R_j$, for $i\neq j$, and is called a {\em join variable}. Let ${\mathcal J}_Q$ be the tuple of all join variables in $Q$, ordered using a fixed total order on $\vars(Q)$.

For any variable $Y\in\vars(Q)$, $\Dom(Y)$ is defined as: $\Dom(Y) \defeq \bigcup_{R_i(\vect X_i)\in\at(Y)} \pi_{Y} R_i$, where $\pi$ is the standard projection operator in relational algebra. The marginalization of variables $\vect Y\subseteq \vect X$ from a relation $R$ over variables $\vect X$, denoted by $S(\vect Z) = \sum_{\vect Y} R(\vect X)$ for $\vect Z = \vect X \setminus \vect Y$, is defined by: 
$\forall \vect z\in\Dom(\vect Z): 
S(\vect z) \defeq \sum \{R(\vect x) \mid \vect x\in\Dom(\vect X)\wedge \vect z = \vect x.\vect Z \}$, where $\vect x.\vect Z$ is the restriction of the tuple $\vect x$ to the values of the variables in schema $\vect Z$. The union of two relations $R$ and $S$ over the same variable set $\vect X$, denoted by $T = R\cup S$, is defined as: $\forall \vect x\in\Dom(\vect X): T(\vect x) = R(\vect x) + S(\vect x)$. 
The query $Q$ in Eq.~\eqref{eq:query-explicit} defines the relation: $\forall \vect f\in\Dom(\vect F): Q(\vect f) = \sum_{\vect x\in\Dom(\vars(Q)): \vect f = \vect x.\vect F} R_1(\vect x.\vect X_1)\cdot\ldots\cdot R_k(\vect x.\vect X_k)$.

%%%%%%%%%%%%%%%%%%%
\subsection{Data Updates}
Following prior work~\cite{FIVM}, we model database \emph{updates} as a sequence of single-tuple inserts and deletes. 
The insert (delete) of a tuple $\vect x$ into (from) a relation $R$ is expressed as a delta relation $\delta R$ that maps $\vect x$ to 1 (and $-1$, respectively). The updated relation is the union of the old relation and the delta relation: $R := R \cup \delta R$. A delete, whose effect is a tuple with negative multiplicity in the updated relation, is rejected.
Updates are defined for joins of relations using the classical delta rule: $\delta (V_1(\vect Z_1)\cdot V_2(\vect Z_2)) = (\delta V_1 (\vect Z_1)\cdot V_2(\vect Z_2))\cup (V_1 (\vect Z_1)\cdot \delta V_2(\vect Z_2)) \cup (\delta V_1 (\vect Z_1)\cdot \delta V_2(\vect Z_2))$. This generalizes to a join of arbitrary relations by taking $V_2$ to be the  join of relations and applying recursively the delta rule to $\delta V_2$. If only $V_1$ is  changed, then $\delta V_2 = \emptyset$ and $\delta (V_1(\vect Z_1)\cdot V_2(\vect Z_2)) = \delta V_1 (\vect Z_1)\cdot V_2(\vect Z_2)$. Updates commute with variable marginalization: $\delta (\sum_{\vect Y} V) = \sum_{\vect Y} \delta V$. 
If a database relation has several copies due to our assumption on distinct relation symbols, then each of these copies needs to be updated and triggers a delta query. 

%%%%%%%%%%%%%%%%%%%
\subsection{Delta View Trees} Our maintenance approach is supported by a tree of materialized views.

\begin{definition}[(Delta) View Tree]
\label{def:viewtree}
A \emph{view tree} $T$ for a query $Q$ is a rooted tree with the properties:
\begin{itemize}
    \item There is a one-to-one mapping between the leaves of $T$ and the atoms of $Q$.
    \item Each inner node is a view over some variables of $Q$.
    \item If a node $V'(\vect{Y})$ has a single child node $V(\vect{X})$, then it is a \emph{projection view} defined by marginalizing variables of $V(\vect{X})$, i.e., $V'(\vect{Y}) = \sum_{\vect{X} \setminus \vect{Y}} V(\vect{X})$. 
    Furthermore, every atom of $Q$ with a variable from $\vect{X} \setminus \vect{Y}$  occurs in the subtree rooted at $V(\vect{X})$.
    \item If a node $V(\vect{X})$ has several children $V_1(\vect{X_1}), \ldots , V_1(\vect{X_n})$, for $n \geq 2$, then it is a \emph{join view} defined by the natural join of the  child views, i.e., $V(\vect{X}) = V_1({\vect{X_1}})\cdot\ldots\cdot V_k(\vect{X_n})$.
\end{itemize}
Under an update to a relation $R$, the view tree $T$ becomes a \emph{delta view tree}, denoted by $\delta T_R$, where $R$ is replaced by $\delta R$, and each view $V$ along the path from $\delta R$ to the root view is replaced by $\delta V$.
\end{definition}

For a view tree $T$, an update to a relation $R$, and a view $V(\vect X) \in T$, let $\leaves(\delta V(\vect X), \delta T_R)$ denote the set of all leaves of the subtree of $\delta T_R$ rooted at $\delta V(\vect X)$. When $\delta T_R$ is clear from context, it is omitted. We use $\calT(Q)$ to denote the set of all view trees of $Q$. Fig.~\ref{fig:4-cycle-view-trees} depicts view trees for the 4-cycle query. 

A procedure enumerates the query output with {\em constant delay} if the time is constant between: (i) the start of the enumeration process and the output of the first tuple; (ii) outputting any two consecutive tuples; and (iii) outputting the last tuple and the end of the enumeration process~\cite{ConstDelayEnum}. Any view tree for a join query allows for the constant-delay enumeration of the query output~\cite{FIVM}.

%%%%%%%%%%%%%%%%%%%%%%%%%
\subsection{Heavy-Light Data Partitioning}

For a join variable $Y$, we partition the $Y$-values in the database in {\em light} and {\em heavy} according to a {\em threshold parameter} $\epsilon\in[0,1]$ and database size $N$: 
\begin{align*}
    Light(Y) := \{y \in \Dom(Y) \mid \sum_{R_i(\vect{X_i})\in\at(Y)} |\sigma_{Y=y} R_i| \leq N^\epsilon\} \hspace*{2em} Heavy(Y):=\Dom(Y)\setminus Light(Y).
\end{align*}
That is, a light $Y$-value $y$ occurs in at most $N^{\epsilon}$ tuples across all relations, while there are at most $N^{1-\epsilon}$  heavy $Y$-values. For an atom $R(Y,\vect{Z})$, the relation $R$ is the disjoint union of its fragment where $Y$ is light and its fragment where $Y$ is heavy.

\begin{definition}[Degree Configuration]\label{def:degree-configuration}
    Given a query $Q$ with the tuple of join variables $(Y_1,\ldots, Y_n)$, a {\em degree configuration} is a tuple $(d_1,\ldots,d_n)$, where $d_i=L$ in case $Y_i$ is light and $d_i=H$ in case $Y_i$ is heavy ($i\in[n]$).
\end{definition}

There are $2^n$ degree configurations for a query $Q$ with $n$ join variables. We denote by ${\mathcal D}(Q)$ the set of all such degree configurations.

\begin{definition}[Relation Restriction]\label{def:relation-part}
    Given a query $Q$ with the tuple of join variables $(Y_1, \ldots, Y_n)$, degree configuration $\vect{d} = (d_1, ..., d_n)$, and atom $R(\vect{X})$ in $Q$, the \emph{$\vect{d}$-restriction} of relation $R$ is:
    \[
    \{\vect r \mid \vect r\in R\wedge \forall i\in [n]: Y_i\in \vect X \rightarrow (d_i=L\wedge \vect r. Y_i\in Light(Y_i) \vee d_i=H\wedge \vect r.Y_i\in Heavy(Y_i))\}.
    \]
\end{definition}

\subsection{Degree Constraints} 

Using the degree configurations and the updates, we derive constraints on the degrees for the database values.

\begin{definition}[Degree Constraint]\cite[Def.~1]{hung_algo}
\label{def:degree_constraint}   
    A \emph{degree constraint} on a database of size $N$ is a tuple \((\vect Z|\vect Y, N^{p_{\vect Z|\vect Y}})\), where $\vect Y\subsetneq \vect Z$ and \(p_{\vect Z|\vect Y}\in\mathbb{Q}_{\geq 0}\). Let \(Q\) be a query and $R_i(\vect X_i)$ an atom of \(Q\). Then $R_i(\vect X_i)$ \emph{guards} the degree constraint  \((\vect Z|\vect Y, N^{p_{\vect Z|\vect Y}})\) if \(\vect Z\subseteq \vect X_i\) and $\max_{\vect t} |\pi_{\vect Z}(\sigma_{\vect Y=\vect t.\vect Y} R_i)|\leq N^{p_{\vect Z|\vect Y}}$. \nop{We denote by \(\guard_Q(c)\) the set of atoms of $Q$ that guard the degree constraint \(c\).}
\end{definition}

For a set of degree constraints $\C$, we denote the set of all variables appearing in \(\C\) by \(\vars(\C) \defeq \bigcup_{(\vect Z\mid \vect Y,N^{p_{\vect Z|\vect Y}})\in\C} \vect Z\). A {\em projection} of a set \(\C\) of degree constraints  onto a set \(\vect V\) of variables, denoted by \(\C[\vect V]\), is defined as the set $\{(\vect Z\cap\vect V|\vect Y, N^{p_{\vect Z|\vect Y}}) \mid (\vect Z|\vect Y, N^{p_{\vect Z|\vect Y}})\in\C\wedge \vect Y\subsetneq \vect Z\cap\vect V\}$. Projections may not only shrink the set of variables covered by a degree constraint ($\vect Z\cap\vect V$), but also remove a constraint from $\C$ altogether. This happens when $\vect Y\subsetneq \vect Z\cap\vect V$ is violated. For instance, the projection $(\vect Y|\vect Y, N^{p_{\vect Z|\vect Y}})$ of the constraint $(\vect Z|\vect Y, N^{p_{\vect Z|\vect Y}})$ onto $\vect Y$ is uninformative and therefore discarded.

For a database of size $N$, delta view $\delta V(\vect X)$ in a delta view tree $T_{R_j}$ for an update $\delta R_j$, and degree configuration $\vect d$, we  define the set of degree constraints that are guarded by the database relations at the leaves $\mathcal{L}=\leaves(\delta V(\vect X))$ of $\delta V(\vect X)$ in $T_{R_j}$ as: 
\begin{align*}
    \DC(\mathcal{L}, \vect d) \defeq &\{
    (\vect X_i|\emptyset, N) \mid R_i(\vect X_i)\in \mathcal{L} \wedge i \neq j\} \cup &\text{ (size constr.)}\\
    &\{(\vect X_i|Y, N^{\epsilon}) \mid R_i(\vect X_i)\in \mathcal{L} \wedge Y\in \vect X_i \wedge Y \text{ is light in } \vect{d} \wedge i \neq j\} \cup &\text{ (light constr.)}\\
    &\{(Y|\emptyset, N^{1-\epsilon}) \mid  R_i(\vect X_i)\in \mathcal{L} \wedge Y\in \vect X_i \wedge Y \text{ is heavy in } \vect{d} \wedge i \neq j\}\cup &\text{ (heavy constr.)}\\
    &\{(A|\emptyset, 1) \mid  A \in \vect{X}_j\} &\text{ (update constr.)}
\end{align*}
Note that this set does not contain constraints that are guarded by the relation $R_j$ itself, since this relation is not in $\delta T_{R_j}$. 
If a set $\C$ of degree constraints is guarded by a relation, then this also holds for its projection $\C[\vect X]$ onto any set of variables $\vect X\subseteq \vars(\C)$.

\begin{definition}[Acyclic Sets of Degree Constraints]
    For any set $\C$ of degree constraints, associate a directed graph \(G_\C\) with a node for every variable in $\vars(\C)$ and with a directed edge \((y,z)\in \vect Y\times (\vect Z-\vect Y)\) for every degree constraint \((\vect Z\mid \vect Y,N^{p_{\vect Z|\vect Y}})\in\) \(\C\). If \(G_\C\) is acyclic, then \(\C\) is called \emph{acyclic}. We denote by ${\mathcal A}(\C)$ the set of all maximal acyclic subsets of \(\C\).
\end{definition}

All maximal acyclic subsets of a set $\C$ of degree constraints contain all size, heavy, and update constraints from $\C$, as these constraints have the form $(\vect Z|\emptyset, N^{p_{\vect Z|\emptyset}})$ and do not create edges in the constraint graph. It is only the light constraints that create edges and therefore cycles in this graph.

%%%%%%%%%%%%%%
\subsection{RAM Model of Computation}

We assume that each relation $R_i$ is implemented by a data structure of size $\bigO(|R_i|)$ that can: (i) look up, insert, and
delete tuples in $R_i$ in amortized constant time, and (ii) enumerate all tuples in $R_i$ with constant delay. 
For a set $\vect S \subsetneq \vect{X_i}$, we use an index data structure that, for any tuple $\vect{x_S}$ over the variables in $\vect S$, (iii) can enumerate all tuples in $\sigma_{\vect S = \vect{x_S}} R_i$ with constant delay, and (iv) insert and delete index entries in amortized constant time. We also need indices to (v) enumerate with constant delay the distinct tuples in each relation constructed by marginalizing any subset of variables of $R_i$. We report the time complexity as a function of the database size $N$ only, where the query is considered fixed and of constant size (data complexity). Therefore, the enumeration delay is constant when it does not depend on the database size.

\section{Overview of Our Adaptive IVM approach}
\label{sec:overview}

In this section, we overview our maintenance approach.
Given a join query $Q$, a database of size $N$, and a single-tuple update (tuple insert or delete), our approach updates  the query output and allows for constant-delay enumeration of the tuples in the query output.
Our approach is adaptive: For different heavy-light partitioning of the database on the columns corresponding to the join variables in $Q$ (as described in Sec.~\ref{sec:prelims}), it may use a different maintenance strategy.

\begin{example}\label{ex:4cycle-degree-constraints}
\rm
    We use as running example throughout this section the 4-cycle query: 
    \[
        Q(A,B,C,D) = R(A,B)\cdot S(B,C) \cdot T(C,D) \cdot U(D,A).
    \]
    The query has $4$ join variables, $2^4=16$ degree configurations. For instance, the degree configuration $\vect d = (L,L,H,H)$ for the tuple of join variables $(A,B,C,D)$ corresponds to the relation restrictions where $A$ and $B$ are light while $C$ and $D$ are heavy. Under an update $\delta R$, the delta query is: $\delta Q(A,B,C,D) = \delta R(A,B)\cdot S(B,C) \cdot T(C,D) \cdot U(D,A)$.
    For the degree configuration $\vect d$, threshold parameter $\epsilon$, and database size $N$, we have the following degree constraints: 
    \begin{align*}
        \DC(\at(\delta Q), \vect d) = \{&
          (BC | \emptyset, N), (CD | \emptyset, N), (AD | \emptyset, N),
          (AD | A, N^{\epsilon}), (BC | B, N^{\epsilon}),\\
          & (C | \emptyset, N^{1-\epsilon}), (D | \emptyset, N^{1-\epsilon}), (A | \emptyset, 1), (B | \emptyset, 1)\}.
    \end{align*}
    The first three constraints are size constraints (relations $S$, $T$, and $U$ have sizes at most $N$). The next two constraints are light constraints, e.g., there are at most $N^\epsilon$ $D$-values for a given $A$-value. The first two constraints in the second line are heavy constraints, e.g., there are at most $N^{1-\epsilon}$ $C$-values, while the last two constraints express that each of $A$ and $B$ is set to one value (due to the update $\delta R$). 
    
    \nop{Fig.~\ref{fig:degree-graphs} shows the graph of these degree constraints and its two maximal acyclic subgraphs:
    \begin{equation*}
        \C_1 = \DC(\at(Q), (L,L,H,H), \delta R)\setminus\set{(AB|B,N^\epsilon)},
        \C_2 = \DC(\at(Q),(L,L,H,H),\delta R)\setminus\set{(AB|A,N^\epsilon)}
    \end{equation*}
    }
\end{example}

We maintain the output of a join query  using trees of materialized views (Def.~\ref{def:viewtree}). 

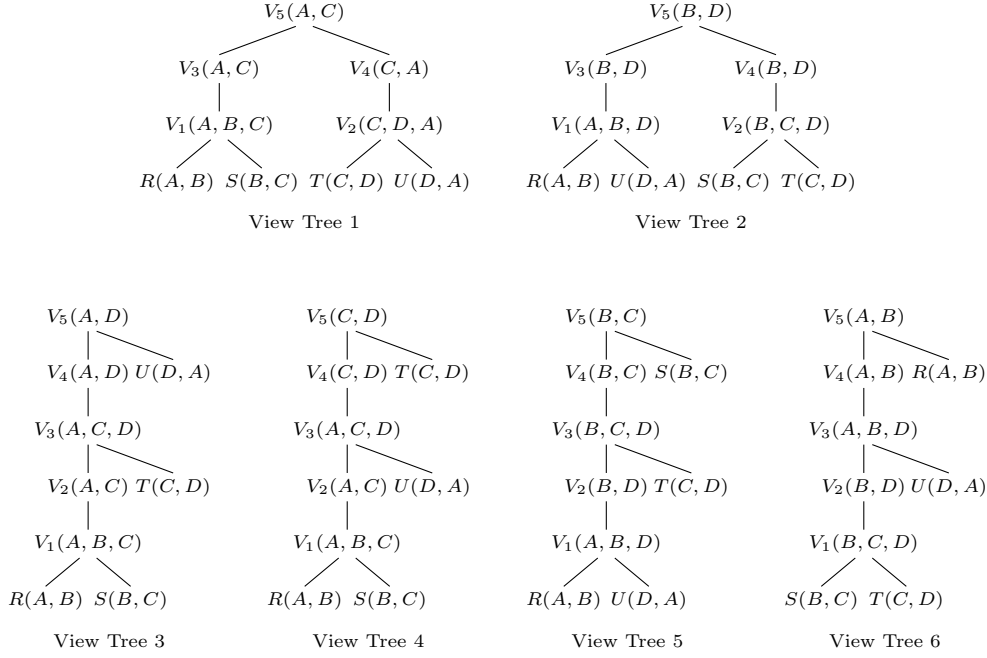
\begin{figure}
    \centering
    \setlength{\extrarowheight}{4pt}
    \setlength{\tabcolsep}{10pt} % default is 6pt
    \begin{tabular}{c}
    %------------------ Top row: Figures 1–2 ------------------%
    \begin{tabular}{cc}
        % Figure 1
        \begin{tikzpicture}[
            scale=0.75,
            every node/.style={
            circle,
            draw=none,
            minimum size=2.5em,
            inner sep=0pt,
            font=\scriptsize
            },
            every edge/.style={>=Stealth, line width=0.7, draw=black},
            ]
        
            \node at (0,0) {$R(A, B)$};
            \node at (1.5,0) {$S(B, C)$};
            \node at (3,0) {$T(C, D)$}; 
            \node at (4.5,0) {$U(D, A)$};
            
            \node at (0.75,1) {$V_1(A,B,C)$}; 
            \node at (3.75,1) {$V_2(C,D,A)$};
               
            \node at (0.75,2) {$V_3(A,C)$};
            \node at (3.75,2) {$V_4(C,A)$};
    
            \node at (2.25,3) {$V_5(A,C)$};
    
            \draw (0,0.25) -- (0.6,0.75);
            \draw (1.5,0.25) -- (0.9,0.75); 
            \draw (3,0.25) -- (3.6,0.75);
            \draw (4.5,0.25) -- (3.9,0.75);
    
            \draw (0.75,1.25) -- (0.75,1.75);
            \draw (3.75,1.25) -- (3.75,1.75);
    
            \draw (0.75,2.25) -- (2.1,2.75);
            \draw (3.75,2.25) -- (2.4,2.75);
    
            \node at (2.25,-0.7) {View Tree 1};
        \end{tikzpicture}
        &
        % Figure 2
        \begin{tikzpicture}[
            scale=0.75,
            every node/.style={
            circle,
            draw=none,
            minimum size=2.5em,
            inner sep=0pt,
            font=\scriptsize
            },
            every edge/.style={>=Stealth, line width=0.7, draw=black},
            ]
        
            \node at (0,0) {$R(A, B)$};
            \node at (1.5,0) {$U(D, A)$};
            \node at (3,0) {$S(B, C)$}; 
            \node at (4.5,0) {$T(C, D)$};
            
            \node at (0.75,1) {$V_1(A,B,D)$}; 
            \node at (3.75,1) {$V_2(B,C,D)$};
               
            \node at (0.75,2) {$V_3(B,D)$};
            \node at (3.75,2) {$V_4(B,D)$};
    
            \node at (2.25,3) {$V_5(B,D)$};
    
            \draw (0,0.25) -- (0.6,0.75);
            \draw (1.5,0.25) -- (0.9,0.75); 
            \draw (3,0.25) -- (3.6,0.75);
            \draw (4.5,0.25) -- (3.9,0.75);
    
            \draw (0.75,1.25) -- (0.75,1.75);
            \draw (3.75,1.25) -- (3.75,1.75);
    
            \draw (0.75,2.25) -- (2.1,2.75);
            \draw (3.75,2.25) -- (2.4,2.75);
    
            \node at (2.25,-0.7) {View Tree 2};
        \end{tikzpicture}
    \end{tabular}
    \\[-0.5em]
    %------------------ Bottom row: Figures 3–6 ------------------%
    \begin{tabular}{cccc}
        \begin{tikzpicture}[
            scale=0.75,
            every node/.style={
            circle,
            draw=none,
            minimum size=2.5em,
            inner sep=0pt,
            font=\scriptsize
            },
            every edge/.style={>=Stealth, line width=0.7, draw=black},
            ]
        
            \node at (0,0) {$R(A, B)$};
            \node at (1.5,0) {$S(B, C)$};
            
            \node at (0.75,1) {$V_1(A,B,C)$}; 
               
            \node at (0.75,2) {$V_2(A,C)$};
            \node at (2.25,2) {$T(C,D)$};
    
            \node at (0.75,3) {$V_3(A,C,D)$};

            \node at (0.75,4) {$V_4(A,D)$};
            \node at (2.25,4) {$U(D,A)$};

            \node at (0.75,5) {$V_5(A,D)$};
    
            \draw (0,0.25) -- (0.6,0.75);
            \draw (1.5,0.25) -- (0.9,0.75);
    
            \draw (0.75,1.25) -- (0.75,1.75);

            \draw (0.75,2.25) -- (0.75,2.75);
            \draw (2.25,2.25) -- (0.9,2.75);

            \draw (0.75,3.25) -- (0.75,3.75);

            \draw (0.75,4.25) -- (0.75,4.75);
            \draw (2.25,4.25) -- (0.9,4.75);
    
            \node at (1.125,-0.7) {View Tree 3};
        \end{tikzpicture}
        &
        \begin{tikzpicture}[
            scale=0.75,
            every node/.style={
            circle,
            draw=none,
            minimum size=2.5em,
            inner sep=0pt,
            font=\scriptsize
            },
            every edge/.style={>=Stealth, line width=0.7, draw=black},
            ]
        
            \node at (0,0) {$R(A, B)$};
            \node at (1.5,0) {$S(B, C)$};
            
            \node at (0.75,1) {$V_1(A,B,C)$}; 
               
            \node at (0.75,2) {$V_2(A,C)$};
            \node at (2.25,2) {$U(D,A)$};
    
            \node at (0.75,3) {$V_3(A,C,D)$};

            \node at (0.75,4) {$V_4(C,D)$};
            \node at (2.25,4) {$T(C,D)$};

            \node at (0.75,5) {$V_5(C,D)$};
    
            \draw (0,0.25) -- (0.6,0.75);
            \draw (1.5,0.25) -- (0.9,0.75);
    
            \draw (0.75,1.25) -- (0.75,1.75);

            \draw (0.75,2.25) -- (0.75,2.75);
            \draw (2.25,2.25) -- (0.9,2.75);

            \draw (0.75,3.25) -- (0.75,3.75);

            \draw (0.75,4.25) -- (0.75,4.75);
            \draw (2.25,4.25) -- (0.9,4.75);
    
            \node at (1.125,-0.7) {View Tree 4};
        \end{tikzpicture}
        &
        \begin{tikzpicture}[
            scale=0.75,
            every node/.style={
            circle,
            draw=none,
            minimum size=2.5em,
            inner sep=0pt,
            font=\scriptsize
            },
            every edge/.style={>=Stealth, line width=0.7, draw=black},
            ]
        
            \node at (0,0) {$R(A,B)$};
            \node at (1.5,0) {$U(D,A)$};
            
            \node at (0.75,1) {$V_1(A,B,D)$}; 
               
            \node at (0.75,2) {$V_2(B,D)$};
            \node at (2.25,2) {$T(C,D)$};
    
            \node at (0.75,3) {$V_3(B,C,D)$};

            \node at (0.75,4) {$V_4(B,C)$};
            \node at (2.25,4) {$S(B,C)$};

            \node at (0.75,5) {$V_5(B,C)$};
    
            \draw (0,0.25) -- (0.6,0.75);
            \draw (1.5,0.25) -- (0.9,0.75);
    
            \draw (0.75,1.25) -- (0.75,1.75);

            \draw (0.75,2.25) -- (0.75,2.75);
            \draw (2.25,2.25) -- (0.9,2.75);

            \draw (0.75,3.25) -- (0.75,3.75);

            \draw (0.75,4.25) -- (0.75,4.75);
            \draw (2.25,4.25) -- (0.9,4.75);
    
            \node at (1.125,-0.7) {View Tree 5};
        \end{tikzpicture}
        &
        \begin{tikzpicture}[
            scale=0.75,
            every node/.style={
            circle,
            draw=none,
            minimum size=2.5em,
            inner sep=0pt,
            font=\scriptsize
            },
            every edge/.style={>=Stealth, line width=0.7, draw=black},
            ]
        
            \node at (0,0) {$S(B,C)$};
            \node at (1.5,0) {$T(C,D)$};
            
            \node at (0.75,1) {$V_1(B,C,D)$}; 
               
            \node at (0.75,2) {$V_2(B,D)$};
            \node at (2.25,2) {$U(D,A)$};
    
            \node at (0.75,3) {$V_3(A,B,D)$};

            \node at (0.75,4) {$V_4(A,B)$};
            \node at (2.25,4) {$R(A,B)$};

            \node at (0.75,5) {$V_5(A,B)$};
    
            \draw (0,0.25) -- (0.6,0.75);
            \draw (1.5,0.25) -- (0.9,0.75);
    
            \draw (0.75,1.25) -- (0.75,1.75);

            \draw (0.75,2.25) -- (0.75,2.75);
            \draw (2.25,2.25) -- (0.9,2.75);

            \draw (0.75,3.25) -- (0.75,3.75);

            \draw (0.75,4.25) -- (0.75,4.75);
            \draw (2.25,4.25) -- (0.9,4.75);
    
            \node at (1.125,-0.7) {View Tree 6};
        \end{tikzpicture}
    \end{tabular}
    \end{tabular}
    \vspace*{-2em}
    \caption{The six view trees used to maintain the 4-cycle query.}
    \label{fig:4-cycle-view-trees}
    % \Description{Six view trees for the 4-cycle query.}
\end{figure}

\begin{example}
\rm
    Fig.~\ref{fig:4-cycle-view-trees} depicts six possible view trees for the 4-cycle query. The leaves are atoms that correspond to the four relations, while the intermediate nodes are join or projection views. 
\end{example}

Each view tree admits a simple maintenance mechanism~\cite{FIVM,dynamicrelation}. Given an update to a relation, all views along the path from the leaf corresponding to the updated relation to the root may be affected by the update and we compute deltas for them. We refer to the modified tree, where the updated relation and the views along the path to the root are replaced by their deltas, as the delta view tree.
The enumeration of the tuples in the query output proceeds top-down in the view tree and needs constant delay per tuple.

\begin{example}\label{ex:fivm-maintenance}
\rm
    An update $\delta R: \{(a,b)\mapsto m\}$ to relation $R$, where $m=+1$ for an insert and $m=-1$ for a delete, in the first view tree in Fig.~\ref{fig:4-cycle-view-trees} triggers the computation of updates for the views along the path from the leaf $R$ to the root of the view tree:
    \begin{align*}
        \delta V_1 (A,B,C) & = \delta R(A,B) \cdot S(B,C) \hspace*{2em} &V_1 := V_1 \cup \delta V_1\\
        \delta V_3(A,C) &= \sum_B \delta V_1 (A,B,C) &V_3 := V_3 \cup \delta V_3\\
        \delta V_5(A,C) &= \delta V_3(A,C) \cdot V_4(C,A) &V_5 := V_5 \cup \delta V_5
    \end{align*}    
    This bottom-up propagation of the updates ensures that the views are calibrated top-down. For instance, all pairs $(a,c)$ in $V_5$ are also in all views and relations below $V_5$; furthermore, the $B$-values ($D$-values) paired with $(a,c)$ in $V_1$ ($V_2$) are also in $R$ and $S$ (respectively $T$ and $U$). Consequently, the tuples in the query output can be enumerated with constant delay. We enumerate with constant delay: the pairs $(a,c)$ in $V_5$, the $B$-values paired with $(a,c)$ in $V_1$ and the $D$-values paired with $(a,c)$ in $V_2$.    
\end{example}

View trees (and variants thereof) have been previously used by several IVM systems~\cite{QHierarchical,DBT:VLDBJ:2014,DynYannakakis:SIGMOD:2017,FIVM,dynamicrelation,CROWN}. For instance, F-IVM~\cite{FIVM} compiles a given query into one view tree, which is then maintained under updates as shown in Ex.~\ref{ex:fivm-maintenance}. As discussed in the introduction, the update time achieved by these approaches for arbitrary queries can be suboptimal.

Our approach can use several view trees. It considers a heavy-light partitioning of the data and can use different view trees for different  degree configurations. This adaptivity can lead to lower update times than when using a single view tree. The challenge brought by adaptivity is to algorithmically find (1) an asymptotically best view tree for each degree configuration and (2) the heavy-light threshold parameter $\epsilon$  that minimizes the update time for any given query.  We address this challenge as follows. 

In our approach, the update time is a function of the threshold parameter $\epsilon$. In particular, given a delta view tree and a degree configuration, each delta view is computed under the degree constraints parameterized by $\epsilon$. By parameterization, we mean that for a degree constraint $(\vect Z|\vect Y, N^{p_{\vect Z | \vect Y}})$, the exponent $p_{\vect Z | \vect Y}$ is a linear function of $\epsilon$, which evaluates to a positive rational number for a given value for $\epsilon$. The compute time for a delta view is given by $\bigO(N^s)$, where $N$ is the database size and $s$ is the optimal solution of a linear program that computes the polymatroid bound under degree constraints~\cite{CSMA,hung_algo}. This polymatroid bound is a generalization of the well-known AGM bound~\cite{AGM:2008} from size constraints to more general degree constraints that also include the light constraints.
The linear program of such bounds assigns a positive weight to each constraint such that for each query variable, the sum of the weights of the constraints that cover the variable is at least one. The objective is to minimize the sum of all weights, where each weight is multiplied by the exponent $p_{\vect Z | \vect Y}$ in the corresponding constraint.
A key observation is that the inequalities of the linear program do not depend on $\epsilon$. Therefore, the vertices of the polyhedron given by the feasible region of the linear program are independent of $\epsilon$. This also means that the optimal solution, which is given by one of these vertices, can be expressed as the minimum over all vertices of the objective instantiated for each vertex.

\begin{example}
\rm
\label{ex:update-times}
    Let us consider View Tree 4 (Fig.~\ref{fig:4-cycle-view-trees}), the degree configuration $\vect d =(L,L,H,H)$ for join variables $(A,B,C,D)$, and an update $\delta R$. This update triggers updates to the views along the path from $R$ to the root of the view tree. 
    We first consider $\delta V_1(A,B,C) = \delta R(A,B)\cdot S(B,C).$ 
    The degree constraints that hold at the leaves of $\delta V_1(A,B,C)$ in the delta view tree are: $$\C_1=\DC(\leaves(\delta V_1(A,B,C)),\vect d) = \{(BC|\emptyset, N), (BC| B, N^\epsilon), (C|\emptyset, N^{1-\epsilon}), (A|\emptyset, 1), (B|\emptyset, 1)\}.$$ We obtain the exponent $s$ of an upper bound $\bigO(N^s)$ on the time to compute $\delta V_1(A,B,C)$ using the following linear program that assigns a  weight (positive number) to each constraint. We have the following vector of weights: $
          \vect w
          =
          \bigl(
            w_{1},
            w_{2},
            w_{3},
            w_{4},
            w_{5}
          \bigr).
        $
    We assume that the order of the weights follows the order of the above constraints. The linear program is as follows:
    \begin{align*}
      \text{minimize}\quad
        & w_{1}\cdot 1
        + w_{2}\cdot \epsilon
        + w_{3}\cdot (1-\epsilon)
        + w_4 \cdot 0
        + w_5 \cdot 0 &&\\  
      \text{subject to}\quad
        & w_{4} \;\geq\; 1 && \quad \text{// covers } A\\
        & w_{1} + w_{5} \;\geq\; 1 && \quad \text{// covers } B\\
        & w_{1} + w_{2} +  w_{3} \;\geq\; 1 && \quad \text{// covers } C\\
        & \mathrlap{w_{1}, \dots, w_{5} \;\geq\; 0}
    \end{align*}
    The objective is the sum of all weights, each multiplied by the base-$N$ logarithm of the bound in the corresponding constraint; note the coefficient of $w_4$ and $w_5$ is $\log_N 1 = 0$, so these weights do not contribute to the objective. The program has one inequality per variable: the sum of the weights of those constraints that cover the variable must be at least 1. For the optimal solution, it is enough to consider the weight vectors that are the vertices of the convex polyhedron given by the feasible region of the above linear program: \((1, 0, 0, 1, 0), (0, 1, 0, 1, 1), (0, 0, 1, 1, 1)\).
    
    None of these vertices depend on $\epsilon$, yet the program solutions may be parameterized by $\epsilon$. For instance, 
    the solution given by the 3rd vertex, which sets $w_{3}=w_{4}=w_{5}=1$ and all other weights to 0, is $1-\epsilon$, whereas the solution given by the 2nd vertex, which sets $w_{2}=w_{4}=w_{5}=1$ and all other weights to 0, is $\epsilon$. The minimal solution is therefore at most $\min(\epsilon, 1-\epsilon)$. The delta view $\delta V_1$ is a join query that can be computed using an adaptation of a worst-case optimal join algorithm~\cite{hung_algo} in time $\bigO(N^{\min(\epsilon, 1-\epsilon)})$.

    We next consider $\delta V_2(A,C) = \sum_B \delta V_1(A,B,C)$. The time to compute $\delta V_2$ is asymptotically the same as for $\delta V_1$, since the former can be computed in one pass over the latter. In the paper, we introduce a systematic approach to upper bound the time to compute conjunctive queries with bound variables, such as $\delta V_2$. For this, we consider the set of constraints that hold at the leaves of $\delta V_2$, which is $\C_1$, projected onto different supersets of the set of its free variables $\{A,C\}$, i.e., $\C_2 = \C_1[AC] = \{(C|\emptyset, N), (C|\emptyset, N^{1-\epsilon}), (A|\emptyset, 1)\}$ and $\C_1=\C_1[ABC]$. We next discuss each of these two cases.

    Using $\C_2$, we can define the join query $Q_{\C_2}$ that over-approximates $\delta V_2$ in the sense that all tuples in $\delta V_2$ are also in $Q_{\C_2}$ (multiplicities are ignored): 
    $Q_{\C_2}(A,C) = \delta R'(A)\cdot S'(C)$, where $\delta R'$ is the projection of $\delta R$ onto $A$ and $S'$
    is the projection of $S$ onto $C$ (Def.~\ref{def:guarding-query}). The query $Q_{\C_2}$ can be computed in time $\bigO(N^{1-\epsilon})$: $\C_2$ states that we have one $A$-value and at most $N^{1-\epsilon}$ $C$-values. The output of $\delta V_2$ can be recovered by semi-join reducing $Q_{\C_2}$ with $\delta V_1$. The time to compute $\delta V_2$ is therefore asymptotically the same as for $Q_{\C_2}$. 
    
    In case we use $\C_1$, we retain $B$ and with it the light constraint and the update constraint on $B$. The join query defined by $\C_1$ is precisely $\delta V_1$, which is an over-approximation of $\delta V_2$ in the sense that we can recover the tuples of $\delta V_2$ from $\delta V_1$ in the same $\bigO(N^{\min(\epsilon, 1-\epsilon)})$ time as for  computing $\delta V_1$.
    
    We next consider $\delta V_3(A,C,D) = \delta V_2(A,C)\cdot U(D,A)$. We take the set $\C_3=\DC(\leaves(\delta V_3),\vect d)$ of constraints that hold at the leaves of $\delta V_3$ and project it onto the two possible supersets of the set of its free variables: $\C'_3=\C_3[ACD]=\{(C|\emptyset, N), (C|\emptyset, N^{1-\epsilon}), (A|\emptyset, 1),(D|\emptyset, N^{1-\epsilon}), (AD|\emptyset, N),(AD|A, N^{\epsilon})\}$ and $\C''_3=\C_3[ABCD]=\C_1\cup\{(D|\emptyset, N^{1-\epsilon}), (AD|\emptyset, N),(AD|A, N^{\epsilon})\}$. 
    
    In case of $\C'_3$, we use the join query $Q_{\C'_3}$ that over-approximates $\delta V_3$: $Q_{\C'_3}(A,C,D) = \delta R'(A)\cdot S'(C)\cdot U(D,A)$. We use that: There is one $A$-value that is also light, $D$ is heavy, and $C$ is heavy. This yields a $N^{\min(\epsilon,1-\epsilon)}$ bound on the number of $D$-values. The time is then $\bigO(N^{\min(\epsilon,1-\epsilon)+1-\epsilon}) = \bigO(N^{\min(1,2-2\epsilon)})$. 
    
    In case of $\C''_3$, we use the join query $Q_{\C''_3}(A,B,C,D) = \delta R(A,B)\cdot S(B,C)\cdot U(D,A)$. This query over-approximates $\delta V_3$ in the sense that for every tuple $\vect t$ in $\delta V_3$ there is a tuple $\vect t'$ in $Q_{\C'_3}$ such that $\vect t'.(ACD) = \vect t$. To compute this query, we observe that the number of $C$-values or of $D$-values is upper bounded by $N^{\min(\epsilon,1-\epsilon)}$, since both $C$ and $D$ are heavy and there is one $A$-value and one $B$-value and both values are light. This yields the compute time $\bigO(N^{2\min(\epsilon,1-\epsilon)})$.
    The analysis follows similarly for $\delta V_4(C,D) = \sum_A \delta V_3(A,C,D)$ and $\delta V_5(C,D) = \delta V_4(C,D) \cdot T(C,D)$. 
    
    Our approach needs to account for all such evaluation strategies for each delta view, as their time may depend on $\epsilon$ and it is only clear which ones support the lowest overall update time once we find the value of $\epsilon$ that minimizes the update time across all degree configurations and relation updates. 
\end{example}

The update time for a view tree and degree configuration is given by the maximum time to compute any delta view in the view tree. For a degree configuration, the update time is the minimum update time over all possible view trees. The overall update time is the maximum update time over all degree configurations. We call the base-$N$ logarithm of the overall update time for a query $Q$ the {\em maintenance width} of $Q$ and denote it by $\mw(Q)$ (Def.~\ref{def:maintenance-width}). This width is given by a nesting of minimizations and maximizations of linear functions in the threshold parameter $\epsilon$. Its minimal value, and $\epsilon$ that gives its minimal value, can be computed using a set of linear programs (Sec.~\ref{subsec:computing-mw}).

\begin{example}
\rm
    Table~\ref{4-cycle-update-times} gives the base-$N$ logarithm of the update time for each degree configuration, when using the best view tree for that configuration. Each of the six view trees from Fig.~\ref{fig:4-cycle-view-trees} is used by at least one configuration. No further view trees are needed to achieve the overall lowest update time.  
    The maintenance width $\mw(Q)$ of the 4-cycle query $Q$ is then the minimum over all expressions in the table column for the base-$N$ logarithm of the update time (we ignore wlog the expressions for all remaining view trees as they do not yield a smaller maintenance width). This width can be computed as the minimization of optimal solutions of linear programs:
    \begin{align*}
    \mw(Q) &=  \min_\epsilon \max(\epsilon, 1-\epsilon, \max(\min(2\epsilon, 2-2\epsilon), \min(2\epsilon, 1))) \\
    &=  \min_\epsilon \max(\epsilon, 1-\epsilon, \min(2\epsilon, 2-2\epsilon), \min(2\epsilon, 1)) \\
    &\overset{*}{=} \min_\epsilon \min(
    \max(\epsilon, 1-\epsilon,2\epsilon,2\epsilon), 
    \max(\epsilon, 1-\epsilon,2\epsilon,1),\\
    &\hspace*{4.75em}
    \max(\epsilon, 1-\epsilon,2-2\epsilon,2\epsilon), 
    \max(\epsilon, 1-\epsilon,2-2\epsilon,1))\\
    &\overset{+}{=} \min (\min_\epsilon \max(\epsilon, 1-\epsilon,2\epsilon,2\epsilon) \text{ s.t. } 0\leq \epsilon \leq 1\\
    &\hspace*{3em} \min_\epsilon\max(\epsilon, 1-\epsilon,2-2\epsilon,1) \text{ s.t. } 0\leq \epsilon \leq 1\\
    &\hspace*{3em} \min_\epsilon\max(\epsilon, 1-\epsilon,2-2\epsilon,2\epsilon) \text{ s.t. } 0\leq \epsilon \leq 1\\
    &\hspace*{3em} \min_\epsilon\max(\epsilon, 1-\epsilon,2-2\epsilon,1) \text{ s.t. } 0\leq \epsilon \leq 1 )
\end{align*}

%%%%%%%%%%%%%%%
\begin{table}[t]
\renewcommand{\arraystretch}{1.2}
\centering
\begin{tabular}{|c|c|r||c|c|r|}
\hline
\textbf{Configuration} & \textbf{View} & \textbf{\multirow{2}{*}{$\log_N$ UpdateTime}} & \textbf{Configuration} & \textbf{View} & \textbf{\multirow{2}{*}{$\log_N$ UpdateTime}} \\
$(A, B, C, D)$ & \textbf{Tree} & & $(A, B, C, D)$ & \textbf{Tree} &  \\
\hline
$*, L, *, L$ & $1$ & $\epsilon$ & $H, L, L, H$ & $3$ & $f(\epsilon)$ \\
\hline
$L, L, L, H$ & $2$ & $\epsilon$ & $H, L, H, H$ & $1$ & $1 - \epsilon$ \\
\hline
$L, L, H, H$ & $4$ & $f(\epsilon)$ & $H,H,L,L$ & $6$ & $f(\epsilon)$ \\
\hline
$L, H, L, *$ & $2$ & $f(\epsilon)$ & $H,H,L,H$ & $2$ & $1 - \epsilon$ \\
\hline
$L, H, H, L$ & $5$ & $f(\epsilon)$ & $H,H,H,*$ & $1$ & $1 - \epsilon$ \\
\hline
$L, H, H, H$ & $2$ & $1 - \epsilon$ & & & \\
\hline
\end{tabular}
\caption{The base-$N$ logarithm of the update time  for each degree configuration and a specific view tree from Fig.~\ref{fig:4-cycle-view-trees}.
Here, $f(\epsilon) \defeq \max(\min(2\epsilon, 2-2\epsilon)), \min(2\epsilon, 1))$; $(*)$ in the degree configuration indicates that the join variable can be either heavy or light.}
\label{4-cycle-update-times}
\end{table}
%%%%%%%%%%%%%%%

The equality (*) holds due to the distributivity of $\max$ over $\min$, while the equality (+) is due to the commutativity of the two $\min$ functions. We then have to take the minimum of the optimal solutions of four optimization problems, which can be encoded as linear programs. We show the equivalent linear program for the first optimization problem above:
\begin{align*}
    &\min_\epsilon\ \ q \ \ \text{ s.t. }
    \hspace*{2em} q \geq \epsilon
    \hspace*{2em} q \geq 1-\epsilon
    \hspace*{2em} q \geq 2\epsilon
    \hspace*{2em} 0 \leq \epsilon \leq 1.    
\end{align*}
The optimal solution is $2/3$ and obtained for $\epsilon=1/3$.

Once we know the value of $\epsilon$, we can decide which evaluation strategy is best for each delta view in each view tree and for each degree configuration.
\end{example}
We are now ready to state the main technical result of this paper. Our incremental view maintenance approach follows the setting of prior work~\cite{IVMeps:ICDT:2019,TriangleQuery,IVMeps:PODS:2020}. It 
has a preprocessing phase, in which the heavy-light partitioning happens and the views of the used view trees are materialized. Then, it receives a sequence of single-tuple updates (inserts and deletes) and processes one update at a time. After each update, it can resolve requests to enumerate the tuples in the query output with constant delay. After some updates, the heavy/light degree assignment of some values may become invalid as: (1) light (heavy) values become heavy (light) according to the current threshold $N^\epsilon$; or (2) the threshold itself changes significantly as $N$ changes~\cite{IVMeps:ICDT:2019}. In the first case, we need to move values between the light and heavy parts of relations; this is called minor rebalancing. In the second case, we need to recompute the partitioning and then the views from scratch; this is called major rebalancing.  Major and minor rebalancing only need to be performed after a sufficiently large number of updates so that the partitioning still guarantees the desired asymptotic complexity for the update time. By amortizing the cost of rebalancing over many updates, the update time remains the same, albeit amortized.   

\begin{restatable}{theorem}{ThMain}
    \label{th:main}
    Any join query $Q$ can be maintained with $\bigO(N^{1+ \mw(Q)})$ preprocessing time, amortized \(\bigO(N^{\mw(Q)})\) single-tuple update time, and $O(1)$ enumeration delay, where $N$ is the size of the database at the time of update and $\mw(Q)$ is the maintenance width of $Q$.
\end{restatable}

\section{The Maintenance Width}
\label{sec:maintenance-width}
In this section, we introduce the \emph{maintenance width} and show how to compute it. This measure is central to our approach as it is the exponent of the update time incurred by our approach for maintaining the query output under database updates.

Consider a view \(V(\vect X)\) defined over a database that satisfies a set of degree constraints. Given an update \(\delta R_i\), we bound the time needed to compute the  delta view \(\delta V\) using a formulation~\cite{CSMA,hung_algo} as a linear optimization problem that generalizes the well-known AGM bound~\cite{AGM:2008} to incorporate the degree constraints. 

Adapting this framework to our setting presents two challenges. First, the original algorithm~\cite{CSMA,hung_algo} is designed for join queries (all variables are free); we must extend it to support views, where some variables are marginalized out. Second, our degree constraints depend on the threshold parameter \(\epsilon\), which is not fixed in advance. Consequently, we cannot solve the bounding linear program numerically. Instead, we solve it symbolically to obtain an objective value expressed as an explicit function of \(\epsilon\). We then select the optimal \(\epsilon\) that minimizes the worst-case maintenance cost across all degree configurations.

\subsection{The Polymatroid Bound under Degree Constraints}

In this section, we first revisit prior work on the polymatroid bound~\cite{hung_algo}. This bound is essential to our analysis of the maintenance time under different degree configurations and updates.

\begin{definition}[Polymatroid Bound under Degree Constraints] \cite[Eq.~48]{hung_algo}
\label{def:PBD}
Given a set \(\C=\{c_1,\ldots,c_m\}\) of degree constraints, we define the \emph{Polymatroid Bound under Degree Constraints}, denoted by \(\bound(\C)\), as the optimal value of the following linear program:

    \begin{align}
        \label{eq:bounding-LP}
        \text{minimize } & \sum_{c_i=(\vect Z|\vect Y,N^{p_{i}})\in\C} w_{i} \cdot p_{i}
        \\
        \label{eq:feasible1}
        \text{subject to } &\sum_{\substack{
        c_i=(\vect Z|\vect Y,N^{p_{i}})\in\C\\
        A\in \vect Z\setminus\vect Y
        }}w_{i}\geq 1 && \forall A\in\vars(\C)\\
        \label{eq:feasible2}
        &w_{i}\geq 0 && \forall c_i=(\vect Z|\vect Y,N^{p_{i}})\in\C
    \end{align}
\end{definition}

Prior work~\cite{hung_algo} showed that join queries can be computed worst-case optimally in the presence of acyclic degree constraints. In particular, given a join query $Q$ and an acyclic set $\C$ of constraints, where each constraint is guarded by database relations used in $Q$, then $Q$ can be evaluated in time 
$$\bigO\left(|\vars(Q)|\cdot |\C|\cdot~[N+N^{\bound(\C)}]\cdot \log N\right)~\cite{hung_algo}.$$
The \(\log N\) factor in the runtime is due to the use of B-tree indices and can be dropped 
by using hash maps to represent the relations. The additive \(N\log N\) factor is due to pre-computation. Notice that Inequalities~\ref{eq:feasible1},~\ref{eq:feasible2} are trivially satisfied by setting all variables \(w_{i}\) to 1, so \(\bound(\C)\) is always defined.

\subsection{Upper Bounding the Time to Compute a Delta View}

In our work we need to compute delta views that are not necessarily join queries. For this, we extend the prior work~\cite{hung_algo} to compute queries with arbitrary bound variables under degree constraints. We do this in several steps. First, we show how to derive possible join queries that over-approximate a given delta view. We tailor the set of constraints that are satisfied by the database to those relations and variables that are relevant to the over-approximation join queries.
We then use the algorithm from prior work~\cite{hung_algo} to compute the over-approximation join queries under specific constraints and take the over-approximation with the lowest time complexity. Below, we make this plan concrete.

\begin{definition}[Guarding Query]\label{def:guarding-query}
Let $\vect d$ be a degree configuration, $\delta T_R$ be a delta view tree for an update $\delta R_j$, $\delta V(\vect X)$ be a delta view in $\delta T_{R_j}$, and $\C=\DC(\leaves(\delta V(\vect X)),\vect d)$.
For any set $\vect Y$ of variables such that $\vect X\subseteq \vect Y\subseteq \vars(\C)$ and any acyclic set $\C'\in{\mathcal A}(\C[\vect Y])$ of degree constraints, we define the $\C'$-guarding query of $\delta V$ as the join query: 
\begin{equation*}
\label{eq:projected-ass-query}
Q_{\C'}(\vect Y) = R'_1(\vect X'_1)\cdot\ldots\cdot R'_k(\vect X'_k)
\end{equation*}
where $\set{R_1(\vect X_1),\dots,R_k(\vect X_k)}$ is the set of atoms at the leaves of $\delta V(\vect X)$ in $\delta T_{R_j}$ that guard the constraints in $\C'$ and includes $\delta R_j(\vect X_j)$, and $R'_i(\vect X'_i) = \sum_{\vect X_i\setminus \vect Y} R_i(\vect X_i)$ and $\vect X_i' = \vect X_i\cap \vect Y$ for $i\in[k]$.
\end{definition}

Ex.~\ref{ex:update-times} gives guarding queries for several delta views. 

\begin{remark}
    In Def.~\ref{def:guarding-query}, $\C'$ is a maximal acyclic subset of $\C[\vect Y]$. It contains all size, heavy, and update constraints from $\C[\vect Y]$, since these constraints are of the form $(\vect Z|\emptyset, N^{p_{\vect Z|\emptyset}})$ and cannot be part of a cycle in the constraint graph. 
    Therefore, $\C'$ may only miss some of the light constraints from  
    $\C[\vect Y]$. Yet the variables in these light constraints are also covered by a size constraint. The implication is twofold. First, the atoms that guard the constraints in $\C'$ are also those that guard the constraints in $\C[\vect Y]$. Second, $\vars(\C') = \vars(\C[\vect Y])$.
\end{remark}

\begin{example}
\rm
\label{ex:V5(A,C)}
    We discuss the computation of the delta view \(\delta V_5(A,C)\) in the delta view tree based on View Tree 1 (Fig.~\ref{fig:4-cycle-view-trees}) under the update \(\delta R\) and the degree configuration \(\vect d = (L,L,\cdot,\cdot)\), so where the join variables \(A\) and \(B\) are light and where -- for simplicity here -- we do not partition on the remaining join variables $C$ and $D$. The following degree constraints are guarded by the leaves of $\delta V_5(A,C)$: 
    \begin{align*}
        \C &= \DC(\leaves(\delta V_5),\vect d) \\
        &= \{(A|\emptyset,1),(B|\emptyset,1),(BC|B,N^\epsilon),(BC|\emptyset,N),
        (CD|\emptyset,N), (AD|A,N^\epsilon), (AD|\emptyset,N)\}.
    \end{align*}
    The \(\C\)-guarding query is \(Q_\C(A,B,C,D) = \delta R(A,B)\cdot S(B,C)\cdot T(C,D)\cdot U(A,D)\) and its output can be computed in worst-case optimal time \(\bigO(N^{\bound(\C)})\)~\cite{hung_algo}. Since the two queries \(Q_{\C}\) and \(\delta V_5\) have the same body, it follows that if we project the output of \(Q_{\C}(A,B,C,D)\) onto the variables \(\set{A,C}\) we obtain the output of \(\delta V_5\). However, \(\bound(\C)\) evaluates to \(1\) or \(2\epsilon\) (by using the first, second, and fifth constraints or by using the first, second, third, and sixth constraints to cover all variables) which gives the upper bounds \(\bigO(N)\) and \(\bigO(N^{2\epsilon})\). Depending on the value of \(\epsilon\), both of these bounds can be tight upper bounds on the time needed to compute the output of \(Q_\C\), but they are loose on the time needed to compute the output of \(\delta V_5\). As we show in later examples in this section, the output of \(\delta V_5\) can be computed in time \(\bigO(N^\epsilon)\). The reason is that the linear program for \(\bound(\C)\) needs to cover all variables, including the bound variables \(B\) and \(D\), and this increases the cost of its optimal solution. 
\end{example}

Any guarding query of a delta view {\em over-approximates} the delta view in the sense that  we can recover the output of a delta view from the output of any of its guarding queries.

\begin{lemma}\label{lemma:over-approximation}
    For any delta view $\delta V(\vect X)$ and any of its $\C[\vect Y]$-guarding queries $Q_{\C[\vect Y]}(\vect Y)$, where $\C$ is the set of constraints that are guarded by the leaves of $\delta V(\vect X)$ and $\vect X\subseteq \vect Y$, it holds: For any tuple $\vect x$ in the output of $\delta V$ over any database, there is a tuple $\vect y$ in the output of $Q_{\C[\vect Y]}$ over the same database such that $\vect x = \vect y.\vect X$.     
\end{lemma}
\begin{proof}
    Let $\{R_1(\vect X_1),\ldots, R_k(\vect X_k)\} = \leaves(\delta V(\vect X))$. The delta view is thus defined by $$\delta V(\vect X) = \sum_{(\bigcup_{i\in[k]} \vect X_i)\setminus \vect X} R_1(\vect X_1)\cdot\ldots\cdot R_k(\vect X_k).$$    
    Let $\C$ be the set of constraints that are guarded by the database relations $R_1,\ldots,R_k$. Since $\C$ always contains the size constraints of these relations, it holds that $\vars(\C) = \bigcup_{i\in[k]} \vect X_i$. By definition of the $\C[\vect Y]$-guarding query $Q_{\C[\vect Y]}(\vect Y)$ (Def.~\ref{def:guarding-query}), we have that $\vect X \subseteq \vect Y \subseteq \vars(\C)$ and:
    $$Q_{\C[\vect Y]}(\vect Y) = R_1'(\vect X_1')\cdot\ldots\cdot R_k'(\vect X_k'), \text{ where } R_i'(\vect X_i') = \sum_{\vect X_i\setminus \vect Y} R_i(\vect X_i) \text{ and } \vect X_i' = \vect X_i\cap \vect Y \text{ for } i\in[k].$$

    The set of tuples in the output of $\delta V$ is:  
    \begin{align} \label{eq:output-deltaV}
        \set{\vect t.\vect X\mid \sum_{\vect t\in\Dom(\vars(\C))} R_1 (\vect t.\vect X_1)\cdot\ldots\cdot R_k(\vect t.\vect X_k)>0},    
    \end{align}
    whereas the set of tuples in the output of $Q_{\C[\vect Y]}$ is:
    \begin{align} \label{eq:output-guarding-query}
        \set{\vect y\in\Dom(\vect Y)\mid R'_1(\vect y.\vect X'_1)\cdot\ldots\cdot R'_k(\vect y.\vect X'_k)>0}.
    \end{align}

    If the output of $\delta V$ is empty, then the statement of the lemma follows trivially. Otherwise, consider a tuple \(\vect x\) in the output of $\delta V$ given in Eq.~\eqref{eq:output-deltaV}. Since the sum in Eq.~\eqref{eq:output-deltaV} is positive, it must contain at least one positive term. Thus, there exists a tuple \(\vect t\in\Dom(\vars(\C))\) with \(\vect t.\vect X = \vect x\) such that \(R_1(\vect t.\vect X_1),\dots, R_k(\vect t.\vect X_k)\) are all positive. It then also follows that \(R'_1(\vect t.\vect X'_1)\cdot\ldots\cdot R'_k(\vect t.\vect X'_k)>0\), since $\vect X'_i \subseteq \vect X_i$ and $R'_i$ is defined by marginalizing $\vect X_i\setminus \vect Y$ from $R_i$, for $i\in[k]$. 
    We conclude that for any tuple \(\vect x\) such that \(\delta V(\vect x)>0\) there exists at least one tuple \(\vect y\in\Dom(\vect Y)\) with \(\vect y.\vect X=\vect x\) such that \(Q_{\C[\vect Y]}(\vect y)>0\). 
\end{proof}

Guarding queries for a delta view are join queries, so we can compute them efficiently under an acyclic set of degree constraints~\cite{hung_algo}. Furthermore, we can compute the delta view in time proportional to the time need to compute any of its guarding queries.

\begin{lemma}\label{lemma:compute-time-delta-view}
    Given a database of size $N$, a delta view tree $\delta T_R$, a delta view $\delta V(\vect X)$ in $\delta T_R$, and any of the $\C$-guarding queries of $\delta V(\vect X)$, where $\C$ is an acyclic set of constraints, then $\delta V$ can be computed in time $\bigO(N^{\bound(\C)})$ given that the child views of $\delta V(\vect X)$ in $\delta T_R$ are already computed.
\end{lemma}

\begin{proof}
    We compute $\delta V$ in two steps. In the first step, we compute the join query $Q_{\C}$ in time $\bigO(N^{\bound(\C)})$ using the algorithm from prior work~\cite{hung_algo}. Its output is, however, an over-approximation of the output of $\delta V$ as stated in Lemma~\ref{lemma:over-approximation}. In the second step, we first create a relation $V'$ that is the projection of $Q_{\C}$'s output onto the set of variables of the child views of $\delta V$ in $T$. We next semi-join reduce $V'$ with each of the child views of $\delta V$. The multiplicity of each tuple $\vect t$ in $V'$ becomes the product of the multiplicities of the tuples in the child views whose join make $\vect t$ in $V'$. By performing the computation of the delta views bottom-up in the delta view tree, we ensure that the correct multiplicities of the child views are computed before those of the parent views. Finally, we marginalize out all variables of $V'$ except $\vect X$ to obtain the output of $\delta V$. The second step also takes time proportional to the size of $Q_{\C}$'s output, so in $\bigO(N^{\bound(\C)})$ time.
\end{proof}

There are two immediate implications of Lemma~\ref{lemma:compute-time-delta-view}. First, it gives a maintenance strategy for a delta view tree and an update $\delta R_j$: We proceed bottom-up from $\delta R_j$, first compute its parent delta view, then the parent of the parent, and so on until the delta view at the root. Second, it gives an upper bound on the update time for a delta view: This is the minimum of the computation time over all its guarding queries.

A question remains: Why should we consider all guarding queries of a delta view in order to upper bound the compute time for the delta view? Recall there is a guarding query for each set of variables that is a superset of the set of free variables of the delta view and a subset of the set of variables at the leaves of the delta view in the delta view tree. The key observation is that a guarding query with more variables does not necessarily have a higher computation time than another guarding query with less variables.

\begin{example}\label{ex:projection1}
\rm
    Returning to Ex.~\ref{ex:V5(A,C)}, let us consider an over-approximation of \(\delta V_5(A,C)=\sum_{B,D} \delta R(A,B)\cdot S(B,C)\cdot T(C,D)\cdot U(D,A)\). Projecting the set \(\C\) of degree constraints onto \(\set{A,C}\) gives
    \[
        \C[AC] = \set{(A|\emptyset,1), (C|\emptyset, N), (A|\emptyset,N)}
    \]
    with the \(C[AC]\)-guarding query \(Q_{\C[AC]}(A,C) = \delta R'(A)\cdot S'(C)\cdot T'(C)\cdot U'(A)\). By Lemma~\ref{lemma:over-approximation}, we know that this join query is an over-approximation of the conjunctive query \(\delta V_5(A,C)\). We can bound the time needed to compute \(Q_{\C[AC]}(A,C)\) by \(\bigO(N^{\bound(\C[AC])})\) using Lemma~\ref{lemma:compute-time-delta-view}. This bound is $\bigO(N)$ and  obtained by covering the variable $A$ with the first constraint and the variable $C$ with the second constraint. While correct, this bound is  loose since it does not exploit the lightness information on $B$.
\end{example}

The previous example highlights a trade-off: projecting the set $\C$ of constraints onto a subset of its variables guarantees the correctness of the bound, but it may discard useful constraints (like those conditioned on \(B\)).  Although this requires covering more variables in the linear program, it allows us to retain more constraints in the projection. Thus, we need to look at all the subsets \(\vect X'\) with \(\vect X\subseteq \vect X'\subseteq \vars(\C)\).

\begin{example}
\rm
    We continue Ex.~\ref{ex:projection1} and now consider an over-approximation of \(\delta V_5(A,C)\) by projecting the set \(\C\) of degree constraints onto the superset \(\set{A,B,C}\). This new projection preserves a light constraint conditioned on $B$ and gives
    \begin{equation*}
    \C[ABC] = \{(A|\emptyset,1), (B|\emptyset, 1),(BC|\emptyset,N), (C|\emptyset,N),(A|\emptyset,N), (BC|B,N^{\epsilon})\}
    \end{equation*}
    By Lemma~\ref{lemma:over-approximation}, we conclude that \(Q_{\C[ABC]} = \delta R(A,B)\cdot S(B,C)\cdot T'(C)\cdot U'(A)\) is an over-approximation of \(\delta V_5(A,C)\). The time needed to compute \(Q_{\C[ABC]}\) is bounded by \(\bigO(N^{\bound(\C[ABC])})\), which is \(\bigO(N^{\epsilon})\) (by covering $A,B,C$ using the first two and the last constraints). This is tighter than \(\bigO(N)\) for all \(\epsilon < 1\). 
    
    Yet, if we would project $\C$ onto the full set \(\set{A,B,C,D}\) of variables, then we would regain the constraints on \(D\) at the price of having to cover \(D\) in the linear program. As shown in Ex.~\ref{ex:V5(A,C)}, this gives the upper bound of \(\bigO(N^{\min(1, 2\epsilon)})\), which is also worse than $\bigO(N^{\epsilon})$ for all $\epsilon \in (0, 1)$.
\end{example}

\subsection{Symbolic Optimization}

As illustrated in Ex.~\ref{ex:projection1}, the degree constraints can depend on \(\epsilon\) (e.g., \(N^\epsilon\)), while the value of \(\epsilon\) is not fixed in advance. Consequently, in the linear program of Def.~\ref{def:PBD}, we treat the coefficients \(p_{i}\) as functions of \(\epsilon\) in order to obtain an explicit closed form expression for the optimal value as a function of $\epsilon$. All the degree constraints obtained by the function \(\DC(\cdot,\cdot)\) have the form
$c_i=(\vect Z \mid \vect Y, N^{p_{i}} = N^{f_{i}(\epsilon)})$, where
$f_{i}$ is either affine in $\epsilon$ (it is $\epsilon$ or
$1-\epsilon$) or constant (it is $0$ or $1$). For each constraint $c_i$, we view the coefficient $p_{i}$ as the function $f_{i}(\epsilon)$.

Crucially, while the \emph{objective function} in Eq.~\eqref{eq:bounding-LP} varies with \(\epsilon\), the \emph{feasible region} \(\mathcal{P}(\C)\) defined by Eq.~\eqref{eq:feasible1} and \eqref{eq:feasible2} depends only on the sets \(\C\) and \(\vect X\), and is independent of \(\epsilon\). 
It is a standard result in linear programming that, if an optimum exists, then it is attained at a vertex of the feasible polyhedron \(\mathcal{P(\C)}\) \cite{LP-vertices}. Since \(\mathcal{P(\C)}\) is defined by a fixed set of constraints, it has a finite set of vertices (and independent on the database size), which we denote by \(\mathcal{S}_{\C}\).

The vertices \(\vect w = (w_{i})_{c_i\in\C} \in \mathcal{S}_{\C}\) are constant vectors independent of \(\epsilon\). We can compute these weight vectors \(\vect w\) once and reuse them to evaluate the cost for any \(\epsilon\). 
For a specific vertex \(\vect w \in \mathcal{S}_{\C}\) and parameter \(\epsilon\), the symbolic cost is given by:
\[
  \obj(\vect w,\C)(\epsilon)
  \defeq
  \sum_{c_i\in\C} w_{i}\cdot f_{i}(\epsilon).
\]

%=\left(\vect Z| \vect Y,N^{f_{i}(\epsilon)}\right)

\nop{\ahmet{We could skip the following paragraph. I wrote some explanation following Definition~\ref{def:maintenance-width}, which might be enough.}

The symbolic bound \(\bound(\C)\) is then given by the minimum of the objective values attained at the vertices in \(\mathcal{S}_{\C}\):
\[
  \bound(\C)(\epsilon)
  \;=\;
  \min_{\vect v \in \mathcal{S}_{\C}} \obj(\vect v,\C)(\epsilon).
\]
This defines a concave piecewise affine function of \(\epsilon\). Combining this with \Cref{lemma:bound-update-view}, we obtain the following symbolic equation for \(\bound^*\):
\begin{equation}
    \label{eq:symbolic-pbd-star}
  \bound^*(\delta V(\vect X),\C)(\epsilon)
  \;=\;
  \min_{\C'\in\mathcal{A}(\C)} 
  \min_{\vect X\subseteq \vect X'\subseteq \vars(\C')}
    \min_{\vect w \in \mathcal{S}_{\C'[\vect X']}} 
      \obj(\vect w,\C'[\vect X'])(\epsilon).
\end{equation}
In particular, for a delta view tree \(\delta T\), degree configuration \(\vect d\), and update \(\delta R_j\), we can use Eq.~\ref{eq:symbolic-pbd-star} to express the time to compute all the delta views \(\delta V_1,\dots,\delta V_k\) along the path from \(\delta R_j\) to the root of \(\delta T\) as the following function of \(\epsilon\):
\[
\max_{V(\vect X)\in \set{\delta V_1,\dots,\delta V_k}}\bound^*(\delta V(\vect X),\DC(\leaves(\delta V(\vect X),\delta T), \vect d, \delta R_j))(\epsilon)
\]

\ahmet{We could skip until here.}}

We are now ready to introduce the notion of \emph{maintenance width}, which captures the maintenance cost for a given query.
In order to find the lowest maintenance cost for a query, we look at the degree configuration that induces the most expensive set of degree constraints, for which we pick the cheapest view tree. For this view tree, we pick the most expensive view to maintain under the most expensive update.\footnote{Alternatively, the maintenance width can be seen as the outcome of a game where we choose \(\epsilon\) and the view tree \(T\) to minimize cost, while an adversary chooses the data statistics \(\vect d\), the update \(\delta R_i\), and the specific view in \(T\) to maximize cost.} 

\begin{definition}[Maintenance Width]\label{def:maintenance-width}
    For a join query \(Q\), the maintenance width of \(Q\) is

    \begin{align}
    \label{eq:mw}
        \mw(Q) & \defeq
        \min_{\epsilon\in[0,1]}~
        \max_{\vect d\in\mathcal D(Q)} 
        \min_{T\in\calT(Q)}~
        \max_{\substack{V(\vect X)\in T\\
        R(\vect{Y})\in\leaves(V(\vect{X}), T)
        }}~~~ \min_{\substack{\C\subseteq
        \DC(\leaves(\delta V(\vect X), \delta T_R),\vect d)\\
        \vect X\subseteq \vect X'\subseteq \vars(\C)\\
        \C'\in\mathcal{A}(\C[\vect X'])\\
        \vect w \in \mathcal{S}_{\C'}}}
        \obj(\vect w,\C')(\epsilon).
    \end{align}    
\end{definition}

\begin{remark}
Def.~\ref{def:maintenance-width} states that for every degree configuration $\vect d$ and view tree $T$, we take the maximum compute time over all views $V$ of all its delta views $\delta V$ subject to updates at each of its leaves $\delta R$. This computes the maximum update time per view over updates at any of its leaves in $\delta T_R$ and then takes the maximum over all views in the view tree. In contrast, our maintenance approach propagates each update $\delta R$ from a leaf along the path to the root of the delta view tree $\delta T_R$ for $\delta R$ and takes the maximum time to compute the delta views that are the ancestors of $\delta R$ in $\delta T_R$. Yet both ways  to account for the maintenance time for a given view tree yield the same maximum time to update all the views of $T$ under all updates at the leaves. This justifies the equivalent formulation in the definition.    
\end{remark}

We explain how Lemma~\ref{lemma:compute-time-delta-view} implies that any join query $Q$ can be maintained with update time $\bigO(N^{\mw(Q)})$
as stated in Theorem~\ref{th:main}.
Fix an $\epsilon \in [0,1]$. Given a degree configuration $\vect d$ and a view tree $T$ for $Q$, let 
\begin{align*}
\bound^*(\vect{d}, T)(\epsilon) =&~
        \max_{V(\vect X)\in T}
        ~\max_{R(\vect{Y})\in\leaves(V(\vect{X}), T)} ~
        \bound^*(\vect d, \delta V(\vect X), \delta T_R)(\epsilon) \text{ and}\\
\bound^*(\vect d, \delta V(\vect X), \delta T_R)(\epsilon) =& 
\min_{\C\subseteq
        \DC(\leaves(\delta V(\vect X), \delta T_R),\vect d)}
        ~\min_{\vect X\subseteq \vect X'\subseteq \vars(\C)} ~\min_{\C'\in\mathcal{A}(\C[\vect X'])}
\min_{\vect w \in \mathcal{S}_{\C'}}
       \obj(\vect w,\C')(\epsilon).
 \end{align*}

Consider an atom $R(\vect{Y})$ in a view tree $T$ for $Q$ and an update 
$\delta R$ to a relation $R$ in $Q$.
To maintain $T$ under $\delta R$, we derive the delta view tree $\delta T_R$, and compute bottom-up all delta views $\delta V(\vect{X})$ along the path from $R(\vect{Y})$ to the root of $\delta T_R$.
By Lemma~\ref{lemma:compute-time-delta-view}, each such delta view 
$\delta V(\vect{X})$ can be computed in time $\bigO(N^{\bound^*(\vect d, \delta V(\vect{X}), \delta T_R)(\epsilon)})$. The view tree $T$ can then be maintained in time $\bigO(N^{\bound^*(\vect d, T)(\epsilon)})$ by taking the maximum over all views in $T$ and for each such view $V$ the maximum over all atoms $R$ at the leaves of this view of the time to compute the $\delta V$ under the update $\delta R$. By Def.~\ref{def:maintenance-width}, we conclude  that $Q$ can be maintained in time $\bigO(N^{\mw(Q)})$ under any update to its input relations.

\subsection{Computing the Maintenance Width}
\label{subsec:computing-mw}

The computability of the maintenance width relies on the observation that all domains in Def.~\ref{def:maintenance-width}, excluding the interval \([0,1]\), are finite. Let \(\mathcal{H}\) denote the set of affine functions \(f(\epsilon) = \obj(\vect w, \C)(\epsilon)\) generated by the inner minimization over \(\C\) and \(\vect w\). The maintenance width can be viewed as the value of a logical expression involving finite \(\min\) and \(\max\) operations over functions in \(\mathcal{H}\).

We rewrite the nested minimization and maximization steps into a canonical min-of-max form by iteratively using the distributivity of \(\max\) over \(\min\): \(\max(a, \min(b,c)) = \min(\max(a,b), \max(a,c))\). Thus, there exists a finite index set \(\mathcal{M}\) and, for each \(i \in \mathcal{M}\), a finite set of affine functions \(\mathcal{F}_i \subseteq \mathcal{H}\) such that:
\begin{equation}
\label{eq:min_max_canonical}
    \mw(Q) = \min_{\epsilon\in[0,1]} \left( \min_{i \in \mathcal{M}} \left( \max_{f \in \mathcal{F}_i} f(\epsilon) \right) \right).
\end{equation}

Since the \(\min\) operator is commutative, we can swap the continuous minimization over \(\epsilon\) with the discrete minimization over \(i\). This yields:
\begin{equation}
\label{eq:mw_decomposed}
    \mw(Q) = \min_{i \in \mathcal{M}} \left( \min_{\epsilon\in[0,1]}\left( \max_{f \in \mathcal{F}_i} f(\epsilon)\right) \right).
\end{equation}

Eq.~\eqref{eq:mw_decomposed} reduces the optimization problem to finding the minimum of \(|\mathcal{M}|\) independent sub-problems. Each sub-problem aims to minimize the pointwise maximum of a finite set of affine functions over the unit interval. Consequently, the objective is a piecewise linear convex function of \(\epsilon\), which can be efficiently minimized by the following linear program:\\

\begin{align*}
    \text{minimize} \quad &v \\
    \text{subject to} \quad &v \geq f(\epsilon) \quad \forall f \in \mathcal{F}_i \\
    &0 \leq \epsilon \leq 1
\end{align*}
By solving these linear programs, we obtain both the exact value of \(\mw(Q)\) and the optimal parameter \(\epsilon^*\) for which this value is obtained. Fixing \(\epsilon^*\) determines the threshold for data partitioning. We then proceed to select the optimal view tree for each degree configuration \(\vect d\) by choosing the view tree \(T\) that minimizes the maintenance width for this fixed \(\epsilon^*\). In case of ties, we deterministically select a canonical tree (e.g.,~in lexicographical order). We denote this selected tree by \(T(\vect d)\). The set of these selected trees constitutes the set of \emph{active view trees} that we use in our maintenance algorithm. In practice, this set can be significantly smaller than the set of all possible view trees. For instance, we only need six view trees for the optimal maintenance of the $4$-cycle query. Furthermore, for every update \(\delta R_j\) and view \(V(\vect X)\) we determine the optimal set \(\vect X'\) of variable to use for computing a guarding query that is an over-approximation of \(\delta V\) under the update \(\delta R_j\).

%%%%%%%%%%%%%

\section{Comparison with Prior 
Width Measures}
\label{sec:comparison}
In this section, we compare the  maintenance width with other common width measures and discuss our choice of view trees as the maintenance strategy of our approach.

\subsection{Dynamic Width}

The maintenance width generalizes the previously introduced notion of {\em dynamic width}, which defines the update time for maintaining queries under simple size constraints~\cite{DynamicWidth}.
Any query $Q$ can be maintained with $\bigO(N^{\dw(Q)})$ update time, where $\dw(Q)$ denotes the dynamic with of $Q$.
To the best of our knowledge, the dynamic width defines the best update time achieved by approaches that do not rely on heavy-light partitioning.
The F-IVM column in Table~\ref{tab:comparisons} lists the update times for several queries, these times follow the dynamic width\footnote{Although the work on F-IVM~\cite{FIVM} did not formally introduce the notion of dynamic width, the update time achieved by F-IVM for any query $Q$ is of the form $\bigO(N^{\dw(Q)})$, where $\dw(Q)$ denotes the dynamic width of $Q$. The notion of dynamic width was formally introduced in subsequent work~\cite{DynamicWidth}.}.
For any hierarchical query $Q$, we have $\mw(Q) = \dw(Q) = 0$. Hence, both our approach and F-IVM achieve $\bigO(1)$ update time.
For the bow tie query $Q$, we have 
$\mw(Q) = \dw(Q) = 1$, implying that both our approach and F-IVM achieve $\bigO(N)$ update time.
For all other queries in Table~\ref{tab:comparisons}, we have 
$\mw(Q) < \dw(Q)$, which means that our approach outperforms F-IVM.

Computing the dynamic width requires iterating over all view trees, then over the views in each view tree, and finally over the leaves under each view. For the maintenance width, we must further iterate over all  degree configurations, which accounts for the increased complexity of the definition, yet which may yield a smaller width value.
In the following, we introduce the dynamic width $\dw$ and show that $\mw(Q)$ is  upper-bounded by $\dw(Q)$ for any join query $Q$.

%This implies that the update time achieved by our approach is asymptotically upper-bounded by the best update time achieved in prior work. 
%In this section we compare the maintenance width with the dynamic width of a query~\cite{DynamicWidth}. 

We start by recalling the {\em fractional edge cover number} of a set of variables with regard to a query~\cite{AGM:2008}. Given a query $Q$ and a set $\vect{Y} \subseteq \vars(Q)$ of variables, the fractional edge cover number 
$\rho_Q^*(\vect{Y})$ of $\vect{Y}$ with regard to $Q$ is the cost of the optimal solution of the following linear program:
\begin{align*}
  \text{minimize}\quad
    & \sum_{R(\vect{X}) \in \at(Q)} w_{R(\vect{X})} && \\
  \text{subject to}\quad
    & \sum_{R(\vect{X}): B \in \vect{X}} w_{R(\vect{X})} \geq 1 && \text{ for all } B \in \vect{Y}\\
    & w_{R(\vect{X})} \geq 0 && \text{ for all } R(\vect{X}) \in \at(Q)
\end{align*}

Given a set $\mathcal L$ of atoms of a query $Q$, we denote by $Q_{\mathcal L}$ the join query whose body is the conjunction of the atoms in $\mathcal L$. We can now define the dynamic width of a join query\footnote{We introduce here a simplified version of the dynamic width restricted to join queries, whereas the original definition~\cite{DynamicWidth} is for the more general conjunctive queries. This is because our maintenance width is defined here for join queries only.}:

\nop{
Whereas the original definition minimizes over a restricted set of view trees that follow specific partial orders of the query variables, Def.~\ref{def:dynamic_width} ranges over {\em all} possible view trees. The proof of Prop.~\ref{prop:comparison_maintenance_dynamic} shows that even for this more general variant of $\dw$, it holds that $\mw(Q) \leq \dw(Q)$ for any join query $Q$.
}

\begin{definition}[Dynamic Width]\label{def:dynamic_width}
    For any join query \(Q\), the dynamic width of \(Q\) is
    \begin{equation}
        \dw(Q)\defeq 
        \min_{T\in\mathcal{T}(Q)}~
        \max_{\substack{ V(\vect{X}) \in T\\
        R(\vect{Y}) \in \leaves(V(\vect{X}),T)}}~~~
      \rho^*_{Q_{\leaves(V(\vect{X}),T)}}(\vect{X}\setminus \vect{Y}).
    \end{equation}
\end{definition}

\begin{proposition}
\label{prop:comparison_maintenance_dynamic}
    For any join query $Q$, it holds 
    $\mw(Q) \leq \dw(Q)$.
\end{proposition}

\begin{proof}
The proof is implied by the following chain of (in)equalities. We explain each step below.

    \begin{align*}
        \mw(Q) & \defeq
        \min_{\epsilon\in[0,1]}~
        \max_{\vect d\in\mathcal D(Q)} 
        \min_{T\in\calT(Q)}~
        \max_{\substack{V(\vect X)\in T\\
        R(\vect{Y})\in\leaves(V(\vect{X}),T)
        }}~~~ \min_{\substack{\C\subseteq
        \DC(\leaves(\delta V(\vect X), \delta T_R),\vect d)\\
        \vect X\subseteq \vect X'\subseteq \vars(\C)\\
        \C'\in\mathcal{A}(\C[\vect X'])\\
        \vect w \in \mathcal{S}_{\C'}}}
        \obj(\vect w,\C')(\epsilon) \\
        &\overset{(1)}{\leq}
        \max_{\vect d\in\mathcal D(Q)} 
        \min_{T\in\calT(Q)}~
        \max_{\substack{V(\vect X)\in T\\
        R(\vect{Y})\in\leaves(V(\vect{X}), T)
        }}~~~ \min_{\substack{\C\subseteq
        \DC(\leaves(\delta V(\vect X), \delta T_R),\vect d)\\
        \vect X\subseteq \vect X'\subseteq \vars(\C)\\
        \C'\in\mathcal{A}(\C[\vect X'])\\
        \vect w \in \mathcal{S}_{\C'}}}
        \obj(\vect w,\C')(1) \\
        &\overset{(2)}{=} 
        \min_{T\in\calT(Q)}~
        \max_{\substack{V(\vect X)\in T\\
        R(\vect{Y})\in\leaves(V(\vect{X}),T)
        }}~~~ \min_{\substack{\C\subseteq
        \DC(\leaves(\delta V(\vect X), \delta T_R),\vect L)\\
        \vect X\subseteq \vect X'\subseteq \vars(\C)\\
        \C'\in\mathcal{A}(\C[\vect X'])\\
        \vect w \in \mathcal{S}_{\C'}}}
        \obj(\vect w,\C')(1) \\
        &\overset{(3)}{\leq}
        \min_{T\in\calT(Q)}~
        \max_{\substack{V(\vect X)\in T\\
        R(\vect{Y})\in\leaves(V(\vect{X}),T)
        }}~~~ \min_{
        \vect w \in \mathcal{S}_{\hat{\C}}}
        \obj(\vect w,\hat{\C})(1) \\
        &\overset{(4)}{\leq} \min_{T\in\calT(Q)}~
        \max_{\substack{
        V(\vect X)\in T\\
        R(\vect{Y})\in\leaves(V(\vect X),T)}} \rho^*_{Q_{\leaves(V(\vect X),T)}}(\vect X \setminus \vect Y) \\
        &\defeq\ \dw(Q), 
    \end{align*}
    where $\vect L = (L)_{|\mathcal{J}_Q|}$ is the degree constraint  where all join variables are light and   
    $\hat{\calC} = \overline{\calC}[\vect{X}]$
with 
$\overline{\calC} = \{(\vect{Z}|\emptyset, N) | S(\vect{Z})$ 
    $\in$ $\leaves(\delta V(\vect{X}), \delta T_R) \wedge S \neq R\}$ $\cup$ $\{(B|\emptyset,1)|B \in \vect Y)\}$.

    Inequality~(1) is obtained  by fixing $\epsilon$ to $1$. 
    Equality~(2) is implied by the following observation.
    Consider a degree configuration $\vect d\in \mathcal D(Q)$
    and a vector $\vect{w} \in \mathcal{S}_{\C'}$ in the definition of $\mw(Q)$. Assuming that $\epsilon = 1$, any degree constraint resulting from a variable $A$ that is heavy in $\vect{d}$ is of the form $(A \mid \emptyset, N^{1-\epsilon} = N^0)$. This implies that for any vector $\vect{w} \in \mathcal{S}_{\C'}$, the additive factor added to  $\obj(\vect w,\C')(1)$ by such a constraint is of the form $w_i \cdot (1-\epsilon) = 0$. Hence, instead of maximizing $\obj(\vect w,\C')(1)$ over all possible degree configurations, it suffices to restrict to the configuration $\vect L = (L)_{|\mathcal{J}_Q|}$, where all variables are light. 
    
    Inequality (3) holds because $\overline{\calC} \subseteq \DC(\leaves(\delta V(\vect X), \delta T_R),$ $\vect L)$ and both $\overline{\calC}$ and
    $\hat{\calC} = \overline{\calC}[\vect{X}]$ are acyclic,
    since 
    the graph associated with $\overline{\calC}$ does not have any edge.

    For Inequality (4), consider a view $\delta V(\vect{X})$ in a delta view tree for an update $\delta R(\vect{Y})$. Assume that $\hat{C} = \{c_1, \ldots , c_m\} \cup 
  \{c^B\}_{B \in \vect{X} \cap \vect{Y}}$, where each 
  $c_i$ is of the form $(\vect{Z}|\emptyset, N)$ 
  and each $c^B$ is of the form $(B|\emptyset, 1)$.
  The linear program 
    determining 
    $\min_{\vect w \in \mathcal{S}_{\hat{\C}}}
        \obj(\vect w,\hat{\C})(1)$ is as follows:

\begin{align}
        \text{minimize} & \sum_{c_i = (\vect Z|\emptyset,N)} w_i \cdot 1 + \sum_{B \in \vect{X} \cap \vect{Y}} w^B \cdot 0 \nonumber\\
        \text{subject to} & \sum_{\substack{
        c_i = (\vect Z|\emptyset,N) \\
        A\in \vect Z
        }}w_i \geq 1 && \forall A\in \vect X \setminus \vect{Y}  \nonumber\\
        & \sum_{\substack{
        c_i = (\vect Z|\emptyset,N) \\
        B\in \vect Z
        }} w_i + w^B\geq 1 && \forall B\in \vect{X} \cap \vect{Y} \label{eq:ignore}\\
        & w_{i}\geq 0 && \forall i \in [m] \nonumber \\
                & w^{B}\geq 0&& \forall B \in \vect{X} \cap \vect{Y} \nonumber
    \end{align}

The above program simplifies to the linear program 
        determining $\rho^*_{Q_{\leaves(V(\vect X),T)}}(\vect{X} \setminus \vect{Y})$ due to the following two observations. Firstly, 
                 each constraint in $\hat{\calC}$ is of the form  $(\vect{Z}|\emptyset, N)$    or of the form
    $(B|\emptyset,1)$ with $B \in \vect{X} \cap \vect{Y}$.
       % (b) there is a one-to-one mapping between
       % the constraints of the from $(\vect{Z}|\emptyset, N)$
      %  and the atoms of the form  $S(\vect{Z})$, which implies a one-to-one mapping between the weight vectors of the two linear programs; 
    Secondly, the weights $w^B$
    associated with constraints of the form $c^B = (B|\emptyset, 1)$ %with $Y \in \vect{X} \cap \vect{Y}$ 
        do not have any effect on the
        objective function of the above linear program, since they are multiplied with $0$ in the definition of the objective function.
        Hence, the constraints in Line~\eqref{eq:ignore} can be easily satisfied by setting such a weight $w^B$ to 1. This implies that constraints in Line~\eqref{eq:ignore} can be omitted.
 \end{proof}

\nop{
 In some cases, the maintenance width is strictly less than the dynamic width. Consider the 3-path query given by
\[Q_{3p}(X_1, X_2, X_3, X_4) = R_1(X_1,X_2) \cdot R_2(X_2,X_3) \cdot R_3(X_3, X_4)\]
It is shown in Appendix~\ref{sec:path_queries} that $\mw(Q_{3p}) = \frac{1}{2}\log(N)$. However, $\delta(Q_{3p}) = \log(N)$. This is because all view trees in $\mathcal{T}(Q)$ will contain a join view $V(\vect{X})$ formed by either (1) joining $R_1$ (or a projection of it) with $R_2$, (2) symmetrically joining $R_3$ (or a projection of it) with $R_2$, or (3) joining $R_1$ and $R_3$ (or projections of them) with $R_2$. In all three cases, $\{X_2, X_3\} \subseteq \vect X$. In case (1), $\delta R_1$ implies $X_3$ must be covered in the fractional edge cover LP. In cases (2) and (3), $\delta R_3$ implies $X_2$ must be covered. In all three cases, the sum of the weights must be at least 1, and so $\delta(Q_{3p}) \geq \log(N)$. Using View Tree 1 in Fig.~\ref{fig:3-path-view-trees} for maintenance achieves the update time $\delta(Q_{3p}) = \log(N)$, which is asymptotically larger than the update time given by the maintenance width.
}

\subsection{Fractional Hypertree Width and Submodular Width}
The maintenance width does not come with a simple syntactic check. Indeed, to find this width for a given query $Q$, one needs to iterate over all view trees and degree configurations for $Q$ and solve a linear program to cost the update of each view in a view tree triggered by an update to any input relation. Yet this is conceptually not different from well-established width measures, such as the fractional hypertree width~\cite{FHTW} or the submodular width~\cite{PANDA}. These widths are  defined by iterating over all hypertree decompositions~\cite{TreeDecompositions} of $Q$ and by solving a linear program to cost each bag of a hypertree decomposition. Furthermore, the submodular width also requires adaptive computation by data partitioning. There are two differences here: (i) The data partitioning for the submodular width is fine-grained as it yields (poly-logarithmically many in $N$) database parts of uniform degrees, whereas for the maintenance width it is coarse-grained as it yields (constantly many in $N$) database parts with either heavy or light degrees. (ii) The data partitioning for the maintenance width is on the input relations only, whereas for the submodular width it can also be on the intermediate results (so also on the materialized views in our setting). 

There are two aspects of our maintenance width which are distinct from the aforementioned widths. (i) Due to our dynamic setting, we need to consider the cost of view updates in addition to the cost of computing the view only. 
%A dynamic variant of the fractional hypertree width, which is designed to capture this complication, has been proposed in the literature~\cite{IVMeps:PODS:2020}. 
(ii) We do not know the concrete cost of each view update until we fix the threshold $\epsilon$, which can only be done after constructing the function in $\epsilon$ that defines the width. This is novel to our setting.

\subsection{View Trees vs. Hypertree Decompositions}
Each view tree of a query $Q$ can be mapped to a hypertree decomposition of $Q$ (possibly with redundant bags), where each view (relation) becomes a bag consisting of the view variables. Also, from each hypertree decomposition we can construct a view tree, with one view for each bag of the decomposition, possibly additional projection views, and one leaf per relation in $Q$. So both view trees and hypertree decompositions allow us to explore the same space of structural decompositions of $Q$, albeit the view trees are more refined in that they use redundant information in the form of materialized projection views to allow for a faster propagation of updates in the view trees.

\section{Major and Minor Rebalancing of Data Partitioning}
\label{sec:rebalancing}
In the previous sections we showed that for any join query $Q$, the view trees constructed by our approach can be maintained in \(\bigO(N^{\mw(Q)})\) time under {\em one} single-tuple update, where $\mw(Q)$ denotes the maintenance width of $Q$. 
In this section, we extend this analysis to {\em sequences} of single-tuple updates.
We show that, given a {\em sequence} of single-tuple updates, the {\em amortized} single-tuple update time remains 
\(\bigO(N^{\mw(Q)})\). 
Our proof is based on an adaptation of the amortization technique previously developed for the triangle query~\cite{TriangleQuery}. For clarity, we outline the differences from that prior work and  state the central ideas of the argument. 

Each update may affect both the size of the database and the degrees of data values. Whenever the database size exceeds specified bounds, we recompute the view trees for all degree configurations by taking the new database size into account. We refer to this operation as {\em major rebalancing}.
Similarly, if a light value becomes heavy, or vice-versa, as a result of an update, we move the affected tuples between the corresponding relation fragments and update the view trees evaluated over these fragments. We call this operation {\em minor rebalancing}. The cost of both types of rebalancing can be amortized over the number of updates between two consecutive rebalancing steps, yielding an amortized rebalancing cost of \(\bigO(N^{\mw(Q)})\) per single-tuple update.

One difference between the maintenance strategy for the triangle query in prior work~\cite{TriangleQuery} and the approach proposed here is the definition of light and heavy values. In the prior work, lightness and heaviness were defined with respect to each individual relation; that is, an 
$X$-value $x$ is considered light in a relation 
$R$ if $|\sigma_{X = x}R| \leq N^\epsilon$, and heavy otherwise. In contrast, in this work we define lightness and heaviness globally, with respect to \emph{all} relations in the database (see Sec.~\ref{sec:prelims}). This distinction does not change the analysis of the amortized rebalancing time. The detailed analysis of the major and minor rebalancing steps is given in the extended version of this article.

\section{Constant-Delay Enumeration of the Query Output}
\label{sec:enumeration}

For any join query, our approach constructs  a set of view trees and maintains them under single-tuple updates to the input relations. As discussed in this section, we can also enumerate the query output from these view trees with constant delay.

Consider a join query $Q(\vect{X}) = R_1(\vect{X}_1) \cdot \ldots \cdot R_k(\vect{X}_k)$ with degree configurations $\mathcal D (Q)$. 
Let $\out$ denote the set of output tuples of $Q$, and $\out_{\vect d}$ the set of output tuples of $Q$ under the degree configuration $\vect d \in \mathcal D(Q)$. 
\nop{Each degree configuration defines a database fragment such that: (i) the union of all fragments makes up the original database; and (ii) the query outputs $\out_{\vect d}$ for $\vect d\in\mathcal D(Q)$ are disjoint and their union forms the entire output $\out$ of the query.}
The set $\out$ is the disjoint union of the sets $\out_{\vect d}$.
Given a constant-delay enumeration procedure for each tuple set $\out_{\vect d}$, we can therefore enumerate the tuples in $\out$  with constant delay by invoking the procedures one  after the other. 
For each tuple $\vect{t} \in \out$, its multiplicity is
$\sum_{i \in [k]}R_i(\vect{t}.\vect{X}_i)$,
which can be computed in constant time.

We next discuss such a constant-delay enumeration procedure for a single set $\out_{\vect d}$, for any $\vect d \in \mathcal D(Q)$. Recall that for each degree configuration $\vect d \in \mathcal D(Q)$, our approach maintains a view tree $T^{\vect d}$. 
It then suffices to enumerate $\out_{\vect d}$ from $T^{\vect d}$ with constant delay.  Prior work shows how from any view tree, the tuples in the join of the views in the tree can be enumerated with constant delay~\cite[Prop. 11]{dynamicrelation}. The extended version of this article gives further details, illustrates the enumeration procedure for the 4-cycle query, and explains how to slightly change the approach to maintain the count version of the query with the same update time and constant-delay enumeration.
\section{Conclusion}
\label{sec:conclusion}

In this paper we introduced an approach to adaptive maintenance of join queries under database updates (inserts and deletes) that exploits the  constraints that come with partitioning the database using the heavy and light degrees of the values of the join variables. Our approach matches the best known update times stated in the literature, while also generalizing to arbitrary join queries. 

There are several promising directions of future work, we list below a few:
\begin{itemize}
    \item Our complexity results hold for join queries only, but they can be generalized in a standard way to arbitrary conjunctive queries by taking view trees that correspond to free-connex hypertree decompositions of such queries. This free-connex restriction is necessary to ensure constant delay enumeration of the query output. The machinery developed in this paper to upper bound the sizes and compute times for delta views already works for arbitrary conjunctive queries.
    
    \item We can extend our approach to richer constraints beyond degree constraints: Using $\ell_p$-norms on the degree sequences of join columns, we can obtain tighter upper bounds on the size of the query output and on the runtime to compute it~\cite{LpBound:PODS:2024,LpBound:SIGMODRec:2025}. Adopting such statistics to the dynamic setting requires to maintain them efficiently under updates.

    \item We currently partition the data on each join variable. It remains open whether partitioning on tuples of variables and subsets of such tuples can improve the update time. For instance, for a query with body $R(A,B,C)\cdot S(A,B,D) \cdot T(A,E)\cdot U(B,F)$ we currently only partition on $A$ and $B$ separately, although we could also partition on the tuple $A,B$. Our framework extends immediately to this more general setting.

    \item We can lower the update time by using fast matrix multiplication: The $\bigO(N^{2/3})$ barrier for the update time of the 4-cycle query can be broken using a non-combinatorial IVM approach~\cite{FMM4Cycle}. It is unclear however how to generalize this non-combinatorial approach to arbitrary queries.
    
    \item When using a maintenance approach based on single view tree, inserts-only updates can require a lower update time than the general case with both inserts and deletes~\cite{InsertsVSDeletes}. It is open whether this restriction can also lower the update time in our approach.
\end{itemize}

\bibliographystyle{plainurl}
\bibliography{bibliography}

%\newpage

\appendix
\section{Additional Examples}
In this section, we provide additional (and extended) examples of our maintenance approach and show how it can recover all results of existing IVM approaches (where the update time is amortized, and the enumeration delay is constant). Conditional on the OMv conjecture~\cite{OMV:STOC:2015,QHierarchical} or on the conjectured optimality of the submodular width for static query evaluation~\cite{InsertsVSDeletes}, the update times shown here cannot be improved for any query by a polynomial factor, except for the update time of the bow tie query, which exhibits a gap of $\bigO(N^{1/4})$. We also highlight the adaptability of our approach, which can be used to maintain \emph{any} join query. To this end, we demonstrate our approach on query patterns that have not been specifically considered in prior IVM works and provide results on their update times.

\subsection{Extended Example: 4-Cycle Query}

In Sec.~\ref{sec:overview}, we claimed that our approach maintains the 4-cycle query
\[
    Q_{\square}(A,B,C,D) = R(A,B) \cdot S(B,C) \cdot T(C,D) \cdot U(D,A)
\]
with amortized update time $\bigO(N^{2/3})$. The update time cannot be improved by a polynomial factor, conditional on the conjectured optimality of the submodular width for static query evaluation~\cite{InsertsVSDeletes}. In this section, we explain how we obtain the update times given in Table~\ref{4-cycle-update-times}.

As $Q_\square$ has four join variables, there are 16 degree configurations. For each degree configuration, we pick an optimal view tree to use for maintenance. Fig.~\ref{fig:4-cycle-view-trees} shows the six view trees we use and Table~\ref{4-cycle-update-times} shows the view tree chosen for each degree configuration and the update time as a function of $\epsilon$. The update time is the maximum compute time of any delta view in the delta view tree for an update to any relation. We illustrate in detail how to find the update times of the three degree configurations below:
\[\{(L,L,L,L), (L,L,H,H), (L,H,H,H)\}\]
Other degree configurations witness identical update times that are derived similarly.

{\em Note: In our analysis we do not explicitly state the compute time of views corresponding to marginalizations and intersections (except the first time), as their compute times are inherited from a child view whose compute time has already been shown. }

\subsubsection{(L, L, L, L)}
For the degree configuration $\vect d = (L,L,L,L)$ and View Tree 1 in Fig.~\ref{fig:4-cycle-view-trees}, we simulate an update to each of the four relations $R$, $S$, $T$, and $U$, find the compute time of each delta view, and take the maximum. We begin with $\delta R$. The first delta view to be evaluated is $\delta V_1 (A,B,C) = \delta R(A,B) \cdot S(B,C)$. The constraints that holds at the leaves of $\delta V_1$ in the delta view tree are:
\[\C_1 = \DC(\leaves(\delta V_1), \vect d) = \{(BC|\emptyset,N), (BC|B,N^\epsilon), (BC|C, N^\epsilon), (A|\emptyset,1), (B|\emptyset,1)\}.\]
As the free variables of $\delta V_1$ are precisely $\vars(\C_1)$, $\C_1$ is the only constraint set to consider when bounding the compute time of $\delta V_1$, and the $\C_1$-guarding query is $Q_{\C_1} = \delta R(A,B) \cdot S(B,C) = \delta V_1$. $\bound(\C_1) = \epsilon$ using the second, fourth, and fifth constraints, which yields the compute time $\bigO(N^\epsilon)$ for $\delta V_1$. The next delta view to be evaluated is $\delta V_3(A, C) = \sum_B \delta V_1(A, B,C)$. The constraint set that holds at the leaves of $\delta V_2$ is $\C_1$. We consider $\C_1$ projected onto different supersets of the free variables of $\delta V_3$, namely $\C_1[AC] = \{(C|\emptyset,N), (A|\emptyset,1)\}$ and $\C_1[ABC] = \C_1$. The join query $Q_{\C_1}$ that over-approximates $\delta V_3$ gives the best upper bound of $\bigO(N^{\epsilon})$ on the compute time of $\delta V_3$. Next, we have $\delta V_5(A, C) = \delta V_3(A,C) \cdot V_4(C,A)$. The constraints that holds at the leaves of $\delta V_5$ are:
\begin{align*}
    \C_5 = \DC(\leaves(\delta V_5), \vect d) = \{&(BC|\emptyset,N), (CD|\emptyset,N), (DA|\emptyset,N), (BC|B,N^\epsilon), (BC|C,N^\epsilon),\\
    &(CD|C, N^\epsilon), (CD|D,N^\epsilon), (DA|A,N^\epsilon) (DA|D,N^\epsilon),\\
    &(A|\emptyset,1), (B|\emptyset,1)\}
\end{align*}
We consider the constraint sets $\C_5[AC]$, $\C_5[ABC]$, $\C_5[ACD]$, and $\C_5[ABCD] = \C_5$. The constraint set
\[\C_5[ABC] = \{(BC|\emptyset,N), (C|\emptyset,N),(A|\emptyset,N),(BC|B,N^\epsilon),(BC|C,N^\epsilon),(A,\emptyset,1), (B|\emptyset,1)\}\]
has the guarding query $Q_{\C_5[ABC]} = \delta R(A,B) \cdot S(B,C) \cdot T'(C) \cdot U'(A)$, where $T'$ is the projection of $T$ onto $C$ and $U'$ is the projection of $U$ onto $A$. $Q_{\C_5[ABC]}$ over-approximates $\delta V_5$ and gives the best upper bound of $\bigO(N^\epsilon)$ on the compute time by using the fourth, sixth, and seventh constraints.

Thus, the update time for $\delta R$ is $\bigO(N^\epsilon)$. Because all join variables are light, updates to the remaining relations are symmetric in the respective delta view trees and obtain the same update time. Thus, the update time for the degree configuration $(L, L, L, L)$ and View Tree 1 is $\bigO(N^\epsilon)$.

\subsubsection{(L, L, H, H)}
Next, we show the update time for $(L, L, H, H)$ using View Tree 4 in Fig.~\ref{fig:4-cycle-view-trees}. Consider an update $\delta R$. The first delta view to be evaluated is $\delta V_1(A, B, C) = \delta R(A, B) \cdot S(B, C)$. The constraints that holds at the leaves of $\delta V_1$ are:
\[\C_1 = \DC(\leaves(\delta V_1), \vect d) = \{(BC|\emptyset,N), (BC|B,N^\epsilon), (C|\emptyset,N^{1-\epsilon}), (A|\emptyset,1), (B|\emptyset,1)\}\]
We can use the second, fourth, and fifth constraints to yield the compute time $\bigO(N^\epsilon)$, or we can use the third, fourth, and fifth constraints to yield the compute time $\bigO(N^{1 - \epsilon})$. Thus, the compute time for $\delta V_1$ is $\bigO(N^{\min(\epsilon,1-\epsilon)})$. Next, we have $\delta V_3(A, C, D) = \delta V_2(A, C) \cdot U(D, A)$. The constraints that hold at the leaves of $\delta V_3$ are:
\begin{align*}
    \C_3 = \DC(\leaves(\delta V_3), \vect d) = \{&(BC|\emptyset,N), (DA|\emptyset,N), (BC|B,N^\epsilon), (DA|A,N^\epsilon),\\
    & (C|\emptyset,N^{1-\epsilon}), (D|\emptyset,N^{1-\epsilon}), (A|\emptyset,1), (B|\emptyset,1)\}.
\end{align*}
We consider the constraint sets $\C_3[ACD]$ and $\C_3[ABCD] = \C_3$. The $\C_3$-guarding query that over-approximates $\delta V_3$ has compute time $\bigO(N^{2\epsilon})$ (selecting the third, fourth, seventh, and eighth constraints) or $\bigO(N^{2 - 2\epsilon})$ (selecting the fifth, sixth, seventh, and eighth constraints). The projection $\C_3[ACD]$ does not give a better compute time. Thus, the compute time of $\delta V_3$ is $\bigO(N^{\min(2\epsilon, 2-2\epsilon)})$.

Now we consider an update $\delta S$. The first delta view to be evaluated is $\delta V_1(A,B,C) = R(A,B) \cdot \delta S(B,C)$, and the constraints that hold at the leaves of $\delta V_1$ are:
\[\C_1 = \DC(\leaves(\delta V_1), \vect d) = (AB|\emptyset,N), (AB|A,N^\epsilon), (AB|B,N^\epsilon), (B|\emptyset,1), (C|\emptyset,1).\]
Selecting the second, fourth, and fifth constraints gives the compute time $\bigO(N^\epsilon)$ for $\delta V_1$. $\delta V_3(A, C, D) = \delta V_2(A, C) \cdot U(D, A)$ and the constraints that hold at the leaves of $\delta V_3$ are:
\begin{align*}
    \C_3 = \DC(\leaves(\delta V_3), \vect d) = \{&(AB|\emptyset,N), (DA|\emptyset,N), (AB|A,N^\epsilon), (AB|B,N^\epsilon),\\
    &(DA|A,N^\epsilon), (D|\emptyset,N^{1-\epsilon}), (B|\emptyset,1), (C|\emptyset,1)\}.
\end{align*}
We consider the constraint sets $\C_3[ACD]$ and $\C_3[ABCD] = \C_3$. The $\C_3$-guarding query that over-approximates $\delta V_3$ presents two ways to bound the compute time: $\bigO(N)$ (selecting the second, seventh, and eighth constraints), or $\bigO(N^{2\epsilon})$ (selecting the fourth, fifth, sixth, and eighth constraints). $\C_3[ACD]$ does not give a better compute time. Thus, the compute time of $\delta V_3$ is $\bigO(N^{\min(2\epsilon,1)})$, which is also update time for $\delta S$.

Now we consider an update $\delta U$. The first delta view to be evaluated is $\delta V_3$, and we have $\delta V_3(A, C, D) = V_2(A, C) \cdot \delta U(D, A)$. The constraints that hold at the leaves of $\delta V_3$ are:
\begin{align*}
    \C_3 = \DC(\leaves(\delta V_3), \vect d) = \{&(AB|\emptyset,N), (BC|\emptyset,N), (AB|A,N^\epsilon), (AB|B,N^\epsilon),\\
    &(BC|B,N^\epsilon), (C|\emptyset,N^{1-\epsilon}), (D|\emptyset,1), (A|\emptyset,1)\}.
\end{align*}
We consider the constraint sets
\[\C_3[ACD] = \{(A|\emptyset,N), (C|\emptyset,N), (C|\emptyset,N^{1-\epsilon}), (D|\emptyset,1), (A|\emptyset,1)\}\]
and $\C_3[ABCD] = \C_3$. The $\C_3[ACD]$-guarding query that over-approximates $\delta V_3$ has compute time $\bigO(N^{1-\epsilon})$ (selecting the third, fourth, and fifth constraints). The $\C_3$-guarding query has compute time $\bigO(N^{2\epsilon})$ (selecting the third, fifth, seventh, and eighth constraints).  Thus, the compute time of $\delta V_3$ is $\bigO(N^{\text{min}(1 - \epsilon, 2\epsilon)})$, which is also the overall update time for $\delta U$.

Finally, and update $\delta T$ has constant update time because
\begin{align*}
    \C_5[CD] = \DC(\leaves(\delta V_5), \vect d)[CD] = \{&(C|\emptyset,N), (D|\emptyset,N),(C|\emptyset,N^{1-\epsilon}),\\
    & (D|\emptyset,N^{1-\epsilon}), (C|\emptyset,1), (D|\emptyset,1)\}
\end{align*}
and the $\C_5[CD]$-guarding query that over-approximates $\delta V_5$ has compute time $\bigO(1)$ (selecting the fifth and sixth constraints).

Overall, the update time for degree configuration $(L, L, H, H)$ and view tree four is given by $O(N^{f(\epsilon)})$, where $f(\epsilon) = \text{max}(\text{min}(2\epsilon, 2-2\epsilon), \text{min}(2\epsilon, 1))$, as shown in 
Table~\ref{4-cycle-update-times}.

\subsubsection{(L, H, H, H)}
Finally, we show the update time for the degree configuration $(L, H, H, H)$ using View Tree 2 in Fig.~\ref{fig:4-cycle-view-trees}. 

Consider an update $\delta R$. We evaluate
$\delta V_1(A, B, D) = \delta R(A, B) \cdot U(D, A)$.
The constraints that hold at the leaves of $\delta V_1$ in the delta view tree are:
$$\C_1 = \DC(\leaves(\delta V_1), \vect d) = \{(DA|\emptyset,N), (DA|A,N^\epsilon), (D|\emptyset,N^{1-\epsilon}), (A|\emptyset,1), (B|\emptyset,1)\}.$$
The $\C_1$-guarding query is precisely $\delta V_1$, and the compute time is $\bigO(N^\epsilon)$ (selecting the second, fourth, and fifth constraints) or $\bigO(N^{1-\epsilon})$ (selecting the third, fourth, and fifth constraints). Thus, the compute time of $\delta V_1$ is $\bigO(N^{\min(\epsilon, 1-\epsilon)})$, which is also the overall update time for $\delta R$.

An update $\delta U$ is symmetric to $\delta R$ and yields the same update time. 

For an update $\delta S$, we evaluate $\delta V_2(B, C, D) = \delta S(B, C) \cdot T(C, D)$. The constraints that hold at the leaves of $\delta V_2$ in the delta view tree are:
\[\C_2 = \DC(\leaves(\delta V_2), \vect d) = \{(CD|\emptyset,N), (C|\emptyset,N^{1-\epsilon}), (D|\emptyset,N^{1-\epsilon}), (B|\emptyset,1), (C|\emptyset,1)\}.\]
The $\C_2$-guarding query is precisely $\delta V_2$, and the compute time is $\bigO(N^{1-\epsilon})$ (selecting the third, fourth, and fifth constraints), which is also the overall update time for $\delta S$. 

An update $\delta T$ is symmetric to $\delta S$ and achieves the same update time. After simplifying, we have that the update time for degree configuration $(L, H, H, H)$ using View Tree 2 is $\bigO(N^{1 - \epsilon})$.

\subsubsection{Maintenance Width}

We have shown how to derive all three unique update times shown in Table~$\ref{4-cycle-update-times}$. The update times of the remaining degree configurations can be obtained through similar analysis. The maintenance width $\mw(Q_\square)$ of the 4-cycle query is then the minimum over all expressions in the table for the base $N$ logarithm of the update time. This width can be computed as the minimization over $\epsilon$ of the maximum optimal solution of linear programs.
    \begin{align*}
    \mw(Q_\square) &=  \min_\epsilon \max(\epsilon, 1-\epsilon, \max(\min(2\epsilon, 2-2\epsilon), \min(2\epsilon, 1))) \\
    &=  \min_\epsilon \max(\epsilon, 1-\epsilon, \min(2\epsilon, 2-2\epsilon), \min(2\epsilon, 1)) \\
    &\overset{*}{=} \min_\epsilon \min(
    \max(\epsilon, 1-\epsilon,2\epsilon,2\epsilon), 
    \max(\epsilon, 1-\epsilon,2\epsilon,1),\\
    &\hspace*{4.75em}
    \max(\epsilon, 1-\epsilon,2-2\epsilon,2\epsilon), 
    \max(\epsilon, 1-\epsilon,2-2\epsilon,1))\\
    &\overset{+}{=} \min (\min_\epsilon \max(\epsilon, 1-\epsilon,2\epsilon,2\epsilon) \text{ s.t. } 0\leq \epsilon \leq 1\\
    &\hspace*{3em} \min_\epsilon\max(\epsilon, 1-\epsilon,2-2\epsilon,1) \text{ s.t. } 0\leq \epsilon \leq 1\\
    &\hspace*{3em} \min_\epsilon\max(\epsilon, 1-\epsilon,2-2\epsilon,2\epsilon) \text{ s.t. } 0\leq \epsilon \leq 1\\
    &\hspace*{3em} \min_\epsilon\max(\epsilon, 1-\epsilon,2-2\epsilon,1) \text{ s.t. } 0\leq \epsilon \leq 1 )
\end{align*}
The equality (*) holds due to the distributivity of $\max$ over $\min$, while the equality (+) is due to the commutativity of the two $\min$ functions. We then have to take the minimum of the optimal solutions of four optimization problems, which can be encoded as linear programs. We show the equivalent linear program for the first optimization problem above:
\begin{align*}
    &\min_\epsilon\ \ q \ \ \text{ s.t. }
    \hspace*{2em} q \geq \epsilon
    \hspace*{2em} q \geq 1-\epsilon
    \hspace*{2em} q \geq 2\epsilon
    \hspace*{2em} 0 \leq \epsilon \leq 1.    
\end{align*}
We can observe that the optimal solution is $2/3$ and obtained for $\epsilon=1/3$. The other optimization problems give larger solutions. The update time of $Q_\square$ is then $\bigO(N^{2/3})$.

\subsection{Loomis-Whitney Queries}
\emph{Loomis-Whitney queries} generalize the triangle query from a clique of $k = 3$ to higher degrees \cite{LoomisWhitneyQueries}. The Loomis-Whitney $k$ query of degree $k \geq 3$ (denoted LW-$k$) is defined as
\[Q(X_1, ..., X_k) = \prod_{i \in [k]} R_i(\{X_1, ..., X_k\} \setminus \{X_i\})\]
We show that our approach achieves update time $\bigO(N^{1/2})$ for the LW-4 query and then extend the argument to LW-$k$. These update times cannot be improved by a polynomial factor, conditional on the OMv conjecture~\cite{OMV:STOC:2015,QHierarchical} and on the conjectured optimality of the submodular width for static query evaluation~\cite{InsertsVSDeletes}.

\subsubsection{Loomis-Whitney 4}
LW-4 is given by
\[Q_{\text{LW-}4}(A,B,C,D) = R(B,C,D) \cdot S(A,C,D) \cdot T(A,B,D) \cdot U(A,B,C)\]
The four view trees used to maintain $Q_{\text{LW-}4}$ are shown in Fig.~\ref{fig:lw-4-trees}. Table~\ref{tab:comparisons-lw} shows each degree configuration, the view tree used to maintain it, and the update time. The maintenance width $\mw(Q_{\text{LW-}4})$ is then the minimum over all expressions in the table for the base $N$ logarithm of the update time. This width can be computed as the minimization of optimal solutions of linear programs.
    \begin{align*}
    \\mw(Q_{\text{LW-}4}) &=  \min_\epsilon \max(\epsilon, 1-\epsilon, \min(\epsilon, 1 - \epsilon)) \\
    &=  \min_\epsilon \max(\epsilon, 1-\epsilon)
\end{align*}
We can observe that the optimal solution is $1/2$ and obtained for $\epsilon=1/2$. By Theorem~\ref{th:main}, the update time of $Q_{\text{LW-}4}$ is $\bigO(N^{1/2})$.

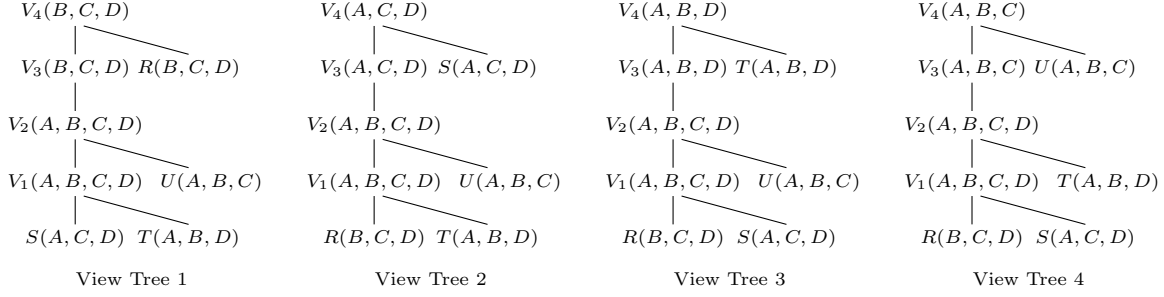
\begin{figure}
    \centering
    \setlength{\extrarowheight}{4pt}
    \setlength{\tabcolsep}{8pt} % default is 6pt
    \begin{tabular}{cccc}
        \begin{tikzpicture}[
            scale=0.75,
            every node/.style={
            circle,
            draw=none,
            minimum size=2.5em,
            inner sep=0pt,
            font=\scriptsize
            },
            every edge/.style={>=Stealth, line width=0.7, draw=black},
            ]
        
            \node at (0,0) {$S(A,C,D)$};
            \node at (2,0) {$T(A,B,D)$};
            
            \node at (0,1) {$V_1(A,B,C,D)$}; 
            \node at (2.4,1) {$U(A,B,C)$};
               
            \node at (0,2) {$V_2(A,B,C,D)$};
            
            \node at (0,3) {$V_3(B,C,D)$};
            \node at (2,3) {$R(B,C,D)$};

            \node at (0,4) {$V_4(B,C,D)$};
    
            \draw (0,0.25) -- (0,0.75);
            \draw (2,0.25) -- (0.15,0.75);
    
            \draw (0,1.25) -- (0,1.75);
            \draw (2,1.25) -- (0.15,1.75);

            \draw (0,2.25) -- (0,2.75);

            \draw (0,3.25) -- (0,3.75);
            \draw (2,3.25) -- (0.15,3.75);
    
            \node at (1,-0.7) {View Tree 1};
        \end{tikzpicture}
        &
        \begin{tikzpicture}[
            scale=0.75,
            every node/.style={
            circle,
            draw=none,
            minimum size=2.5em,
            inner sep=0pt,
            font=\scriptsize
            },
            every edge/.style={>=Stealth, line width=0.7, draw=black},
            ]
        
            \node at (0,0) {$R(B,C,D)$};
            \node at (2,0) {$T(A,B,D)$};
            
            \node at (0,1) {$V_1(A,B,C,D)$}; 
            \node at (2.4,1) {$U(A,B,C)$};
               
            \node at (0,2) {$V_2(A,B,C,D)$};
            
            \node at (0,3) {$V_3(A,C,D)$};
            \node at (2,3) {$S(A,C,D)$};

            \node at (0,4) {$V_4(A,C,D)$};
    
            \draw (0,0.25) -- (0,0.75);
            \draw (2,0.25) -- (0.15,0.75);
    
            \draw (0,1.25) -- (0,1.75);
            \draw (2,1.25) -- (0.15,1.75);

            \draw (0,2.25) -- (0,2.75);

            \draw (0,3.25) -- (0,3.75);
            \draw (2,3.25) -- (0.15,3.75);
    
            \node at (1,-0.7) {View Tree 2};
        \end{tikzpicture}
        &
        \begin{tikzpicture}[
            scale=0.75,
            every node/.style={
            circle,
            draw=none,
            minimum size=2.5em,
            inner sep=0pt,
            font=\scriptsize
            },
            every edge/.style={>=Stealth, line width=0.7, draw=black},
            ]
        
            \node at (0,0) {$R(B,C,D)$};
            \node at (2,0) {$S(A,C,D)$};
            
            \node at (0,1) {$V_1(A,B,C,D)$}; 
            \node at (2.4,1) {$U(A,B,C)$};
               
            \node at (0,2) {$V_2(A,B,C,D)$};
            
            \node at (0,3) {$V_3(A,B,D)$};
            \node at (2,3) {$T(A,B,D)$};

            \node at (0,4) {$V_4(A,B,D)$};
    
            \draw (0,0.25) -- (0,0.75);
            \draw (2,0.25) -- (0.15,0.75);
    
            \draw (0,1.25) -- (0,1.75);
            \draw (2,1.25) -- (0.15,1.75);

            \draw (0,2.25) -- (0,2.75);

            \draw (0,3.25) -- (0,3.75);
            \draw (2,3.25) -- (0.15,3.75);
    
            \node at (1,-0.7) {View Tree 3};
        \end{tikzpicture}
        &
        \begin{tikzpicture}[
            scale=0.75,
            every node/.style={
            circle,
            draw=none,
            minimum size=2.5em,
            inner sep=0pt,
            font=\scriptsize
            },
            every edge/.style={>=Stealth, line width=0.7, draw=black},
            ]
        
            \node at (0,0) {$R(B,C,D)$};
            \node at (2,0) {$S(A,C,D)$};
            
            \node at (0,1) {$V_1(A,B,C,D)$}; 
            \node at (2.4,1) {$T(A,B,D)$};
               
            \node at (0,2) {$V_2(A,B,C,D)$};
            
            \node at (0,3) {$V_3(A,B,C)$};
            \node at (2,3) {$U(A,B,C)$};

            \node at (0,4) {$V_4(A,B,C)$};
    
            \draw (0,0.25) -- (0,0.75);
            \draw (2,0.25) -- (0.15,0.75);
    
            \draw (0,1.25) -- (0,1.75);
            \draw (2,1.25) -- (0.15,1.75);

            \draw (0,2.25) -- (0,2.75);

            \draw (0,3.25) -- (0,3.75);
            \draw (2,3.25) -- (0.15,3.75);
    
            \node at (1,-0.7) {View Tree 4};
        \end{tikzpicture}
    \end{tabular}
    \vspace*{-2em}
    % \Description{Four view trees for the Loomis-Whitney query.}
    \caption{The four view trees used to maintain the LW-4 query.}
    \label{fig:lw-4-trees}
\end{figure}

\begin{table}[t]
\renewcommand{\arraystretch}{1.2}
\centering
\begin{tabular}{|c|c|r||c|c|r|}
\hline
\textbf{Configuration} & \textbf{View} & \textbf{\multirow{2}{*}{$\log_N$ UpdateTime}} & \textbf{Configuration} & \textbf{View} & \textbf{\multirow{2}{*}{$\log_N$ UpdateTime}} \\
$(A, B, C, D)$ & \textbf{Tree} & & $(A, B, C, D)$ & \textbf{Tree} & \\
\hline
$L, L, *, *$ & $1$ & $\epsilon$ & $H, L, H, L$ & $2$ & $\epsilon$ \\
\hline
$L, H, L, *$ & $1$ & $\epsilon$ & $H, L, H, H$ & $2$ & $f(\epsilon)$  \\
\hline
$L,H,H,L$ & $1$ & $\epsilon$ & $H, H, L, L$ & $3$ & $\epsilon$ \\
\hline
$L, H, H, H$ & $1$ & $f(\epsilon)$ & $H, H, L, H$ & $3$ & $f(\epsilon)$ \\
\hline
$H,L,L,*$ & $2$ & $\epsilon$ & $H, H, H, L$ & $4$ & $f(\epsilon)$ \\
\hline
& & & $H, H, H, H$ & $4$ & $1-\epsilon$ \\
\hline
\end{tabular}
\caption{The (base $N$) logarithm of the update time for each degree configuration and a specific view tree from Fig.~\ref{fig:lw-4-trees}. Note that $f(\epsilon) = \min(\epsilon, 1 - \epsilon)$; $(*)$ in the degree configuration indicates that the join variable can be either heavy or light.}
\label{tab:comparisons-lw}
\end{table}

\subsubsection{Loomis-Whitney $k$}

In this section, we prove our approach cannot be improved by a polynomial factor for LW-$k$ queries conditional on the aforementioned conjectures.

\begin{proposition}\label{prop:lw-update-time}
    If $Q$ is a Loomis-Whitney $k$ query where $k \geq 3$, then $Q$ admits $\bigO(N^{1/2})$ (amortized) update time and (non-amortized) $\bigO(1)$ enumeration delay using our approach.
\end{proposition}

\begin{proof}
    Consider the Loomis-Whitney $k$ query $Q$, where $k \geq 3$, and set the partitioning threshold to $\epsilon = 1/2$. Then there are $2^k$ unique degree configurations which fall into two cases: Either the degree configuration indicates at least one light variable, or all variables are heavy.
    
    Consider the first case and wlog, let the light variable be $X_k$, which appears in the schema of all relations except $R_k$. We construct the view tree for this degree configuration by joining $R_1(X_2, ..., X_k)$ and $R_2(X_1, X_3, ..., X_k)$ to create the view $V_3(X_1, ..., X_k)$. The view $V_{i+1}(X_1, ..., X_k)$ is created by joining $V_{i}(X_1, ..., X_k)$ and $R_{i}(X_1, ..., X_{i-1}, X_{i+1}, ..., X_k)$ for all $i \in [3, k-1]$. We create $V_{k+1}$ by projecting $V_{k}$ onto $\vars(R_k)$, and we create $V_{k+2}$ by intersecting $V_{k+1}$ and $R_k$. This view tree is shown in Fig.~\ref{fig:lw-k-tree}. We now show the update time is $\bigO(N^{1/2})$. Consider an update $\delta R_1$. The first delta view to be computed is \[\delta V_3(X_1,...,X_k) = \delta R_1(X_2,...,X_k) \cdot R_2(X_1,X_3,...,X_k).\]
    The number of $X_1$-values for a given $X_k$-value is at most $N^{1/2}$, so the compute time is $\bigO(N^{1/2})$. Each $V_4, ..., V_k$ is simply the semi-join reduction of $V_{i-1}$ with $R_{i-1}$ and can be inductively shown to have the same compute time. $V_{k+1}$ is a projection and $V_{k+2}$ is an intersection and achieve the same compute time. An update to $R_2$ is symmetric and achieves the same update time. Now consider an update $\delta R_3$. 
    \[\delta V_4(X_1,...,X_k) = V_3(X_1,...,X_k) \cdot \delta R_3(X_1,X_2,X_4,...,X_k)\]
    and the number of $X_3$-values for a given $X_k$-value is at most $N^{1/2}$, and so $\delta V_4$ has compute time $\bigO(N^{1/2})$. As before, the semi-joins $V_5, ..., V_k$, the projection $V_{k+1}$, and the intersection $V_{k+2}$ have the same compute time. Updates to $R_4, ..., R_{k-1}$ are symmetric and have the same update time. Now consider an update $\delta R_k$. Because $V_{k+2}$ is simply an intersection of $V_{k+1}$ and a single tuple, the compute time is $\bigO(1)$. Thus, the update time is $\bigO(N^{1/2})$ for such degree configurations.

    Now consider the case where all variables are heavy. Construct the same view tree as in the first case, and choose $X_k$ arbitrarily. Consider an update $\delta R_1$. There are at most $N^{1/2}$ different $X_1$-values, so the compute time is $\bigO(N^{1/2})$. As before, the semi-joins $V_4, ..., V_k$, the projection $V_{k+1}$, and the intersection $V_{k+2}$ have the same compute time. Thus, the update time of $\delta R_1$ is $\bigO(N^{1/2})$. An update to $R_2$ is symmetric and achieves the same update time. Now consider an update $\delta R_3$. There are at most $N^{1/2}$ different $X_1$-values, and so $\delta V_4$ and has compute time $\bigO(N^{1/2})$. The remaining delta views have the same compute time. Updates to $R_4, ..., R_{k-1}$ are symmetric and have the same update time. Now consider an update $\delta R_k$. As before, $V_{k+2}$ is simply an intersection of $V_{k+1}$ and a single tuple, so the compute time is $\bigO(1)$. Thus, the update time is $\bigO(N^{1/2})$ for this degree configuration.

    Our approach recovers this result, and $Q$ admits update time $\bigO(N^{1/2})$.
\end{proof}

\begin{figure}
    \centering
    \begin{tikzpicture}[
        scale=0.75,
        every node/.style={
        circle,
        draw=none,
        minimum size=2.5em,
        inner sep=0pt,
        font=\scriptsize
        },
        every edge/.style={>=Stealth, line width=0.7, draw=black},
        ]
    
        \node at (0,0) {$R_1(\{X_1,...,X_k\} \setminus \{X_1\})$};
        \node at (4.6,0) {$R_2(\{X_1,...,X_k\} \setminus \{X_2\})$};
        
        \node at (0,1) {$V_3(X_1,...,X_k)$}; 
        \node at (4,1) {$R_3(\{X_1,...,X_k\} \setminus \{X_3\})$};
           
        \node at (0,2) {$V_4(X_1,...,X_k)$};

        \node at (0, 2.6) {$\vdots$};
        
        \node at (0,3) {$V_{k-1}(X_1,...,X_k)$};
        \node at (4.8,3) {$R_{k-1}(\{X_1,...,X_k\} \setminus \{X_{k-1}\})$};

        \node at (0,4) {$V_k(X_1,...,X_k)$};
        
        \node at (0,5) {$V_{k+1}(\{X_1,...,X_k\} \setminus \{X_{k}\})$};
        \node at (5,5) {$R_k(\{X_1,...,X_k\} \setminus \{X_k\})$};

        \node at (0,6) {$V_{k+2}(\{X_1,...,X_k\} \setminus \{X_{k}\})$};

        \draw (0,0.25) -- (0,0.75);
        \draw (2.5,0.25) -- (0.15,0.75);

        \draw (0,1.25) -- (0,1.75);
        \draw (2.5,1.25) -- (0.15,1.75);

        \draw (0,3.25) -- (0,3.75);
        \draw (2.5,3.25) -- (0.15,3.75);

        \draw (0,4.25) -- (0,4.75);

        \draw (0,5.25) -- (0,5.75);
        \draw (3,5.25) -- (0.15,5.75);
    \end{tikzpicture}
    \vspace*{-3.5em}
    % \Description{The view tree constructed in the proof of Proposition~\ref{prop:lw-update-time}.}
    \caption{The view tree constructed in the proof of Proposition~\ref{prop:lw-update-time}. LW-$k$ achieves optimal update time using only left-deep view trees of this structure, where we permute the input relations, and after each join, project on the variables which appear in subsequent joins.}
    \label{fig:lw-k-tree}
\end{figure}
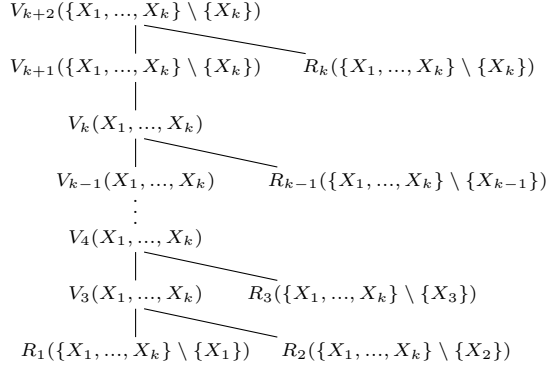

It's clear that the view trees used to maintain LW-4 (shown in Fig.~\ref{fig:lw-4-trees}) follow the same view tree construction given in the proof of Proposition~\ref{prop:lw-update-time} (shown in Fig.~\ref{fig:lw-k-tree}). That is, we use left-deep view trees in which we permute the input relations, and after each join, project out the variables which do not appear in subsequent joins.

\subsection{Hierarchical Queries}

Hierarchical queries are a sub-class of acyclic queries. A query is called \emph{hierarchical} if for any two variables $X$ and $Y$, it holds that $\at(X) \subseteq \at(Y)$, $\at(Y) \subseteq \at(X)$, or $\at(X) \cap \at(Y) = \emptyset$~\cite{ProbabilisticDatabases}. We show that our approach obtains the optimal amortized update time for the class of hierarchical join queries (so all variables are free).

\begin{proposition}\label{prop:hierarchical-update-time}
    If $Q$ is a hierarchical query, then $Q$ admits $\bigO(1)$ (amortized) update time and (non-amortized) $\bigO(1)$ enumeration delay using our approach.
\end{proposition}

\begin{proof}[Proof of Proposition~\ref{prop:hierarchical-update-time}]
    For any hierarchical query, there is a view tree that contains for each atom $R(\vect X)$ with $\vect X = \{X_1, \ldots, X_n \}$ a root-to-leaf path of the form
    \[V_1'(\vect X_1') \leftarrow V_1(\vect X_1) \leftarrow \ldots \leftarrow  V_n'(\vect X_n') \leftarrow V_n(\vect X_n) \leftarrow R(\vect X)\]
    such that for each $i \in [n]$, it holds: (1) $V_i(\vect X_i)$ is a join view such that all views and atoms containing $X_i$ are in the subtree rooted at $V_i(\vect X_i)$; (2) $V_i'(\vect X_i')$ is a projection view that projects away $X_i$ from $V_i(\vect X_i)$. Given the structural properties of hierarchical queries, this implies that any two sibling views in the view tree must be over the same set of variables. For any single-tuple update to relation $R$, each of the views $V_i$ and $V_i'$ on the path from the atom $R(\vect X)$ to the root of the view tree  can be updated as follows: if $V_i$ is a join view, we do constant-time look-ups in the child views of $V_i$ (corresponding to an intersection); if  $V_i'$ is a projection view, we project the update tuple onto the variables of $V_i'$, which takes only constant time. The constant-delay enumeration is a general property of the view trees considered in this paper.

    Our approach will explore the space of view trees and for each degree configuration will pick the above view tree as it has the smallest update time.
\end{proof}

\subsection{Path Queries}\label{sec:path_queries}
The \emph{$k$-path} query is defined as
\[Q(X_1, ..., X_{k+1}) = \prod_{i \in [k]} R_i(X_i, X_{i+1})\]
and so there are $k-1$ join variables. In this section, we illustrate our approach on the 3- and 4-path queries. The update time our approach achieves cannot be improved by a polynomial factor for either query, conditional on the OMv conjecture~\cite{OMV:STOC:2015,QHierarchical} and on the conjectured optimality of the submodular width for static query evaluation~\cite{InsertsVSDeletes}. 

\subsubsection{3-Path Query}
Consider the $3$-path query
\[Q_{\threepath}(X_1, X_2, X_3, X_4) = R_1(X_1, X_2) \cdot R_2(X_2, X_3) \cdot R_3(X_3, X_4)\]
Our approach uses two different view trees, shown in Fig.~\ref{fig:3-path-view-trees}. Table~\ref{tab:comparisons-3path} shows each degree configuration, the view tree used for maintenance, and the corresponding update time. The maintenance width $\mw(Q_{\threepath})$ is then the minimum over all expressions in the table for the base $N$ logarithm of the update time. This width can be computed as the minimization of optimal solutions of linear programs.
    \begin{align*}
    \mw(Q_{\threepath}) &=  \min_\epsilon \max(\epsilon, 1-\epsilon, \min(\epsilon, 1-\epsilon)) \\
    &=  \min_\epsilon \max(\epsilon, 1-\epsilon) \\
\end{align*}
We can observe that the optimal solution is $1/2$ and obtained for $\epsilon=1/2$. By Theorem~\ref{th:main}, the update time of $Q_{\threepath}$ is $\bigO(N^{1/2})$.

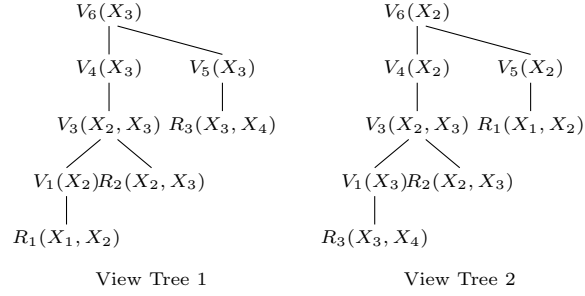
\begin{figure}
    \centering
    \setlength{\extrarowheight}{4pt}
    \setlength{\tabcolsep}{8pt} % default is 6pt
    \begin{tabular}{cc}
        % Figure 1
        \begin{tikzpicture}[
            scale=0.75,
            every node/.style={
            circle,
            draw=none,
            minimum size=2.5em,
            inner sep=0pt,
            font=\scriptsize
            },
            every edge/.style={>=Stealth, line width=0.7, draw=black},
            ]
        
            \node at (0,0) {$R_1(X_1,X_2)$};

            \node at (0,1) {$V_1(X_2)$};
            \node at (1.5,1) {$R_2(X_2,X_3)$};
               
            \node at (0.75,2) {$V_3(X_2,X_3)$};
            \node at (2.75,2) {$R_3(X_3,X_4)$};

            \node at (0.75,3) {$V_4(X_3)$};
            \node at (2.75,3) {$V_5(X_3)$};
    
            \node at (0.75,4) {$V_6(X_3)$};

            \draw (0,0.25) -- (0,0.75);
    
            \draw (0,1.25) -- (0.6,1.75);
            \draw (1.5,1.25) -- (0.9,1.75);

            \draw (0.75,2.25) -- (0.75,2.75);
            \draw (2.75,2.25) -- (2.75,2.75);
    
            \draw (0.75,3.25) -- (0.75,3.75);
            \draw (2.75,3.25) -- (0.9,3.75);
    
            \node at (1.5,-0.7) {View Tree 1};
        \end{tikzpicture}
        &
        % Figure 2
        \begin{tikzpicture}[
            scale=0.75,
            every node/.style={
            circle,
            draw=none,
            minimum size=2.5em,
            inner sep=0pt,
            font=\scriptsize
            },
            every edge/.style={>=Stealth, line width=0.7, draw=black},
            ]
        
            \node at (0,0) {$R_3(X_3,X_4)$};

            \node at (0,1) {$V_1(X_3)$};
            \node at (1.5,1) {$R_2(X_2,X_3)$};
               
            \node at (0.75,2) {$V_3(X_2,X_3)$};
            \node at (2.75,2) {$R_1(X_1,X_2)$};

            \node at (0.75,3) {$V_4(X_2)$};
            \node at (2.75,3) {$V_5(X_2)$};
    
            \node at (0.75,4) {$V_6(X_2)$};

            \draw (0,0.25) -- (0,0.75);
    
            \draw (0,1.25) -- (0.6,1.75);
            \draw (1.5,1.25) -- (0.9,1.75);

            \draw (0.75,2.25) -- (0.75,2.75);
            \draw (2.75,2.25) -- (2.75,2.75);
    
            \draw (0.75,3.25) -- (0.75,3.75);
            \draw (2.75,3.25) -- (0.9,3.75);
    
            \node at (1.5,-0.7) {View Tree 2};
        \end{tikzpicture}
    \end{tabular}
    \vspace*{-2em}
    \caption{The two view trees used to maintain the 3-path query.}
    \label{fig:3-path-view-trees}
    % \Description{Two view trees for the 3-path query.}
\end{figure}

\begin{table}[t]
\renewcommand{\arraystretch}{1.2}
\centering
\begin{tabular}{|c|c|r||c|c|r|}
\hline
\textbf{Configuration} & \textbf{View} & \textbf{\multirow{2}{*}{$\log_N$ UpdateTime}} & \textbf{Configuration} & \textbf{View} & \textbf{\multirow{2}{*}{$\log_N$ UpdateTime}} \\
$(X_2,X_3)$ & \textbf{Tree} & & $(X_2,X_3)$ & \textbf{Tree} & \\
\hline
$L, L$ & $1$ & $\epsilon$ & $H, L$ & $2$ & $\min(\epsilon, 1-\epsilon)$ \\
\hline
$L, H$ & $1$ & $\min(\epsilon,1-\epsilon)$ & $H, H$ & $1$ & $1-\epsilon$ \\
\hline
\end{tabular}
\caption{The (base $N$) logarithm of the update time for each degree configuration and a specific view tree from Fig.~\ref{fig:3-path-view-trees}. $(*)$ in the degree configuration indicates that the join variable can be either heavy or light.}
\label{tab:comparisons-3path}
\end{table}

\subsubsection{4-Path Query}
Consider the $4$-path query
\[Q_{\fourpath}(X_1, X_2, X_3, X_4, X_5) = R_1(X_1, X_2) \cdot R_2(X_2, X_3) \cdot R_3(X_3, X_4) \cdot R_4(X_4, X_5)\]
Our approach uses three different view trees, shown in Fig.~\ref{fig:4-path-view-trees}. Table \ref{tab:comparisons-4path} shows each degree configuration, the view tree used for maintenance, and the corresponding update time. The maintenance width $\mw(Q_{\fourpath})$ is then the minimum over all expressions in the table for the base $N$ logarithm of the update time. This width can be computed as the minimization of optimal solutions of linear programs.
    \begin{align*}
    \mw(Q_{\fourpath}) &=  \min_\epsilon \max(\epsilon, 1-\epsilon, \min(\epsilon, 1 - \epsilon), \max(\epsilon, \min(2\epsilon, 1 - \epsilon))) \\
    &=  \min_\epsilon \max(\epsilon, 1-\epsilon, \min(2\epsilon, 1 - \epsilon)) \\
    &\overset{*}{=} \min_\epsilon \min(\max(\epsilon, 1-\epsilon, 2\epsilon), \max(\epsilon, 1-\epsilon, 1 - \epsilon)) \\
    &\overset{+}{=} \min (\min_\epsilon \max(\epsilon, 1-\epsilon,2\epsilon) \text{ s.t. } 0\leq \epsilon \leq 1\\
    &\hspace*{3.25em} \min_\epsilon\max(\epsilon,1-\epsilon,1-\epsilon) \text{ s.t. } 0\leq \epsilon \leq 1)\\
\end{align*}
The equality (*) holds due to the distributivity of $\max$ over $\min$, while the equality (+) is due to the commutativity of the two $\min$ functions. We then have to take the minimum of the optimal solutions of two optimization problems, which can be encoded as linear programs. We can observe that the optimal solution is $1/2$ and obtained for $\epsilon=1/2$. By Theorem~\ref{th:main}, the update time of $Q_{\fourpath}$ is $\bigO(N^{1/2})$.
\begin{figure}
    \centering
    \setlength{\extrarowheight}{4pt}
    \setlength{\tabcolsep}{8pt} % default is 6pt
    \begin{tabular}{ccc}
        % Figure 1
        \begin{tikzpicture}[
            scale=0.75,
            every node/.style={
            circle,
            draw=none,
            minimum size=2.5em,
            inner sep=0pt,
            font=\scriptsize
            },
            every edge/.style={>=Stealth, line width=0.7, draw=black},
            ]
        
            \node at (-0.65,0) {$R_1(X_1,X_2)$}; 
            \node at (3,0) {$R_4(X_4,X_5)$};

            \node at (-0.65,1) {$V_1(X_2)$};
            \node at (1.2,1) {$R_2(X_2,X_3)$};
            \node at (4.8,1) {$R_3(X_3,X_4)$};
            \node at (3,1) {$V_2(X_4)$};
               
            \node at (0.75,2) {$V_3(X_2,X_3)$};
            \node at (3.75,2) {$V_4(X_3,X_4)$};

            \node at (0.75,3) {$V_5(X_3)$};
            \node at (3.75,3) {$V_6(X_3)$};
    
            \node at (2.25,4) {$V_7(X_3)$};

            \draw (-0.65,0.25) -- (-0.65,0.75);
            \draw (3,0.25) -- (3,0.75);
    
            \draw (-0.35,1.25) -- (0.6,1.75);
            \draw (1.5,1.25) -- (0.9,1.75); 
            \draw (3,1.25) -- (3.6,1.75);
            \draw (4.5,1.25) -- (3.9,1.75);

            \draw (0.75,2.25) -- (0.75,2.75);
            \draw (3.75,2.25) -- (3.75,2.75);
    
            \draw (0.75,3.25) -- (2.1,3.75);
            \draw (3.75,3.25) -- (2.4,3.75);
    
            \node at (2.25,-0.7) {View Tree 1};
        \end{tikzpicture}
        &
        % Figure 2
        \begin{tikzpicture}[
            scale=0.75,
            every node/.style={
            circle,
            draw=none,
            minimum size=2.5em,
            inner sep=0pt,
            font=\scriptsize
            },
            every edge/.style={>=Stealth, line width=0.7, draw=black},
            ]
        
            \node at (1.25,0) {$R_2(X_2,X_3)$};
            \node at (3.45,0) {$R_3(X_3,X_4)$};
            
            \node at (4.8,1) {$R_1(X_1,X_2)$}; 
            \node at (2.25,1) {$V_1(X_2,X_3,X_4)$}; 
               
            \node at (2.25,2) {$V_2(X_2,X_4)$};
            \node at (4.9,2) {$V_3(X_2)$};
    
            \node at (2.25,3) {$V_4(X_2,X_4)$};
            \node at (4.7,3) {$R_4(X_4,X_5)$};

            \node at (2.25,4) {$V_5(X_4)$};
            \node at (4.5,4) {$V_6(X_4)$};

            \node at (2.25,5) {$V_7(X_4)$};

            \draw (1.25,0.25) -- (2.1,0.75);
            \draw (3.25,0.25) -- (2.4,0.75);

            \draw (2.25,1.25) -- (2.25,1.75);
            \draw (4.5,1.25) -- (4.5,1.75);

            \draw (2.25,2.25) -- (2.25,2.75);
            \draw (4.5,2.25) -- (2.4,2.75);

            \draw (2.25,3.25) -- (2.25,3.75);
            \draw (4.5,3.25) -- (4.5,3.75);

            \draw (2.25,4.25) -- (2.25,4.75);
            \draw (4.5,4.25) -- (2.4,4.75);
    
            \node at (3.25,-0.7) {View Tree 2};
        \end{tikzpicture}
        &
        % Figure 3
        \begin{tikzpicture}[
            scale=0.75,
            every node/.style={
            circle,
            draw=none,
            minimum size=2.5em,
            inner sep=0pt,
            font=\scriptsize
            },
            every edge/.style={>=Stealth, line width=0.7, draw=black},
            ]
        
            \node at (1.25,0) {$R_2(X_2,X_3)$};
            \node at (3.45,0) {$R_3(X_3,X_4)$};
            
            \node at (4.8,1) {$R_4(X_4,X_5)$}; 
            \node at (2.25,1) {$V_1(X_2,X_3,X_4)$}; 
               
            \node at (2.25,2) {$V_2(X_2,X_4)$};
            \node at (4.8,2) {$V_3(X_4)$};
    
            \node at (2.25,3) {$V_4(X_2,X_4)$};
            \node at (4.5,3) {$R_1(X_1,X_2)$};

            \node at (2.25,4) {$V_5(X_2)$};
            \node at (4.5,4) {$V_6(X_2)$};

            \node at (2.25,5) {$V_7(X_2)$};

            \draw (1.25,0.25) -- (2.1,0.75);
            \draw (3.25,0.25) -- (2.4,0.75);

            \draw (2.25,1.25) -- (2.25,1.75);
            \draw (4.5,1.25) -- (4.5,1.75);

            \draw (2.25,2.25) -- (2.25,2.75);
            \draw (4.5,2.25) -- (2.4,2.75);

            \draw (2.25,3.25) -- (2.25,3.75);
            \draw (4.5,3.25) -- (4.5,3.75);

            \draw (2.25,4.25) -- (2.25,4.75);
            \draw (4.5,4.25) -- (2.4,4.75);
    
            \node at (3.25,-0.7) {View Tree 3};
        \end{tikzpicture}
    \end{tabular}
    \vspace*{-2em}
    \caption{The three view trees used to maintain the 4-path query.}
    \label{fig:4-path-view-trees}
    % \Description{Three view trees for the 4-path query.}
\end{figure}
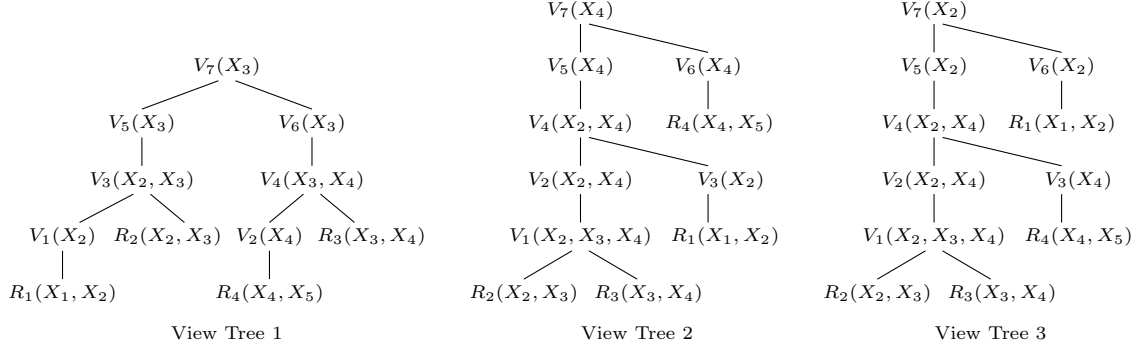

\begin{table}[t]
\renewcommand{\arraystretch}{1.2}
\centering
\begin{tabular}{|c|c|r||c|c|r|}
\hline
\textbf{Configuration} & \textbf{View} & \textbf{\multirow{2}{*}{$\log_N$ UpdateTime}} & \textbf{Configuration} & \textbf{View} & \textbf{\multirow{2}{*}{$\log_N$ UpdateTime}} \\
$(X_2, X_3, X_4)$ & \textbf{Tree} & & $(X_2, X_3, X_4)$ & \textbf{Tree} & \\
\hline
$L, L, L$ & $1$ & $\epsilon$ & $H, L, L$ & $3$ & $\max(\epsilon, \min(2\epsilon, 1 - \epsilon))$  \\
\hline
$L, L, H$ & $2$ & $\max(\epsilon, \min(2\epsilon, 1 - \epsilon))$ & $H, L, H$ & $2$ & $1 - \epsilon$ \\
\hline
$L, H, L$ & $1$ & $\min(\epsilon, 1 - \epsilon)$ & $H, H, *$ & $1$ & $1- \epsilon$ \\
\hline
$L, H, H$ & $1$ & $1 - \epsilon$ & & & \\
\hline
\end{tabular}
\caption{The (base $N$) logarithm of the update time for each degree configuration and a specific view tree from Fig.~\ref{fig:4-path-view-trees}. $(*)$ in the degree configuration indicates that the join variable can be either heavy or light.}
\label{tab:comparisons-4path}
\end{table}

\subsection{Bow Tie Query}
Prior works in IVM often focus on join queries with relatively few join variables. Our approach not only recovers these results, but can be used to maintain \emph{any} join query with arbitrarily many join variables. In this section, we demonstrate our approach on the \emph{bow tie} query, which admits update time $\bigO(N)$ and has five join variables. 
\[Q_{\bowtie} = R_1(X_1,X_2) \cdot R_2(X_2,X_3) \cdot R_3(X_3,X_1) \cdot R_4(X_3,X_4) \cdot R_5(X_4,X_5) \cdot R_6(X_5,X_3)\]
The bow tie can be visualized as two triangle queries that share a join variable, shown in Fig.~\ref{fig:query-structures}.

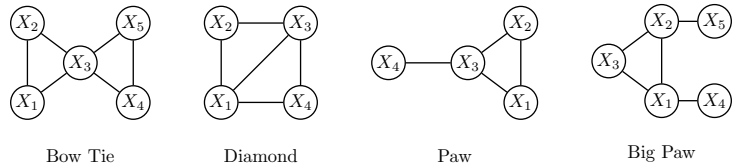
\begin{figure}
    \centering
    \setlength{\extrarowheight}{4pt}
    \setlength{\tabcolsep}{10pt} % default is 6pt
    \begin{tabular}{cccc}
    \begin{tikzpicture}[scale=0.7, transform shape, line width=0.5pt]

        % Coordinates
        \node[draw,circle,inner sep=1] (a) at (0,0) {$X_1$};
        \node[draw,circle,inner sep=1] (b) at (0,1.5) {$X_2$};
        \node[draw,circle,inner sep=1] (c) at (1,0.75) {$X_3$};
        \node[draw,circle,inner sep=1] (d) at (2,0) {$X_4$};
        \node[draw,circle,inner sep=1] (e) at (2,1.5) {$X_5$};
        
        % Left triangle
        \draw (a) -- (b) -- (c) -- (a);
        
        % Right triangle
        \draw (c) -- (d) -- (e) -- (c);

        \node at (1,-1) {Bow Tie};
    \end{tikzpicture}
    &
    \begin{tikzpicture}[scale=0.7, transform shape, line width=0.5pt]

        % Coordinates
        \node[draw,circle,inner sep=1] (a) at (0,0) {$X_1$};
        \node[draw,circle,inner sep=1] (b) at (0,1.5) {$X_2$};
        \node[draw,circle,inner sep=1] (c) at (1.5,1.5) {$X_3$};
        \node[draw,circle,inner sep=1] (d) at (1.5,0) {$X_4$};
        
        % Left triangle
        \draw (a) -- (b) -- (c) -- (d) -- (a);
        \draw (a) -- (c);

        \node at (0.75,-1) {Diamond};

    \end{tikzpicture}
    &
    \begin{tikzpicture}[scale=0.7, transform shape, line width=0.5pt]
    
        \node[draw,circle,inner sep=1] (c) at (1,0.75) {$X_3$};
        \node[draw,circle,inner sep=1] (d) at (2,0) {$X_1$};
        \node[draw,circle,inner sep=1] (e) at (2,1.5) {$X_2$};
        \node[draw,circle,inner sep=1] (f) at (-0.5,0.75) {$X_4$};
        
        \draw (c) -- (d) -- (e) -- (c) -- (f);

        \node[draw=none] at (0.75,-1) {Paw};
    \end{tikzpicture}
    &
    \begin{tikzpicture}[scale=0.7, transform shape, line width=0.5pt]
    
        \node[draw,circle,inner sep=1] (c) at (0,0.75) {$X_3$};
        \node[draw,circle,inner sep=1] (d) at (1,0) {$X_1$};
        \node[draw,circle,inner sep=1] (e) at (1,1.5) {$X_2$};
        \node[draw,circle,inner sep=1] (f) at (2,0.0) {$X_4$};
        \node[draw,circle,inner sep=1] (g) at (2,1.5) {$X_5$};
        
        \draw (c) -- (d) -- (e) -- (c);
        \draw (d) -- (f);
        \draw (e) -- (g);

        \node[draw=none] at (1,-1) {Big Paw};
    \end{tikzpicture}
    \end{tabular}
    % \Description{The structures of some of the less common queries described in the examples.}
    \caption{The structures of some of the less common queries described in the examples.}
    \label{fig:query-structures}
\end{figure}

Consider a degree configuration $(d_1, d_2, d_3, d_4, d_5)$ of $Q_{\bowtie}$. To construct the corresponding view tree $T$, we construct view trees $T_1$ and $T_2$ for the two triangle queries defined by vertices $X_1, X_2, X_3$ and $X_3, X_4, X_5$, which have degree configurations $(d_1, d_2, d_3)$ and $(d_3, d_4, d_5)$, respectively. When the root views of $T_1$ and $T_2$ have only $X_3$ in their schema, then given any update, we can intersect these two views to form $T$ and incur no extra cost. Their intersection is the root view of $T$. Then the update time of $Q_{\bowtie}$ for this degree configuration becomes the maximum of the update time of $T_1$ (given an update to $R_1$, $R_2$, or $R_3$) and $T_2$ (given an update to $R_4$, $R_5$, or $R_6$). The construction of $T$ for one degree configuration is shown in Fig.~\ref{fig:bowtie-tree-construction}.

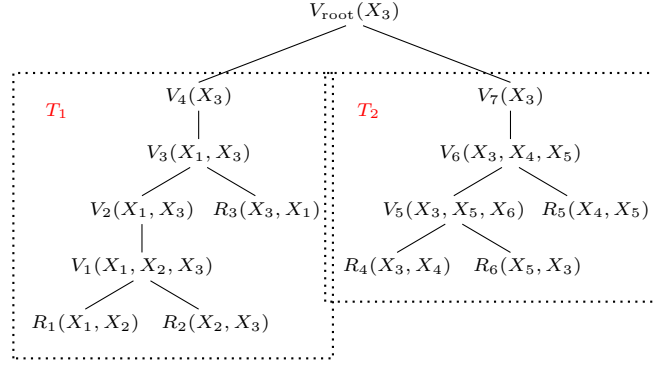
\begin{figure}
    \centering
    \begin{tikzpicture}[
        scale=0.75,
        every node/.style={
        circle,
        draw=none,
        minimum size=2.5em,
        inner sep=0pt,
        font=\scriptsize
        },
        every edge/.style={>=Stealth, line width=0.7, draw=black},
        ]
    
        \begin{scope}[local bounding box=Tone]

            \node at (0,0) {$R_1(X_1, X_2)$};
            \node at (2.3,0) {$R_2(X_2, X_3)$};
                
            \node at (1,1) {$V_1(X_1,X_2,X_3)$};
                   
            \node at (1,2) {$V_2(X_1,X_3)$};
            \node at (3.2,2) {$R_3(X_3,X_1)$};
        
            \node at (2,3) {$V_3(X_1,X_3)$};
        
            \node at (2,4) {$V_4(X_3)$};
        
            \draw (0,0.25) -- (0.85,0.75);
            \draw (2,0.25) -- (1.15,0.75);
        
            \draw (1,1.25) -- (1,1.75);
        
            \draw (1,2.25) -- (1.85,2.75);
            \draw (3,2.25) -- (2.15,2.75);
        
            \draw (2,3.25) -- (2,3.75);
        \end{scope}
        
        \draw[dotted, thick]
            ([xshift=-3mm,yshift=3.5mm]Tone.south west) rectangle
            ([xshift=2mm,yshift=-2mm]Tone.north east);
        
        \begin{scope}[xshift=5.5cm,yshift=1cm,local bounding box=Tone2]
            \node at (0,0) {$R_4(X_3, X_4)$};
            \node at (2.3,0) {$R_6(X_5, X_3)$};
            
            \node at (1,1) {$V_5(X_3,X_5,X_6)$};
            \node at (3.5,1) {$R_5(X_4,X_5)$};
    
            \node at (2,2) {$V_6(X_3,X_4,X_5)$};

            \node at (2,3) {$V_7(X_3)$};
    
            \draw (0,0.25) -- (0.85,0.75);
            \draw (2,0.25) -- (1.15,0.75);
    
            \draw (1,1.25) -- (1.85,1.75);
            \draw (3,1.25) -- (2.15,1.75);

            \draw (2,2.25) -- (2,2.75);
        \end{scope}

        \draw[dotted, thick]
            ([xshift=-3mm,yshift=3.5mm]Tone2.south west) rectangle
            ([xshift=2mm,yshift=-2mm]Tone2.north east);

        \node at (4.75,5.5) {$V_{\text{root}}(X_3)$};
        \draw (2,4.25) -- (4.6,5.25);
        \draw (7.5,4.25) -- (4.9,5.25);

        \node[color=red] at (-0.5,3.75) {$T_1$};
        \node[color=red] at (5,3.75) {$T_2$};
    \end{tikzpicture}
    \vspace*{-1.5em}
    % \Description{The construction for a bow tie view tree.}
    \caption{The view tree $T$ for the bow tie query with degree configuration $(H, L, L, H, H)$. The boxes indicates subtrees $T_1$ and $T_2$, which correspond to view trees 1 and 6, respectively, in fig.~\ref{fig:bowtie-view-trees}. $T_1$ and $T_2$ are view trees for the triangle query under degree configurations $(H, L, L)$ and $(L, H, H)$, respectively. This can be verified in Table~\ref{tab:comparisons-bowtie}.}
    \label{fig:bowtie-tree-construction}
\end{figure}

The triangle query admits $\bigO(N)$ update time when the root view must contain only $X_3$ in its schema. All eight degree configurations of this triangle query can be maintained using three view trees. In Fig.~\ref{fig:bowtie-view-trees}, we show these three view trees for both triangles of $Q_{\bowtie}$. $T$ is formed by intersecting the root view of $T_1$ (view tree 1, 2, or 3), with the root view of $T_2$ (view tree 4, 5, or 6). The choice of $T_1$ and $T_2$ for a given degree configuration is shown in Table~\ref{tab:comparisons-bowtie}. The maintenance width $\mw(Q_{\bowtie})$ is then the minimum over all expressions in the table for the base $N$ logarithm of the update time. This width can be computed as the minimization of optimal solutions of linear programs.
    \begin{align*}
    \mw(Q_{\bowtie}) &=  \min_\epsilon \max(\epsilon, 1-\epsilon, \min(\epsilon, 1 - \epsilon), 1) \\
    &=  \min_\epsilon \max(\epsilon, 1-\epsilon, 1)
\end{align*}
We can observe that any choice of $\epsilon$ results in the optimal solution $1$, so we pick $\epsilon = 0$. By Theorem~\ref{th:main}, the update time of $Q_{\bowtie}$ is $\bigO(N)$.

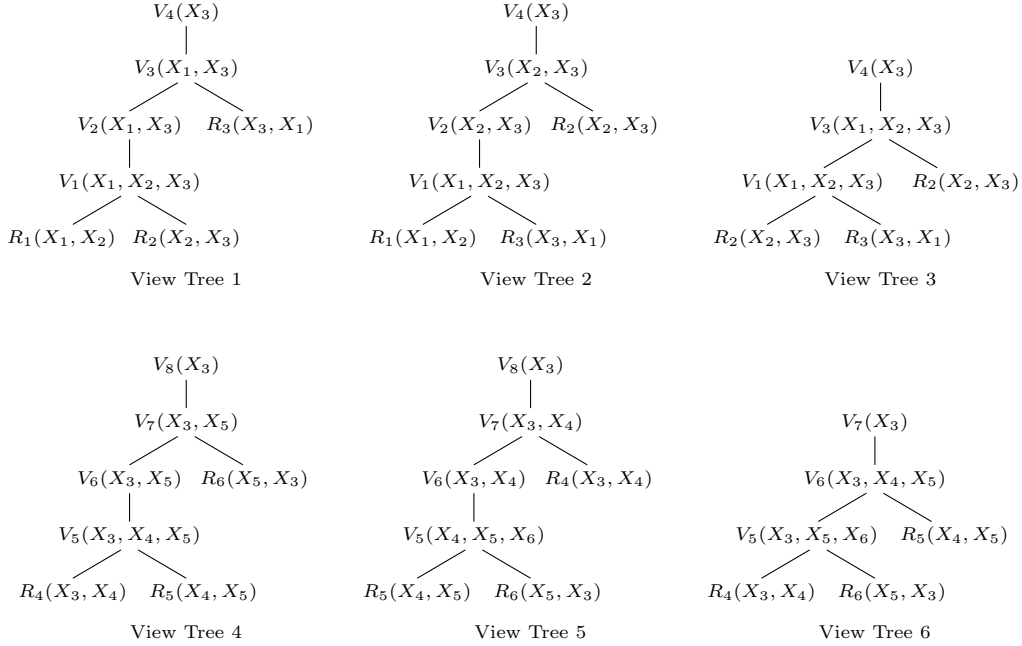
\begin{figure}
    \centering
    \setlength{\extrarowheight}{4pt}
    \setlength{\tabcolsep}{10pt} % default is 6pt
    \begin{tabular}{c}
    %------------------ Top row: Figures 1–3 ------------------%
    \begin{tabular}{ccc}
        % Figure 1
        \begin{tikzpicture}[
            scale=0.75,
            every node/.style={
            circle,
            draw=none,
            minimum size=2.5em,
            inner sep=0pt,
            font=\scriptsize
            },
            every edge/.style={>=Stealth, line width=0.7, draw=black},
            ]
        
            \node at (-0.2,0) {$R_1(X_1, X_2)$};
            \node at (2,0) {$R_2(X_2, X_3)$};
            
            \node at (1,1) {$V_1(X_1,X_2,X_3)$};
               
            \node at (1,2) {$V_2(X_1,X_3)$};
            \node at (3.3,2) {$R_3(X_3,X_1)$};
    
            \node at (2,3) {$V_3(X_1,X_3)$};

            \node at (2,4) {$V_4(X_3)$};
    
            \draw (0,0.25) -- (0.85,0.75);
            \draw (2,0.25) -- (1.15,0.75);
    
            \draw (1,1.25) -- (1,1.75);
    
            \draw (1,2.25) -- (1.85,2.75);
            \draw (3,2.25) -- (2.15,2.75);

            \draw (2,3.25) -- (2,3.75);
    
            \node at (2,-0.7) {View Tree 1};
        \end{tikzpicture}
        &
        % Figure 2
        \begin{tikzpicture}[
            scale=0.75,
            every node/.style={
            circle,
            draw=none,
            minimum size=2.5em,
            inner sep=0pt,
            font=\scriptsize
            },
            every edge/.style={>=Stealth, line width=0.7, draw=black},
            ]
        
            \node at (0,0) {$R_1(X_1, X_2)$};
            \node at (2.3,0) {$R_3(X_3, X_1)$};
            
            \node at (1,1) {$V_1(X_1,X_2,X_3)$};
               
            \node at (1,2) {$V_2(X_2,X_3)$};
            \node at (3.2,2) {$R_2(X_2,X_3)$};
    
            \node at (2,3) {$V_3(X_2,X_3)$};

            \node at (2,4) {$V_4(X_3)$};
    
            \draw (0,0.25) -- (0.85,0.75);
            \draw (2,0.25) -- (1.15,0.75);
    
            \draw (1,1.25) -- (1,1.75);
    
            \draw (1,2.25) -- (1.85,2.75);
            \draw (3,2.25) -- (2.15,2.75);

            \draw (2,3.25) -- (2,3.75);
    
            \node at (2,-0.7) {View Tree 2};
        \end{tikzpicture}
        &
        % Figure 3
        \begin{tikzpicture}[
            scale=0.75,
            every node/.style={
            circle,
            draw=none,
            minimum size=2.5em,
            inner sep=0pt,
            font=\scriptsize
            },
            every edge/.style={>=Stealth, line width=0.7, draw=black},
            ]
        
            \node at (0,0) {$R_2(X_2, X_3)$};
            \node at (2.3,0) {$R_3(X_3, X_1)$};
            
            \node at (0.8,1) {$V_1(X_1,X_2,X_3)$};
            \node at (3.5,1) {$R_2(X_2,X_3)$};
    
            \node at (2,2) {$V_3(X_1,X_2,X_3)$};

            \node at (2,3) {$V_4(X_3)$};
    
            \draw (0,0.25) -- (0.85,0.75);
            \draw (2,0.25) -- (1.15,0.75);
    
            \draw (1,1.25) -- (1.85,1.75);
            \draw (3,1.25) -- (2.15,1.75);

            \draw (2,2.25) -- (2,2.75);
    
            \node at (2,-0.7) {View Tree 3};
        \end{tikzpicture}
    \end{tabular}
    \\[-0.5em]
    %------------------ Bottom row: Figures 4–6 ------------------%
    \begin{tabular}{ccc}
        % Figure 4
        \begin{tikzpicture}[
            scale=0.75,
            every node/.style={
            circle,
            draw=none,
            minimum size=2.5em,
            inner sep=0pt,
            font=\scriptsize
            },
            every edge/.style={>=Stealth, line width=0.7, draw=black},
            ]
        
            \node at (0,0) {$R_4(X_3, X_4)$};
            \node at (2.3,0) {$R_5(X_4, X_5)$};
            
            \node at (1,1) {$V_5(X_3,X_4,X_5)$};
               
            \node at (1,2) {$V_6(X_3,X_5)$};
            \node at (3.2,2) {$R_6(X_5,X_3)$};
    
            \node at (2,3) {$V_7(X_3,X_5)$};

            \node at (2,4) {$V_8(X_3)$};
    
            \draw (0,0.25) -- (0.85,0.75);
            \draw (2,0.25) -- (1.15,0.75);
    
            \draw (1,1.25) -- (1,1.75);
    
            \draw (1,2.25) -- (1.85,2.75);
            \draw (3,2.25) -- (2.15,2.75);

            \draw (2,3.25) -- (2,3.75);
    
            \node at (2,-0.7) {View Tree 4};
        \end{tikzpicture}
        &
        % Figure 2
        \begin{tikzpicture}[
            scale=0.75,
            every node/.style={
            circle,
            draw=none,
            minimum size=2.5em,
            inner sep=0pt,
            font=\scriptsize
            },
            every edge/.style={>=Stealth, line width=0.7, draw=black},
            ]
        
            \node at (0,0) {$R_5(X_4, X_5)$};
            \node at (2.3,0) {$R_6(X_5, X_3)$};
            
            \node at (1,1) {$V_5(X_4,X_5,X_6)$};
               
            \node at (1,2) {$V_6(X_3,X_4)$};
            \node at (3.2,2) {$R_4(X_3,X_4)$};
    
            \node at (2,3) {$V_7(X_3,X_4)$};

            \node at (2,4) {$V_8(X_3)$};
    
            \draw (0,0.25) -- (0.85,0.75);
            \draw (2,0.25) -- (1.15,0.75);
    
            \draw (1,1.25) -- (1,1.75);
    
            \draw (1,2.25) -- (1.85,2.75);
            \draw (3,2.25) -- (2.15,2.75);

            \draw (2,3.25) -- (2,3.75);
    
            \node at (2,-0.7) {View Tree 5};
        \end{tikzpicture}
        &
        % Figure 6
        \begin{tikzpicture}[
            scale=0.75,
            every node/.style={
            circle,
            draw=none,
            minimum size=2.5em,
            inner sep=0pt,
            font=\scriptsize
            },
            every edge/.style={>=Stealth, line width=0.7, draw=black},
            ]
        
            \node at (0,0) {$R_4(X_3, X_4)$};
            \node at (2.3,0) {$R_6(X_5, X_3)$};
            
            \node at (0.8,1) {$V_5(X_3,X_5,X_6)$};
            \node at (3.4,1) {$R_5(X_4,X_5)$};
    
            \node at (2,2) {$V_6(X_3,X_4,X_5)$};

            \node at (2,3) {$V_7(X_3)$};
    
            \draw (0,0.25) -- (0.85,0.75);
            \draw (2,0.25) -- (1.15,0.75);
    
            \draw (1,1.25) -- (1.85,1.75);
            \draw (3,1.25) -- (2.15,1.75);

            \draw (2,2.25) -- (2,2.75);
    
            \node at (2,-0.7) {View Tree 6};
        \end{tikzpicture}
    \end{tabular}
    \end{tabular}
    \vspace*{-2em}
    \caption{Each view tree used to maintain the bow tie query can be constructed by creating a view $V_{\text{root}}(X_3)$ and attaching from this root two view trees: one from the set $\{1,2,3\}$ and one from the set $\{4, 5, 6\}$.}
    \label{fig:bowtie-view-trees}
    % \Description{Six view trees used for the construction of the each bow tie query view tree.}
\end{figure}

\begin{table}[t]
\renewcommand{\arraystretch}{1.2}
\centering
\begin{tabular}{|c|c|r||c|c|r|}
\hline
\textbf{Configuration} & \textbf{View} & \textbf{\multirow{2}{*}{$\log_N$ UpdateTime}} & \textbf{Configuration} & \textbf{View} & \textbf{\multirow{2}{*}{$\log_N$ UpdateTime}} \\
$(X_1, X_2, X_3)$ & \textbf{Tree $T_1$} & & $(X_3, X_4, X_5)$ & \textbf{Tree $T_2$} & \\
\hline
$L, L, *$ & $1$ & $\epsilon$ & $L, L, *$ & $4$ & $\epsilon$ \\
\hline
$L, H, L$ & $2$ & $\epsilon$ & $L, H, L$ & $5$ & $\epsilon$ \\
\hline
$L, H, H$ & $2$ & $f(\epsilon)$ & $L, H, H$ & $6$ & $1$ \\
\hline
$H, L, L$ & $1$ & $\epsilon$ & $H, L, L$ & $4$ & $\epsilon$  \\
\hline
$H, L, H$ & $1$ & $f(\epsilon)$ & $H, L, H$ & $4$ & $f(\epsilon)$  \\
\hline
$H, H, L$ & $3$ & $1$ & $H, H, L$ & $5$ & $f(\epsilon)$ \\
\hline
$H, H, H$ & $1$ & $1 - \epsilon$ & $H, H, H$ & $4$ & $1 - \epsilon$  \\
\hline
\end{tabular}
\caption{The (base $N$) logarithm of the update time for each degree configuration and a specific view tree from Fig.~\ref{fig:bowtie-view-trees}. Note that $f(\epsilon) = \min(\epsilon, 1 - \epsilon)$; $(*)$ in the degree configuration indicates that the join variable can be either heavy or light. The bow tie query update time under each degree configuration can be found by looking up the update time of each of its triangles and taking the maximum.}
\label{tab:comparisons-bowtie}
\end{table}

\subsection{Diamond Query}
Now we consider the diamond query
\[Q_{\diamondchord}= R_1(X_1, X_2) \cdot R_2(X_2,X_3) \cdot R_3(X_3,X_4) \cdot R_4(X_4,X_1) \cdot R_5(X_1,X_3)\]
which can be visualized as the 4-cycle with a single chord, shown in Fig.~\ref{fig:query-structures}. $Q_{\diamondchord}$ uses seven view trees shown in Fig.~\ref{fig:diamond-view-trees}. Table~\ref{tab:comparisons-diamond} shows for each degree configuration, the view tree used for maintenance, and the corresponding update time. The maintenance width $\mw(Q_{\diamondchord})$ is then the minimum over all expressions in the table for the base $N$ logarithm of the update time. This width can be computed as the minimization of optimal solutions of linear programs.
    \begin{align*}
    \mw(Q_{\diamondchord}) &=  \min_\epsilon \max(\epsilon, 2 \epsilon, 1 - \epsilon, \min(\epsilon, 1 - \epsilon), \max(\min(2\epsilon, 2-2\epsilon), \min(2\epsilon, 1)), \\
    &\hspace*{4.75em} \max(1-\epsilon,\min(2\epsilon,2-2\epsilon))) \\
    &=  \min_\epsilon \max(2\epsilon, 1-\epsilon, \min(2\epsilon, 2 - 2 \epsilon), \min(2\epsilon, 1)) \\
    &\overset{*}{=} \min_\epsilon \min(
    \max(2\epsilon, 1-\epsilon,2\epsilon,2\epsilon), 
    \max(2\epsilon, 1-\epsilon,2\epsilon,1),\\
    &\hspace*{4.75em}
    \max(2\epsilon, 1-\epsilon,2-2\epsilon,2\epsilon), 
    \max(2\epsilon, 1-\epsilon,2-2\epsilon,1))\\
    &\overset{+}{=} \min (\min_\epsilon \max(2\epsilon, 1-\epsilon,2\epsilon,2\epsilon) \text{ s.t. } 0\leq \epsilon \leq 1\\
    &\hspace*{3.25em} \min_\epsilon\max(2\epsilon, 1-\epsilon,2-2\epsilon,1) \text{ s.t. } 0\leq \epsilon \leq 1\\
    &\hspace*{3.25em} \min_\epsilon\max(2\epsilon, 1-\epsilon,2-2\epsilon,2\epsilon) \text{ s.t. } 0\leq \epsilon \leq 1\\
    &\hspace*{3.25em} \min_\epsilon\max(2\epsilon, 1-\epsilon,2-2\epsilon,1) \text{ s.t. } 0\leq \epsilon \leq 1 )
\end{align*}
The equality (*) holds due to the distributivity of $\max$ over $\min$, while the equality (+) is due to the commutativity of the two $\min$ functions. We then have to take the minimum of the optimal solutions of four optimization problems, which can be encoded as linear programs. The optimal solution is $2/3$ and obtained for $\epsilon=1/3$. By Theorem~\ref{th:main}, the update time of $Q_{\diamondchord}$ is $\bigO(N^{2/3})$. The update time cannot be improved by a polynomial factor, conditional on the conjectured optimality of the submodular width for static query evaluation~\cite{InsertsVSDeletes}.

\input{diamond-trees}

\textbf{\begin{table}[t]
\renewcommand{\arraystretch}{1.2}
\centering
\begin{tabular}{|c|c|r||c|c|r|}
\hline
\textbf{Configuration} & \textbf{View} & \textbf{\multirow{2}{*}{$\log_N$ UpdateTime}} & \textbf{Configuration} & \textbf{View} & \textbf{\multirow{2}{*}{$\log_N$ UpdateTime}} \\
$(X_1, X_2, X_3,X_4)$ & \textbf{Tree} & & $(X_1,X_2,X_3,X_4)$ & \textbf{Tree} & \\
\hline
$L, L, *, L$ & $1$ & $\epsilon$ & $H, L, L, L$ & $1$ & $\epsilon$ \\
\hline
$L, L, L, H$ & $4$ & $2\epsilon$ & $H, L, L, H$ & $4$ & $f(\epsilon) $ \\
\hline
$L, L, H, H$ & $5$ & $f(\epsilon)$ & $H, L, H, L$ & $1$ & $\min(\epsilon, 1 - \epsilon)$ \\
\hline
$L, H, L, L$ & $6$ & $2\epsilon$ & $H, L, H, H$ & $1$ & $1 - \epsilon$ \\
\hline
$L, H, L, H$ & $3$ & $f(\epsilon)$ & $H, H, L, L$ & $7$ & $f(\epsilon)$ \\
\hline
$L, H, H, L$ & $6$ & $f(\epsilon)$ & $H, H, L, H$ & $2$ & $g(\epsilon)$ \\
\hline
$L, H, H, H$ & $3$ & $g(\epsilon)$ & $H, H, H, *$ & $1$ & $1 - \epsilon$ \\
\hline
\end{tabular}
\caption{The (base $N$) logarithm of the update time for each degree configuration and a specific view tree from Fig.~\ref{fig:diamond-view-trees}. Note that $f(\epsilon) = \max(\min(2\epsilon, 2 - 2\epsilon), \min(2\epsilon,1))$; $g(\epsilon) = \max(1-\epsilon, \min(2\epsilon, 2 - 2\epsilon))$; $(*)$ in the degree configuration indicates that the join variable can be either heavy or light.}
\label{tab:comparisons-diamond}
\end{table}}

\subsection{Paw Query}

Consider the paw query 
\[Q_{\paw}(X_1,X_2,X_3,X_4) = R_1(X_1,X_2) \cdot R_2(X_2,X_3) \cdot R_3(X_3,X_1) \cdot R_4(X_3,X_4)\]
which is the triangle query with an additional edge to a new vertex, shown in Fig.~\ref{fig:query-structures}. $Q_{\paw}$ can be maintained using the three view trees shown in Fig.~\ref{fig:paw-view-trees}. Table~\ref{tab:comparisons-paw} shows for each degree configuration, the view tree used for maintenance and the corresponding update time. The maintenance width $\mw(Q_{\paw})$ is then the minimum over all expressions in the table for the base $N$ logarithm of the update time. This width can be computed as the minimization of optimal solutions of linear programs.
\begin{align*}
    \mw(Q_{\paw}) &=  \min_\epsilon \max(\epsilon, 1 - \epsilon, \min(\epsilon, 1 - \epsilon), \min(2\epsilon,2-2\epsilon)), \\
    &=  \min_\epsilon \max(\epsilon, 1 - \epsilon, \min(2\epsilon, 2 - 2\epsilon)) \\
    &\overset{*}{=} \min_\epsilon \min(\max(\epsilon,1-\epsilon,2\epsilon), \max(\epsilon, 1-\epsilon, 2-2\epsilon))\\
    &\overset{+}{=} \min (\min_\epsilon \max(\epsilon, 1 - \epsilon, 2\epsilon) \text{ s.t. } 0\leq \epsilon \leq 1,\\
    &\hspace*{3.25em} \min_\epsilon\max(\epsilon, 1 - \epsilon, 2-2\epsilon) \text{ s.t. } 0\leq \epsilon \leq 1 )
\end{align*}
The equality (*) holds due to the distributivity of $\max$ over $\min$, while the equality (+) is due to the commutativity of the two $\min$ functions. We then have to take the minimum of the optimal solutions of two optimization problems, which can be encoded as linear programs. The optimal solution is $2/3$ and obtained for $\epsilon=1/3$. By Theorem~\ref{th:main}, the update time of $Q_{\paw}$ is $\bigO(N^{2/3})$. The update time cannot be improved by a polynomial factor, conditioned on the conjectured optimality of the submodular width for static query evaluation~\cite{InsertsVSDeletes}.

\begin{figure}
    \centering
    \setlength{\extrarowheight}{4pt}
    \setlength{\tabcolsep}{8pt} % default is 6pt
    \begin{tabular}{ccc}
        % Figure 1
        \begin{tikzpicture}[
            scale=0.75,
            every node/.style={
            circle,
            draw=none,
            minimum size=2.5em,
            inner sep=0pt,
            font=\scriptsize
            },
            every edge/.style={>=Stealth, line width=0.7, draw=black},
            ]
        
            \node at (-0.2,0) {$R_1(X_1,X_2)$}; 
            \node at (2,0) {$R_2(X_2,X_3)$};

            \node at (1,1) {$V_1(X_1,X_2,X_3)$};
               
            \node at (1,2) {$V_2(X_1,X_3)$};
            \node at (3.2,2) {$R_3(X_3,X_1)$};

            \node at (1,3) {$V_3(X_1,X_3)$};
            \node at (3.2,3) {$R_4(X_3,X_4)$};
    
            \node at (1,4) {$V_4(X_3)$};
            \node at (3,4) {$V_5(X_3)$};

            \node at (1,5) {$V_6(X_3)$};

            \draw (0,0.25) -- (0.85,0.75);
            \draw (2,0.25) -- (1.15,0.75);
    
            \draw (1,1.25) -- (1,1.75);

            \draw (1,2.25) -- (1,2.75);
            \draw (3,2.25) -- (1.15,2.75);
    
            \draw (1,3.25) -- (1,3.75);
            \draw (3,3.25) -- (3,3.75);

            \draw (1,4.25) -- (1,4.75);
            \draw (3,4.25) -- (1.15,4.75);
    
            \node at (1.5,-0.7) {View Tree 1};
        \end{tikzpicture}
        &
        % Figure 2
        \begin{tikzpicture}[
            scale=0.75,
            every node/.style={
            circle,
            draw=none,
            minimum size=2.5em,
            inner sep=0pt,
            font=\scriptsize
            },
            every edge/.style={>=Stealth, line width=0.7, draw=black},
            ]
        
            \node at (-0.2,0) {$R_1(X_1,X_2)$}; 
            \node at (2.1,0) {$R_3(X_3,X_1)$};

            \node at (1,1) {$V_1(X_1,X_2,X_3)$};
               
            \node at (1,2) {$V_2(X_2,X_3)$};
            \node at (3.1,2) {$R_2(X_2,X_3)$};

            \node at (1,3) {$V_3(X_2,X_3)$};
            \node at (3.2,3) {$R_4(X_3,X_4)$};
    
            \node at (1,4) {$V_4(X_3)$};
            \node at (3,4) {$V_5(X_3)$};

            \node at (1,5) {$V_6(X_3)$};

            \draw (0,0.25) -- (0.85,0.75);
            \draw (2,0.25) -- (1.15,0.75);
    
            \draw (1,1.25) -- (1,1.75);

            \draw (1,2.25) -- (1,2.75);
            \draw (3,2.25) -- (1.15,2.75);
    
            \draw (1,3.25) -- (1,3.75);
            \draw (3,3.25) -- (3,3.75);

            \draw (1,4.25) -- (1,4.75);
            \draw (3,4.25) -- (1.15,4.75);
    
            \node at (1.5,-0.7) {View Tree 2};
        \end{tikzpicture}
        &
        % Figure 3
        \begin{tikzpicture}[
            scale=0.75,
            every node/.style={
            circle,
            draw=none,
            minimum size=2.5em,
            inner sep=0pt,
            font=\scriptsize
            },
            every edge/.style={>=Stealth, line width=0.7, draw=black},
            ]
        
            \node at (0,0) {$R_2(X_2,X_3)$}; 
            \node at (2,0) {$R_3(X_3,X_1)$};
            \node at (4.2,0) {$R_4(X_3,X_4)$};

            \node at (1,1) {$V_1(X_1,X_2,X_3)$};
            \node at (3,1) {$V_2(X_3)$};
            
            \node at (1,2) {$V_3(X_1,X_2,X_3)$};
            
            \node at (1,3) {$V_4(X_1,X_2)$};
            \node at (3.1,3) {$R_1(X_1,X_2)$};
    
            \node at (1,4) {$V_5(X_1,X_3)$};

            \draw (0,0.25) -- (0.85,0.75);
            \draw (2,0.25) -- (1.15,0.75);
            \draw (4,0.25) -- (3.4,0.75);
    
            \draw (1,1.25) -- (1,1.75);
            \draw (3,1.25) -- (1.15,1.75);

            \draw (1,2.25) -- (1,2.75);
    
            \draw (1,3.25) -- (1,3.75);
            \draw (3,3.25) -- (1.15,3.75);
    
            \node at (2,-0.7) {View Tree 3};
        \end{tikzpicture}
    \end{tabular}
    \vspace*{-2em}
    \caption{The three view trees used to maintain the paw query.}
    \label{fig:paw-view-trees}
    % \Description{Three view trees for the paw query.}
\end{figure}
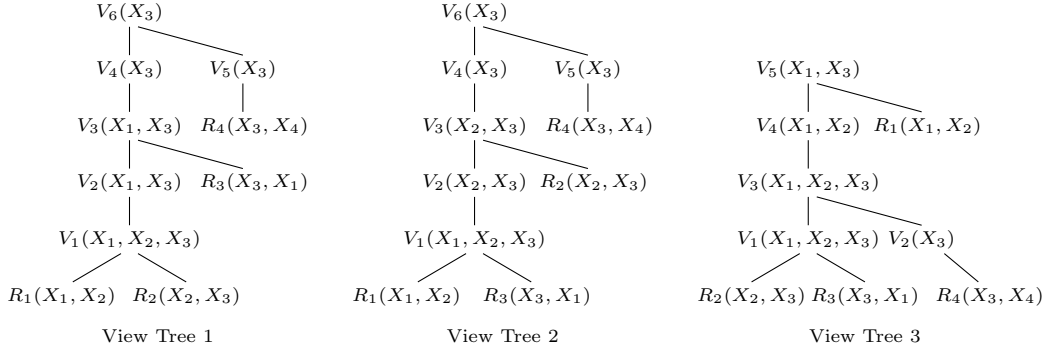

\begin{table}[t]
\renewcommand{\arraystretch}{1.2}
\centering
\begin{tabular}{|c|c|r||c|c|r|}
\hline
\textbf{Configuration} & \textbf{View} & \textbf{\multirow{2}{*}{$\log_N$ UpdateTime}} & \textbf{Configuration} & \textbf{View} & \textbf{\multirow{2}{*}{$\log_N$ UpdateTime}} \\
$(X_1, X_2, X_3)$ & \textbf{Tree} & & $(X_1, X_2, X_3)$ & \textbf{Tree} & \\
\hline
$L, L, *$ & $1$ & $\epsilon$ & $H, L, L$ & $1$ & $\epsilon$  \\
\hline
$L, H, L$ & $2$ & $\epsilon$ & $H, L, H$ & $1$ & $\min(\epsilon, 1 - \epsilon)$  \\
\hline
$L, H, H$ & $2$ & $\min(\epsilon, 1 - \epsilon)$ & $H, H, L$ & $3$ & $\min(2\epsilon, 2 - 2\epsilon)$ \\
\hline
& & & $H, H, H$ & $1$ & $1 - \epsilon$ \\
\hline
\end{tabular}
\caption{The (base $N$) logarithm of the update time for each degree configuration and a specific view tree from Fig.~\ref{fig:paw-view-trees}. $(*)$ in the degree configuration indicates that the join variable can be either heavy or light.}
\label{tab:comparisons-paw}
\end{table}

\subsection{Big Paw Query}

Consider the big paw query 
\[Q_{\bigpaw}(X_1,X_2,X_3,X_4,X_5) = R_1(X_1,X_2) \cdot R_2(X_2,X_3) \cdot R_3(X_3,X_1) \cdot R_4(X_1,X_4) \cdot R_5(X_2,X_5)\]
which is the triangle query with two additional edges, each to a new vertex, shown in Fig.~\ref{fig:query-structures}. $Q_{\bigpaw}$ can be maintained using the three view trees shown in Fig.~\ref{fig:big-paw-view-trees}. Table~\ref{tab:comparisons-big-paw} shows for each degree configuration, the view tree used for maintenance and the corresponding update time. The maintenance width $\mw(Q_{\bigpaw})$ is then the minimum over all expressions in the table for the base $N$ logarithm of the update time. This width can be computed as the minimization of optimal solutions of linear programs.
\begin{align*}
    \mw(Q_{\bigpaw}) &=  \min_\epsilon \max(2\epsilon, 1 - \epsilon, \min(2\epsilon, 1), \min(2\epsilon,2-2\epsilon)), \\
    &\overset{*}{=} \min_\epsilon \min(\max(2\epsilon,1-\epsilon,2\epsilon,2\epsilon), \max(2\epsilon, 1-\epsilon, 2\epsilon, 2-2\epsilon)\\
    &\hspace*{4.75em}
    \max(2\epsilon, 1-\epsilon,1,2\epsilon), 
    \max(2\epsilon, 1-\epsilon,1,2-2\epsilon))\\
    &\overset{+}{=} \min (\min_\epsilon \max(2\epsilon, 1 - \epsilon) \text{ s.t. } 0\leq \epsilon \leq 1,\\
    &\hspace*{3.25em} \min_\epsilon\max(2\epsilon,1-\epsilon,2-2\epsilon) \text{ s.t. } 0\leq \epsilon \leq 1,\\
    &\hspace*{3.25em} \min_\epsilon\max(2\epsilon, 1 - \epsilon, 1,2\epsilon) \text{ s.t. } 0\leq \epsilon \leq 1,\\
    &\hspace*{3.25em} \min_\epsilon\max(2\epsilon, 1 - \epsilon, 1, 2-2\epsilon) \text{ s.t. } 0\leq \epsilon \leq 1)
\end{align*}
The equality (*) holds due to the distributivity of $\max$ over $\min$, while the equality (+) is due to the commutativity of the two $\min$ functions. We then have to take the minimum of the optimal solutions of four optimization problems, which can be encoded as linear programs. The optimal solution is $2/3$ and obtained for $\epsilon=1/3$. By Theorem~\ref{th:main}, the update time of $Q_{\bigpaw}$ is $\bigO(N^{2/3})$. The update time cannot be improved by a polynomial factor, conditioned on the conjectured optimality of the submodular width for static query evaluation~\cite{InsertsVSDeletes}.

\begin{figure}
    \centering
    \setlength{\extrarowheight}{4pt}
    \setlength{\tabcolsep}{10pt} % default is 6pt
    \begin{tabular}{ccc}
        % Figure 1
        \begin{tikzpicture}[
            scale=0.75,
            every node/.style={
            circle,
            draw=none,
            minimum size=2.5em,
            inner sep=0pt,
            font=\scriptsize
            },
            every edge/.style={>=Stealth, line width=0.7, draw=black},
            ]
         
            \node at (3,0) {$R_4(X_1,X_4)$};

            \node at (3,1) {$V_1(X_1)$};
            \node at (0,1) {$R_1(X_1,X_2)$};
               
            \node at (1,2) {$V_2(X_1,X_2)$};
            \node at (3,2) {$R_3(X_3,X_1)$};

            \node at (1,3) {$V_3(X_1,X_2,X_3)$};

            \node at (1,4) {$V_4(X_2,X_3)$};
            \node at (3,4) {$R_2(X_2,X_3)$};

            \node at (1,5) {$V_5(X_2,X_3)$};
            \node at (3,5) {$R_5(X_2,X_5)$};

            \node at (1,6) {$V_6(X_2)$};
            \node at (3,6) {$V_7(X_2)$};

            \node at (1,7) {$V_8(X_2)$};

            \draw (3,0.25) -- (3,0.75);

            \draw (0,1.25) -- (0.85,1.75);
            \draw (3,1.25) -- (1.15,1.75);

            \draw (1,2.25) -- (1,2.75);
            \draw (3,2.25) -- (1.15,2.75);
    
            \draw (1,3.25) -- (1,3.75);

            \draw (1,4.25) -- (1,4.75);
            \draw (3,4.25) -- (1.15,4.75);

            \draw (1,5.25) -- (1,5.75);
            \draw (3,5.25) -- (3,5.75);

            \draw (1,6.25) -- (1,6.75);
            \draw (3,6.25) -- (1.15,6.75);
    
            \node at (1.5,-0.7) {View Tree 1};
        \end{tikzpicture}
        &
        % Figure 2
        \begin{tikzpicture}[
            scale=0.75,
            every node/.style={
            circle,
            draw=none,
            minimum size=2.5em,
            inner sep=0pt,
            font=\scriptsize
            },
            every edge/.style={>=Stealth, line width=0.7, draw=black},
            ]
         
            \node at (3,0) {$R_5(X_2,X_5)$};

            \node at (3,1) {$V_1(X_2)$};
            \node at (0,1) {$R_2(X_2,X_3)$};
               
            \node at (1,2) {$V_2(X_2,X_3)$};
            \node at (3,2) {$R_1(X_1,X_2)$};

            \node at (1,3) {$V_3(X_1,X_2,X_3)$};

            \node at (1,4) {$V_4(X_1,X_3)$};
            \node at (3,4) {$R_3(X_3,X_1)$};

            \node at (1,5) {$V_5(X_1,X_3)$};
            \node at (3,5) {$R_4(X_1,X_5)$};

            \node at (1,6) {$V_6(X_1)$};
            \node at (3,6) {$V_7(X_1)$};

            \node at (1,7) {$V_8(X_1)$};

            \draw (3,0.25) -- (3,0.75);

            \draw (0,1.25) -- (0.85,1.75);
            \draw (3,1.25) -- (1.15,1.75);

            \draw (1,2.25) -- (1,2.75);
            \draw (3,2.25) -- (1.15,2.75);
    
            \draw (1,3.25) -- (1,3.75);

            \draw (1,4.25) -- (1,4.75);
            \draw (3,4.25) -- (1.15,4.75);

            \draw (1,5.25) -- (1,5.75);
            \draw (3,5.25) -- (3,5.75);

            \draw (1,6.25) -- (1,6.75);
            \draw (3,6.25) -- (1.15,6.75);
    
            \node at (1.5,-0.7) {View Tree 2};
        \end{tikzpicture}
        &
        % Figure 3
        \begin{tikzpicture}[
            scale=0.75,
            every node/.style={
            circle,
            draw=none,
            minimum size=2.5em,
            inner sep=0pt,
            font=\scriptsize
            },
            every edge/.style={>=Stealth, line width=0.7, draw=black},
            ]
        
            \node at (0,0) {$R_2(X_2,X_3)$}; 
            \node at (2,0) {$R_3(X_3,X_1)$};

            \node at (1,1) {$V_1(X_1,X_2,X_3)$};

            \node at (-1,2) {$R_4(X_1,X_4)$};
            \node at (1,2) {$V_4(X_1,X_2)$};
            \node at (3,2) {$R_1(X_1,X_2)$};

            \node at (-1,3) {$V_5(X_1)$};
            \node at (1,3) {$V_6(X_1,X_2)$};
    
            \node at (1,4) {$V_7(X_1,X_2)$};
            \node at (3,4) {$R_5(X_2,X_5)$};

            \node at (1,5) {$V_8(X_2)$};
            \node at (3,5) {$V_9(X_2)$};

            \node at (1,6) {$V_{10}(X_2)$};

            \draw (0,0.25) -- (0.85,0.75);
            \draw (2,0.25) -- (1.15,0.75);

            \draw (1,1.25) -- (1,1.75);

            \draw (-1,2.25) -- (-1,2.75);
            \draw (1,2.25) -- (1,2.75);
            \draw (3,2.25) -- (1.15,2.75);

            \draw (1,2.25) -- (1,2.75);

            \draw (-1,3.25) -- (0.85,3.75);
            \draw (1,3.25) -- (1,3.75);

            \draw (1,4.25) -- (1,4.75);
            \draw (3,4.25) -- (3,4.75);

            \draw (1,5.25) -- (1,5.75);
            \draw (3,5.25) -- (1.15,5.75);

            \node at (1,-0.7) {View Tree 3};
        \end{tikzpicture}
    \end{tabular}
    \vspace*{-2em}
    \caption{The three view trees used to maintain the big paw query.}
    \label{fig:big-paw-view-trees}
    % \Description{Three view trees for the big paw query.}
\end{figure}

\begin{table}[t]
\renewcommand{\arraystretch}{1.2}
\centering
\begin{tabular}{|c|c|r||c|c|r|}
\hline
\textbf{Configuration} & \textbf{View} & \textbf{\multirow{2}{*}{$\log_N$ UpdateTime}} & \textbf{Configuration} & \textbf{View} & \textbf{\multirow{2}{*}{$\log_N$ UpdateTime}} \\
$(X_1, X_2, X_3)$ & \textbf{Tree} & & $(X_1, X_2, X_3)$ & \textbf{Tree} & \\
\hline
$L, L, L$ & $1$ & $2\epsilon$ & $H, L, L$ & $2$ & $\min(2\epsilon,1)$  \\
\hline
$L, L, H$ & $1$ & $\min(2\epsilon,1)$ & $H, L, H$ & $2$ & $\min(2\epsilon, 2 - 2\epsilon)$  \\
\hline
$L, H, L$ & $1$ & $\min(2\epsilon,1)$ & $H, H, *$ & $3$ & $1-\epsilon$ \\
\hline
$L, H, H$ & $1$ & $\min(2\epsilon, 2-2\epsilon)$ & & &  \\
\hline
\end{tabular}
\caption{The (base $N$) logarithm of the update time for each degree configuration and a specific view tree from Fig.~\ref{fig:paw-view-trees}. $(*)$ in the degree configuration indicates that the join variable can be either heavy or light.}
\label{tab:comparisons-big-paw}
\end{table}

\section{Proof of Theorem~\ref{th:main}}
\label{sec:main-proof}

In this section we prove:
\ThMain*
    Given a join query \(Q\), we analyze the time complexity of each of the three stages of our algorithm: preprocessing, maintenance, and enumeration.

    \paragraph{Preprocessing.}
    We compute the maintenance width \(\mw(Q)\) and the optimal parameter \(\epsilon^*\) as detailed in Sec.~\ref{subsec:computing-mw}. Next, we partition the active domain into heavy and light values based on the thresholds defined by \(\epsilon^*\). We then compute the set of active view trees and assign the optimal tree \(T^{\vect d}\) to each degree configuration \(\vect d\).

    To compute the views of our active view trees, we start with an empty database and then insert one by one each tuple from the initial database of size $N$. Since the time to process a single update is \(\bigO(N^{\mw(Q)})\), the overall time to compute the views is \(\bigO(N^{1+\mw(Q)})\)). The number of views is only dependent on the query, so independent of the database size. 

    \paragraph{Maintenance.}
    Consider a single tuple update \(\delta R = \{\vect t\mapsto \pm 1\}\). The maintenance procedure follows two steps:
    \begin{enumerate}
        \item \textbf{Selecting the configuration and delta view tree:} We inspect the values in the tuple \(\vect t\). If a value $a$ in $\vect t$ is encountered for the first (i.e., it is not in the active domain), we initialize it as a light value. Otherwise, we use its degree in the data to decide whether it is heavy or light.        
        We then determine the degree configuration \(\vect d\) corresponding to the degrees of values in \(\vect t\) based on the current heavy-light threshold and select the corresponding view tree \(T^{\vect d}\).
        
        \item \textbf{View Updates:} We compute \(\delta T^{\vect d}_R\) for the update \(\delta R\). By Lemma~\ref{lemma:compute-time-delta-view} and the definition of the maintenance width, the cost of this operation is bounded by \(\bigO(N^{\mw(Q)})\) time in data complexity. We the use the computed \(\delta T^{\vect d}_R\) to update the view tree \(T^{\vect d}\). 
    \end{enumerate}
    This bound on update time holds strictly when the degree constraints for \(\vect d\) remain satisfied. However, a sequence of updates may alter value frequencies, violating the constraints. In such cases, a minor or major rebalancing step is triggered. As detailed in Sec.~\ref{sec:rebalancing}, the cost of these rebalancing steps can be amortized over the update sequence, yielding the same (now amortized) update time.

    \paragraph{Enumeration.}
    Upon each enumeration request, we enumerate the distinct tuples in the query output, along with their multiplicities, from the active view trees with constant delay, as described in Sec.~\ref{sec:enumeration}.

\section{Missing Details in Section~\ref{sec:rebalancing}}
\label{app:rebalancing}
In this section, we give a detailed analysis of
the major and minor rebalancing steps explained in Section~\ref{sec:rebalancing}.
In the following, we fix a join query $Q$.

\paragraph{Relaxing the Partition Threshold}
We relax the partition threshold that determines whether a value is classified as light or heavy. The purpose of this relaxation is to prevent sequences of alternating inserts and deletes
from causing a value to change between the heavy and light categories after every single update, thereby triggering a rebalancing step after each update.
Consider a fixed $\epsilon \in [0,1]$.
Given a database of size $N$, let $M \in \N$ be chosen such that $\frac{1}{4} M \leq N < M$. For any join variable $Y$, we partition the $Y$-values in the database into disjoint sets $Light(Y)$ and $Heavy(Y)$ such that (1) for all $y \in Light(Y)$, it holds
$\sum_{R_i(\vect{X_i})\in\at(Y)} |\sigma_{Y=y} R_i| \leq \frac{3}{2} M^\epsilon$ and 
(2) for all $y \in Heavy(Y)$, it holds 
$\sum_{R_i(\vect{X_i})\in\at(Y)} |\sigma_{Y=y} R_i| > \frac{1}{2}M^\epsilon$.
It follows that 
a $Y$-value with $\sum_{R_i(\vect{X_i})\in\at(Y)} |\sigma_{Y=y} R_i| = M^\epsilon$ can be
either in $Light(Y)$ or in $Heavy(Y)$, but it cannot be in both. 
Since $N = \Theta(M)$, all degree constraints stated after Definition~\ref{def:degree_constraint} in Section~\ref{sec:prelims} are satisfied in asymptotic terms:
each atom \(R_i(\vect X_i)\) with join variable $A\in\vect X_i$ guards the degree constraints \((\vect X_i\mid\emptyset, N)\); \((\vect X_i\mid A, \bigO(N^\epsilon))\) in case $A$ is light; and \((A\mid \emptyset, \bigO(N^{1-\epsilon}))\) in case $A$ is heavy. A single-tuple update \(\delta R_i\) implies the degree constraints \((Y\mid \emptyset, 1)\) for all \(Y\in\vect X_i\).

\paragraph{Database States}
%Given a fixed $\varepsilon \in [0,1]$, 
A database state is a triple $S = (M,\vect{\calP}, \vect{\calF}, \vect{\calT})$, where:
$M$ is the threshold base  with 
$\frac{1}{4} M \leq N < M$ with $N$ being the current database size; 
$\vect{\calP}$ consists of the value sets $Heavy(Y)$ and $Light(Y)$
for each join variable $Y$;
$\vect{\calF}$ consists of the database fragments, i.e., it contains for each degree configuration, a corresponding database; and
and $\vect{\calT}$ consists of set of view trees maintained by our approach. In the initial database state, i.e, before processing any update, the threshold base $M$ is set to $2N +1$.   

\paragraph{Major Rebalancing}
If an update causes the database size to drop below 
$\lfloor \frac{1}{4} M \rfloor$, we set 
$M: = \lfloor \frac{1}{2}M\rfloor -1$. 
If, on the other hand, the update causes the database size to reach
$M$, we set $M: = 2M$. 
In either case, we recompute the value partitions in $\vect{\calP}$,  the database fragments in $\vect{\calF}$,
 and the view trees in $\vect{\calT}$, using the partition threshold $M^\varepsilon$.
We refer to this step as major rebalancing.
The time required to compute $\vect{\calP}$ and $\vect{\calF}$ is $\bigO(N)$. The time needed to compute the view trees in $\vect{\calT}$ is given by the preprocessing time  
$\bigO(N^{1+\mw(Q)})$, as stated in Theorem~\ref{th:main}. 
After a major rebalancing step, the database size satisfies
$N = \frac{1}{2}M$ (after doubling), or $N = \frac{1}{2}M -\frac{1}{2}$ or $|N| = \frac{1}{2}M -1$ (after halving). 
To violate the size invariant $\lfloor \frac{1}{4} M\rfloor \leq N < M$ and trigger another major rebalancing step, at least $\frac{1}{4} M = \Omega(M)$ updates are needed. 
By amortizing the $\bigO(N^{1+\mw(Q)}N)$ cost of a major rebalancing step over these preceding $\Omega(M)$ updates and observing that $N = \Theta(M)$, 
we obtain that the amortized cost of major rebalancing per single-tuple update is $\bigO(N^{\mw(Q)})$. 
%Observe that although 
%$\mw(Q)$ is defined in terms of 
%$N$ rather than $M$, the stated upper bound still holds since
%$N = \Theta(M)$.

\paragraph{Minor Rebalancing}
After each single-tuple update $\delta R = \{\vect{x} \mapsto m\}$, we check for each $X$-value $x$ in $\vect{x}$ whether it needs to be moved from 
$Light(X)$ to $Heavy(X)$ or vice-versa. 
Assume that before the update, we have $x \in Light(X)$
and $\sum_{R_i(\vect{X_i})\in\at(X)} |\sigma_{X=x} R_i| =  \lfloor \frac{3}{2} M^{\varepsilon}\rfloor$. Assume that after update, we obtain $\sum_{R_i(\vect{X_i})\in\at(X)} |\sigma_{X=x} R_i| =  \lfloor \frac{3}{2} M^{\varepsilon}\rfloor + 1$. In this case, we move 
$x$ from $Light(X)$ to $Heavy(X)$ and move all tuples 
that contain the value $x$ and are in a database fragment corresponding to a degree configuration where $X$ is light to the fragment where $X$ is heavy.
Additionally, we update the view trees in $\vect{\calT}$ evaluated over the two database fragments that have been changed. 
%We delete the tuples in $\vect{S}$ from the view tree $T_c$ and insert them to the view tree $T_{c'}$. 
If the $X$-value moves from $Heavy(X)$ to $Light(X)$, the computation is analogous. 
We refer to this step as minor rebalancing.
If the value $x$ moves from $Heavy(X)$ to $Light(X)$, at most $\frac{1}{2}M^\varepsilon$ tuples need to be moved between the fragments. If the $x$ moves from $Light(X)$ to $Heavy(X)$, at most $\frac{3}{2}M^\varepsilon + 1$ need to be moved between the fragments. Using the update mechanism of our approach to insert and delete tuples from database fragments, we observe that moving a tuple from one fragment to another takes 
$\bigO(N^{\mw(Q)})$ time, as stated in Theorem~\ref{th:main}. Hence the overall time to do minor rebalancing is $\bigO(N^{\mw(Q)} M^\varepsilon)$. We amortize this minor rebalancing time over over $\Omega(M^\varepsilon)$ updates required between two minor rebalancing steps. This implies that the amortized minor rebalancing time per single-tuple update is $\bigO(N^{\mw(Q)})$.

\nop{
\smallskip
The above analysis implies:
\begin{proposition}
\label{prop:amortized}
    For any join query $Q$ and sequence of single-tuple updates, the amortized major and minor rebalancing time is \(\bigO(2^{\mw(Q)})\) per single-tuple update, where $\mw(Q)$ is the maintenance width of $Q$.
\end{proposition}
}
\section{Missing Details in Section~\ref{sec:enumeration}}
\label{app:enumeration}
In this section, we prove:
\begin{proposition}
    \label{prop:const_delay_enum}    
    For any join query $Q$ and 
    %\eden{is $\T$ a single view tree which is a set of trees?} 
    view tree set $\vect{\calT}$ constructed by our approach for $Q$, it holds that the output of $Q$ can be enumerated from $\vect{\calT}$ with constant delay.
\end{proposition}

Before giving the proof of Proposition~\ref{prop:const_delay_enum}, we illustrate our enumeration strategy for the 4-cycle query:

\begin{example}
\label{ex:enumeration_4-cycle}
Consider the view trees in Fig.~\ref{fig:4-cycle-view-trees} used to maintain the 4-cycle query. 
We illustrate how the tuples in the join of the views in 
View Trees 1 and 3 can be enumerated from these view trees with constant delay. Our enumeration strategy works for any degree configuration. The enumeration strategy for the other view trees is analogous. 

In View Tree 1, we use the view $V_5$ to retrieve distinct $(A,C)$-values and the views $V_1$ and $V_2$ to retrieve distinct $B$- and respectively $D$-values. 
%We visit the views in the order $V_5$, $V_1$, $V_2$. 
To construct the first output tuple, we retrieve one $(A,C)$-value $(a,c)$ from $V_5$, one $B$-value $b$ from $V_1(a,B,c)$, and one $D$-value $d$ from 
$V_2(c,D,a)$. Afterwards, we report the tuple $(a,b,c,d)$. Then, we iterate over the remaining $D$-values in 
$V_2(c,D,a)$ and report for each such value $d'$, 
the tuple $(a,b,c,d')$. After all values in $V_2(c,D,a)$
are exhausted, we retrieve the next $B$-value $b'$ in $V_1(a,B, c)$ and then iterate again over all $D$-values in $V_2(c,D,a)$. For each such value $d$, we report the tuple
$(a,c,b',d)$.
After the view $V_1$ is exhausted, we retrieve the next 
$(A,C)$-value $(a',c')$ in $V_5$ and repeat the enumeration process in the context of $(a',c')$. We are done, when the enumeration is completed in the context of the last $(A,C)$-value in $V_5$.

In View Tree 3, we use the view $V_5$ to retrieve distinct $(A,D)$-values, and the views $V_3$ and $V_1$ to retrieve distinct $C$- and respectively $B$-values. 
%We visit the views in the order $V_5$, $V_3$, $V_1$. 
The enumeration is analogous to the case of View Tree 1.  
\end{example}

We denote by $\vars(\calT)$ the set of variables in a view tree $\calT$.
Given a view tree $\calT$ and a view $V(\vect{X})$ in $\calT$, we say that the view $V$ {\em owns} a variable $X \in \vars(\calT)$ if $X \in \vect{X}$ and each view or atom $V'(\vect{X}')$ with $X \in \vect{X}'$ appears in the subtree rooted at $V$. It follows from the construction of view trees that
each variable is owned by a unique view:

\begin{restatable}{proposition}{PropUniqueOwner}
\label{prop:unique_owner}
For any view tree $\calT$ and variable $X \in \vars(\calT)$, it holds that $\calT$ contains a unique view that owns $X$.
\end{restatable}

\begin{proof}
We say that two views in a view tree are independent 
if they do not appear  on a root-to-leaf path in the view tree.

Consider a view tree $\calT$ and a variable $X \in \vars(\calT)$.
For the sake of contradiction, assume that $\calT$ does not contain a unique view owning $X$. This means that $\calT$ has two independent views $V_1(\vect{X_1})$ and $V_2(\vect{X_2})$ such that (i) $X \in \vect{X}_1$, (ii) $X \in \vect{X}_2$, 
(iii) for any view $V_1'(\vect{X}_1')$ above $V_1(\vect{X}_1)$, it holds $X \notin \vect{X}_1'$,
and (iv) for any view $V_2'(\vect{X}_2')$ above $V_2(\vect{X}_2)$, it holds $X \notin \vect{X}_2'$.
This implies that the parent view $\hat{V}_1(\vect{\hat{X}_1})$ of $V_1(\vect{X}_1)$ is a projection
view that projects away $X$, which means $X \notin \vect{\hat{X}}_1$.
The definition of view trees requires that $V_2$ must be in the subtree rooted at $V_1$ (Definition~\ref{def:viewtree}). This means that $V_1$ and $V_2$ cannot be independent, which is a contradiction. 
\end{proof}

\begin{figure}[t]
\center
\setlength{\tabcolsep}{3pt}
	\renewcommand{\arraystretch}{1.2}
	\renewcommand{\linenumber}{\makebox[2ex][r]{\rownumber\TAB}}
	\setcounter{magicrownumbers}{0}
%	\begin{tabular}[t]{@{}l@{}}
\begin{tabular}[t]{@{\hskip 0.1in}l}
        \toprule
 \textsc{Enumerate} (view tree $\calT$)\\
\midrule
\linenumber \LET $\vect \calV$ be the set of views in $\calT$ that own at least one variable in $\calT$ \\
\linenumber \LET $V_{1}(\vect{X}_1), \ldots , V_{n}(\vect{X}_n)$
be an ordering of the views in $\vect \calV$ that is consistent with $\calT$\\
\linenumber \LET $\vect{O}_i$ be the variables owned by $V_i$ and \LET $\vect{N}_i = \vect{X}_i \setminus \vect{O}_i$, for $i \in [n]$ \\ 
\linenumber \FOREACH $\vect{t}_1 \in V_1(\vect{O}_1)$ \\ 
\linenumber \TAB\LET $\vect{t}_1' = \pi_{\vect{N}_2} \vect{t}_1$ \\
\linenumber \TAB \FOREACH $\vect{t}_2 \in V_2(\vect{t}_1',\vect{O}_2)$ \\ 
\linenumber \TAB\TAB \LET $\vect{t}_2' = \pi_{\vect{N}_3} (\vect{t}_1 \cdot \vect{t}_2)$ \\
\linenumber \TAB\TAB \FOREACH $\vect{t}_3 \in V_2(\vect{t}_2',\vect{O}_3)$ \\ 
\linenumber \TAB\TAB \TAB  ${\cdot}{\cdot}{\cdot}$ \\ 
\linenumber \TAB \TAB \TAB \TAB \LET $\vect{t}_{n-1}' = \pi_{\vect{N}_n} (\vect{t}_1 \cdots \vect{t}_{n-1})$ \\ 
\linenumber \TAB \TAB \TAB \TAB \FOREACH $\vect{t}_n \in V_n(\vect{t}_{n-1}',\vect{O}_n)$ \\ 
\linenumber \TAB\TAB\TAB\TAB \TAB  \textbf{report} $\vect{t}_1{\cdot}{\cdot}{\cdot} \vect{t}_n$ \\
\bottomrule
\end{tabular}
%%%%%%%%%%%%%%%%%%%%%%%%%%%%%%%%
\caption{Constant-delay enumeration of the tuples in the join of the relations at the leaves of a view tree.}
% \Description{Algorithm for the constant-delay enumeration of the tuples in the join of the relations at the leaves of a view tree $\calT$.}
\label{fig:enumeration}
\end{figure}

%\medskip
Equipped with Proposition~\ref{prop:unique_owner}, we are ready to prove Proposition~\ref{prop:const_delay_enum}.

\begin{proof}[Proof of Proposition~\ref{prop:const_delay_enum}]
As explained at the beginning of Section~\ref{sec:enumeration}, it suffices 
to show that for any view tree $\calT$, the set of tuples in the join of the views of  $\calT$
can be enumerated with constant delay. The procedure \textsc{Enumerate} in Fig.~\ref{fig:enumeration}
describes our enumeration strategy for any given view tree. First, we explain the details of the procedure. Then, we show its correctness, i.e., we explain why it enumerates all distinct tuples represented by the view tree. Finally, we show that it enumerates with constant delay. 

\paragraph{Enumeration Procedure}
Let $\vect \calV$ be the set of views in the input view tree $\calT$ that own at least one variable in $\calT$.
For each $i \in [n]$, let $\vect{O}_i$ be the variables owned by $V_i$ and let $\vect{N}_i = \vect{X}_i \setminus \vect{O}_i$.
The procedure \textsc{Enumerate} creates a strict ordering $V_1(\vect{X_1}), \ldots, V_n(\vect{X_n})$ of the views in $\vect \calV$ that is consistent with the partial ordering given by $\calT$ (Line~2), i.e., for any $i, j \in [n]$ it holds: if $V_j$ is on the path from $V_i$ to the root of $\calT$, then $i <j$.
The procedure constructs the first value tuple over $\vars(\calT)$
by traversing the views in the strict order as follows (Lines~4--11). 
All variables in $\vect{X_1}$  
must be owned by $V_1$, hence, $\vect{X}_1 = \vect{O}_1$.
The procedure retrieves a tuple 
$\vect{t}_1$ over $\vect{O}_1$ from $V_1$. For any $i \in \{2, \ldots , n\}$ it proceeds as follows. 
Let $\vect{t}_{i-1}$ be the tuple constructed before visiting the view $V_{i}$ and let $\vect{t}_{i-1}' = \pi_{\vect{N}_i} \vect{t}_{i-1}$. 
The procedure retrieves a tuple $\vect{t}_i$ from $V_i(\vect{t}_{i-1}',\vect{O}_i)$.
After all views in $\vect \calV$ are visited, we have a complete tuple 
$\vect{t} = \vect{t}_1 \cdots \vect{t}_n$, which is reported by the procedure. Then, the procedure iterates over the remaining tuples in 
$V_n(\vect{t}_{n-1}',\vect{O}_n)$ and outputs for each such tuple 
$\vect{\hat{t}}_n$, the tuple 
$\vect{t} = \vect{t}_1 \cdots \vect{t}_{n-1} \cdot \vect{\hat{t}}_n$. 
After all tuples in $V_n$ are exhausted, the procedure backtracks, i.e., it retrieves the next tuple in $V_{n_1}$ and iterates again over the tuples in $V_n$. The procedure stops after each view in 
$\vect \calV$ is exhausted. 

\paragraph{Correctness}
The correctness of the enumeration procedure follows from two observations. 
Firstly, the values of any variable $X$ are retrieved from the view that owns $X$, which by Proposition~\ref{prop:unique_owner} is unique.
This means that the values are retrieved from a view that joins all views containing $X$. 

Secondly, it follows from Proposition~\ref{prop:unique_owner} that for each view $V_i(\vect{X}_i)$, it holds: all variables in $\vect{N}_i$
have their owning views above $V_i$ in the view tree. This means that at the time, the procedure visits the view $V_i$, all variables in $\vect{N}_i$ are already fixed to some constant, which guarantees that all tuples in $V_i(\vect{t}_{i-1}', \vect{O}_i)$ are distinct. 

\paragraph{Time analysis}
The view tree contains constantly many views (in data complexity). For any view $V_i$ and any tuple $\vect{t}$ over 
$\vect{O}_i$, our computation model allows for the constant-delay enumeration of the distinct tuples
in 
$V_i(\vect{t}_{i-1}', \vect{O}_i)$. This implies that 
the procedure $\textsc{Enumerate}$ constructs each output tuple in constant time.
\end{proof}

We conclude this section by showing how our approach can be easily adapted to maintain the query count, i.e., the number of tuples in the query output, with the same  update time and constant-delay enumeration as for full queries. 

\paragraph{Maintaining the Query Count}
%We explain how our approach can be extended so that we can use the view trees to compute the number of tuples in the query output in constant time. 
Given a join query $Q$, consider  the view tree $\calT_{\vect d}$ with root view $V_{\vect d}(\vect{X}_{\vect{d}})$ constructed for each degree configuration 
$\vect d \in \vect D(Q)$. 
We extend each such a view tree with a projection view $\hat{V}_{\vect d}()$ that marginalizes our all variables of $V_{\vect d}$. For any single-tuple update, the  maintenance time for $\hat{V}_{\vect d}$ is upper-bounded by the maintenance time for $V_{\vect d}$.  
Note that $\hat{V}_{\vect d}$ is a constant function that returns the number of tuples in the join of the views in 
$\calT_{\vect d}$.
Hence, the number of output tuples of the query  is 
$\sum_{\vect{d} \in \vect{D}(Q)}\hat{V}_{\vect d}()$, which can be computed in constant time.

% \received{December 2025}
% \received[accepted]{February 2026}

%File name: V4mod109

\end{document}